\documentclass[10pt,twoside]{article} 
\usepackage{latexsym}
\usepackage{lscape}
\usepackage{epstopdf,caption,subcaption,graphicx,xcolor}
\usepackage{tikz}
\usepackage{pgfgantt}
\usepackage{hyperref}

\usetikzlibrary{shapes,arrows,fit,calc,positioning}
\usetikzlibrary{decorations.pathreplacing}

\normalsize
\newtheorem{theorem}{Theorem}
\newtheorem{lemma}{Lemma}

\newtheorem{fact}{Fact}

\newtheorem{remark}{Remark}
\newtheorem{notation}{Notation}
\newtheorem{invariant}{Invariant}
\newtheorem{operation}{Operation}
\newtheorem{definition}{Definition}

\newenvironment{proof}{{\sc Proof. }}{\hfill$\Box$\vspace{0.1in}}
\def\mcF{\mathcal{F}}
\def\mcG{\mathcal{G}}
\def\mcK{\mathcal{K}}
\def\mcM{\mathcal{M}}
\def\mcP{\mathcal{P}}
\title{Approximately covering vertices by order-$5$ or longer paths\thanks{An extended abstract appears in {\em Proceedings of COCOON 2024}.}}
\author{
Mingyang~Gong\thanks{Department of Computing Science, University of Alberta.  Edmonton, Alberta T6G 2E8, Canada.
	Emails: \texttt{\{mgong4, guohui\}@ualberta.ca}}
\and
Zhi-Zhong~Chen\thanks{Division of Information System Design, Tokyo Denki University. Saitama 350-0394, Japan.
  Email: \texttt{zzchen@mail.dendai.ac.jp}}
	\thanks{Correspondence authors.}
\and
Guohui~Lin$^*$$^\S$
\and
Lusheng~Wang\thanks{Department of Computer Science, City University of Hong Kong. Hong Kong SAR, China.
  Email: \texttt{cswangl@city.edu.hk}}
}
\date{}
\begin{document}
\maketitle
\begin{abstract}
This paper studies $MPC^{5+}_v$, which is to cover as many vertices as possible 
in a given graph $G=(V,E)$ by vertex-disjoint $5^+$-paths (i.e., paths each with at least five vertices). 
$MPC^{5+}_v$ is NP-hard and admits an existing local-search-based approximation algorithm which 
achieves a ratio of $\frac {19}7\approx 2.714$ and runs in $O(|V|^6)$ time. 
In this paper, we present a new approximation algorithm for $MPC^{5+}_v$ which achieves a ratio of $2.511$ and runs in $O(|V|^{2.5} |E|^2)$ time.
Unlike the previous algorithm, the new algorithm is based on maximum matching, maximum path-cycle cover, and recursion.

\end{abstract}

\section*{Acknowledgments}
This research is supported by the NSERC Canada,
the Grant-in-Aid for Scientific Research of the Ministry of Education, Science, Sports and Culture of Japan (Grant No. 18K11183),
the National Science Foundation of China (Grant No. 61972329),
the GRF grants for Hong Kong Special Administrative Region, China (CityU 11206120, CityU 11210119, CityU 11218821),
a grant from City University of Hong Kong (CityU 11214522),
and the Ministry of Science and Technology of China (G2022040004L, G2023016016L).

\newpage 

\section{Introduction}\label{sec:intro}
Throughout this paper, a graph always means an undirected graph without parallel edges or self-loops,
and an approximation algorithm always means one running in time polynomial in the input size.

Let $k$ be a positive integer, and $G$ be a graph. 
A {\em $k$-path} (respectively, {\em $k^+$-path}) in $G$ is a path in $G$ having exactly (respectively, at least) $k$ vertices. 
$MPC^{k+}_v$ is the problem of finding a collection of vertex-disjoint $k^+$-paths so that the total number of vertices in these paths is maximized. 
Clearly, $MPC^{1+}_v$ is trivially solvable by choosing $|V|$ vertices of $G$,
while $MPC^{2+}_v$ is equivalent to finding a path cover of $G$ that minimizes the number of $1$-paths and hence can be solved in $O(|V||E|)$ time~\cite{CGL19a}.
$MPC^{3+}_v$ has been extensively studied before, too (see~\cite{HHS06} and the references therein).
Indeed, it is equivalent to the $2$-piece packing problem~\cite{HHS06} which aims to cover the maximum number of vertices by vertex-disjoint $2$-pieces,  
where a $2$-piece of $G$ is a connected subgraph of $G$ such that the degree of each vertex is at most $2$ and 
at least one of them has degree exactly $2$.
Since the 2-piece packing problem can be solved in polynomial time, so is $MPC^{3+}_v$.
Unfortunately, $MPC^{k+}_v$ is NP-hard for every $k\ge 4$ \cite{KLM22,KLM23}.
For other problems related to $MPC^{k+}_v$, the reader is referred to \cite{BK06,AN07,PH08,AN10,RTM14,CCC18,GW20,CCL22} for more details.

The NP-hardness of $MPC^{k+}_v$ for $k \ge 4$ has motivated researchers to design approximation algorithms for it. 
As observed in \cite{GFL22,GEF24,GCL23}, there is a trivial reduction from $MPC^{k+}_v$ to the maximum weighted $(2k-1)$-set packing problem, 
and this reduction together with the best-known approximation algorithm for the latter problem yields an approximation 
algorithm for $MPC^{k+}_v$ achieving a ratio of $k$. 
Gong {\em et al.} \cite{GFL22,GEF24} has improved the ratio to $0.4394k + 0.6576$; their algorithm runs in $O(|V|^{k+1})$ time.

Since $MPC^{4+}_v$ is the simplest NP-hard case among $MPC^{k+}_v$ for various values of $k$, 
several approximation algorithms for $MPC^{4+}_v$ have been designed \cite{KLM22,KLM23,GFL22,GEF24,GCL23}. 
Kobayashi {\em et al.} \cite{KLM22,KLM23} design an approximation algorithm for $MPC^{4+}_v$ achieving a ratio of~$4$.
Afterwards, Gong {\em et al.}  \cite{GFL22,GEF24} present an approximation algorithm for $MPC^{4+}_v$ achieving a ratio of~$2$; 
their algorithm runs in $O(|V|^8)$ time and is based on time-consuming local search.  
As an open question, Gong {\em et al.} \cite{GFL22,GEF24} ask whether one can design better approximation algorithms for $MPC^{4+}_v$ by completely different approaches. 
Their question has been answered in the affirmative recently in \cite{GCL23}. 
In more details, Gong {\em et al.} \cite{GCL23} design an approximation algorithm for $MPC^{4+}_v$ which achieves a ratio of~$1.874$ and 
runs in $O(\min\{|V|^2 |E|^2, |V|^5\})$ time. 
Unlike the previously known algorithms,
their algorithm starts with a maximum matching $M$ of the input graph $G$ and then tries to connect a large portion of the edges in $M$ into a feasible solution.
If the try fails, then their algorithm reduces the problem to a smaller instance which is solved by a recursive call.

Actually, if we fix $k=5$, Gong {\em et al.}~\cite{GFL22,GEF24}'s approximation algorithm for $MPC^{k+}_v$ achieves a ratio of $\frac{19}{7}\approx 2.714$ and runs in $O(|V|^6)$ time. 
Other than this implied one for $MPC^{5+}_v$, no specific approximation algorithm for $MPC^{5+}_v$ has been previously designed. 
As an open question, Gong {\em et al.}~\cite{GFL22,GEF24} ask
whether the local search ideas in their $2$-approximation algorithm for $MPC^{4+}_v$ can be extended to $MPC^{k+}_v$ for $k\ge 5$.
Unfortunately, to the best of our efforts, it seems very difficult to extend the design and analysis.

In this paper, we focus on $MPC^{5+}_v$ and show that the ideas in the $1.874$-approximation algorithm for $MPC^{4+}_v$~\cite{GCL23}
can be nontrivially extended to obtain a $2.511$-approximation algorithm for $MPC^{5+}_v$.
To see the main differences between the two algorithms, we here sketch the former algorithm. 
Roughly speaking, the $1.874$-approximation algorithm has four stages. 
In the first stage, it computes a maximum matching $M$ in the input graph $G$. 
The intuition behind this idea is that the paths in an optimal solution for $G$ can cover at most $\frac 52 |M|$ vertices. 
So, it suffices to find a solution for $G$ of which the paths cover a large fraction of the endpoints of the edges in $M$. 
Thus, in the second stage, the algorithm uses certain edges of $G$ to connect the edges of $M$ into connected components 
in which the longest paths are $5$-paths.
In the third stage, it uses certain edges of $G$ to connect as many {\em bad} components (each of which contains no $4^+$-path) to other components as possible. 
In the last stage, it tries to use only the edges in the finally-obtained components to compute a set $\mcP$ of vertex-disjoint $4^+$-paths. 
If the total number of vertices in the paths in $\mcP$ is large enough, the algorithm outputs $\mcP$ as the solution; 
otherwise, it makes a recursive call on a smaller graph and uses the returned solution to construct a solution for the original graph $G$.

Our new algorithm for $MPC^{5+}_v$ has four similar stages but only the first stage is the same while the other three are significantly different. 
In particular, in the second stage,
the new algorithm tries to connect the edges of $M$ into as many $5$- or $4$-paths as possible by finding {\em augmenting triples and pairs} (cf. Definitions~\ref{def01}--\ref{def03}).
In order to find an augmenting pair, we may need to modify the existing $5$-paths and count how many $4$-paths can be formed in the mean time.
Such a process does not appear in the algorithm for $MPC^{4+}_v$ at all.
In the third stage, it has to deal with new types of bad components, which are now defined as those components containing $4$-paths but no $5$-paths.
In the algorithm for $MPC^{4+}_v$, a bad component is one contains only one edge of $M$,
while in the new algorithm a bad component may contain two edges of $M$ and hence needs to be handled more carefully in order.
Indeed, because of the new types of bad components, the analysis of our new algorithm appears much more complex.

The remainder of the paper is organized as follows.
Section~2 gives basic definitions.
Section~3 presents the algorithm for $MPC^{5+}_v$ and its performance analysis.
Section~4 concludes the paper.

\section{Basic Definitions}\label{sec:def}
In the remainder of this paper, we fix an instance $G$ of $MPC^{5+}_v$ for discussion.
Let $V(G)$ and $E(G)$ denote the vertex and the edge sets of $G$, respectively, and further let $n = |V(G)|$ and $m = |E(G)|$.

For a subset $F$ of $E(G)$, we use $V(F)$ to denote the endpoints of edges in $F$.
A {\em spanning subgraph} of $G$ is a subgraph $H$ of $G$ with $V(H)=V(G)$. 
For a set $F$ of edges in $G$, $G - F$ denotes the spanning subgraph $(V(G),E(G)\setminus F)$.  
In contrast, for a set $F$ of edges with $V(F)\subseteq V(G)$ and $F\cap E(G) = \emptyset$, 
$G + F$ denotes the graph $(V(G),E(G)\cup F)$.
The {\em degree} of a vertex $v$ in $G$, denoted by $d_G(v)$, is the number of edges incident to $v$ in $G$. 
A vertex $v$ of $G$ is {\em isolated} in $G$ if $d_G(v) = 0$. 
The {\em subgraph induced by} a subset $U$ of $V(G)$, denoted by $G[U]$, is the graph $(U, E_U)$, where $E_U=\{\{u,v\}\in E(G) \mid u, v\in U\}$.

A {\em cycle} in $G$ is a connected subgraph of $G$ in which each vertex is of degree~2. 
A {\em path} in $G$ is either a single vertex of $G$ or a connected subgraph of $G$ 
in which exactly two vertices (called the {\em endpoints}) are of degree~$1$ and 
the others (called the {\em internal vertices}) are of degree~$2$. 
A {\em path component} of $G$ is a connected component of $G$ that is a path. 
If a path component is an edge, then it is called an {\em edge component}. 
The {\em order} of a cycle or path $C$, denoted by $|C|$, is the number of vertices in $C$. 
A {\em triangle} of $G$ is a cycle of order~$3$ in $G$. 
A {\em $k$-path} of $G$ is a path of order $k$ in $G$, while a {\em $k^+$-path} of $G$ is a path of order $k$ or more in $G$. 
A {\em matching} of $G$ is a (possibly empty) set of edges of $G$ in which no two edges share an endpoint. 
A {\em maximum matching} of $G$ is a matching of $G$ whose size is maximized over all matchings of $G$. 
A {\em path-cycle cover} of $G$ is a set $F$ of edges in $G$ such that 
in the spanning subgraph $(V(G), F)$, the degree of each vertex is at most~$2$.
A {\em star} (respectively, {\em bi-star}) is a connected graph in which exactly one  vertex is 
(respectively, two vertices are) of degree $\ge 2$ and each of the remaining vertices is of degree~$1$.
The vertices of degree~1 are the {\em satellites} of the star or bi-star, while each other vertex is 
a {\em center} of the star or bi-star.
Clearly, a $3$-path is a star and a $4$-path is a bi-star.

\begin{notation}
\label{nota01}
For a graph $G$, 
\begin{itemize}
\parskip=0pt
\item
	$OPT(G)$ denotes an optimal solution for the instance graph $G$ of $MPC^{5+}_v$,
	and $opt(G)$ denotes the total number of vertices in $OPT(G)$;
\item
	$ALG(G)$ denotes the solution for $G$ outputted by a specific algorithm,
	and $alg(G)$ denotes the total number of vertices in $ALG(G)$.
\end{itemize}
\end{notation}

\section{The Algorithm}\label{sec:1stAlgo}
Our algorithm for $MPC^{5+}_v$ consists of multiple phases.
In the first phase, it computes a maximum matching $M$ in $G$ in $O(\sqrt{n} m)$ time \cite{MV80}, initializes a subgraph $H = (V(M), M)$, 
and then repeatedly modifies $H$ and $M$ as described in Section~\ref{subsec:k=5}.
The next lemma shows that $|V(M)|$ is relatively large compared to $opt(G)$.

\begin{lemma}
\label{lemma01}
$|V(M)| \ge \frac 45 opt(G)$. 
\end{lemma}
\begin{proof}
Consider an arbitrary $(\ell+1)$-path $P$ in $OPT(G)$.
Let $e_1$, \ldots, $e_\ell$ be the edges of $P$ and suppose that they appear in $P$ in this order from one endpoint to the other.
Obviously, $P_o = \{e_i \mid i \mbox{ is odd}\}$ is a matching. 
If $\ell$ is odd, $V(P) = V(P_o)$; otherwise, exactly one vertex of $P$ is not in $V(P_o)$. 
So, $|V(P_o)| \ge \frac \ell{\ell+1} |V(P)| $ always holds.
Since $\ell + 1 \ge 5$ and thus we have $|V(P_o)| \ge \frac 45 |V(P)| $.
Note that $\cup_{P \in OPT(G)} P_o$ is a matching and $opt(G) = \sum_{P \in OPT(G)} |V(P)|$.
So, $|V(M)| \ge |\cup_{P \in OPT(G)} V(P_o)| \ge \frac 45 opt(G)$. 
\end{proof}

\subsection{Modifying $H$ and $M$}\label{subsec:k=5}
We here describe a process for modifying $H$ and $M$ iteratively. 
The process consists of several steps, during which the following will be an invariant.

\begin{invariant}
\label{inva01}
$M$ is both a maximum matching of $G$ and a subset of $E(H)$. 
Each connected component $K$ of $H$ is an edge, a triangle, a star, a bi-star, or a $5$-path.
Moreover, if $K$ is a $5$-path, then the two edges of $E(K)$ incident to the endpoints of $K$ are in $M$;
if $K$ is a bi-star, then $K$ contains exactly two edges of $M$ and each of them connects a center and one of its satellites;
otherwise, exactly one edge of $K$ is in $M$. 
\end{invariant}

Initially, Invariant~\ref{inva01} clearly holds.
In the sequel, we use $p_5(H)$ to denote the number of $5$-path components in $H$. 
The next definition of {\em augmenting triple} is for modifying $H$ and $M$ to increase $p_5(H)$.

\begin{definition}
\label{def01}
An {\em augmenting triple} with respect to $H$ is a triple $(u_0, e_0 = \{v_0, w_0\}, e_1 = \{v_1, w_1\})$ such that
$u_0 \in V(G) \setminus V(H)$, both $e_0$ and $e_1$ are edge components of $H$, 
and they can be merged into a $5$-path. 
{\em Modifying $H$ and $M$ with the triple} merges the triple into a $5$-path and modifies $M$ to contain the two end-edges of the $5$-path.
\end{definition}

Note that if $G$ has an edge connecting two edge components of $H$, then one could connect the two components into a $4$-path component.
The quantity $q_4(H)$ counts the maximum number of such $4$-paths that can be obtained simultaneously for the current $H$.

\begin{definition}
\label{def02}
We define $q_4(H)$ to be $|\mcM|$, where $\mcM$ is a maximum matching in an auxiliary graph $\mcG$
whose vertices one-to-one correspond to the edge components of $H$ and
whose edge set consists of all $\{N_1, N_2\}$ such that $G$ has an edge between the edge components corresponding to $N_1$ and $N_2$.
\end{definition}

The next definition of {\em augmenting pair} is for modifying $H$ and $M$ so that either a new augmenting triple appears, or $q_4(H)$ increases.

\begin{definition}
\label{def03}
An {\em augmenting pair} with respect to $H$ is a pair of a $5$-path component $K=v_1$-$v_2$-$v_3$-$v_4$-$v_5$ and an edge component $e=\{u,w\}$ of $H$
such that $G$ has an edge $\{x,y\}$ with $x\in\{u,w\}$ and $y\in\{v_1,v_3,v_5\}$ and the following S1 and S2 hold:
\begin{itemize}
\parskip=0pt
\item[S1.] For some $i\in\{1,4\}$, $(v_3, \{u,w\}, \{v_i, v_{i+1}\})$ becomes an augmenting triple with respect to $H'$, where $H'$ is the graph obtained from $H$ by removing $v_3$.
\item[S2.] If we modify $H'$ and $M$ with the triple $(v_3, \{u,w\}, \{v_i, v_{i+1}\})$, then $q_4(H') > q_4(H)$ or a new augmenting triple with respect to $H'$ appears.
\end{itemize}
{\em Modifying $H$ and $M$ with an augmenting pair} is done by first removing $v_3$, then 
modifying $H$ and $M$ with the triple $(v_3, \{u,w\}, \{v_i, v_{i+1}\})$, and further 
modifying $H$ and $M$ with a new augmenting triple specified in S2 if it exists.
\end{definition}

One sees that modifying $H$ and $M$ with an augmenting triple increases $p_5(H)$ but may decrease $q_4(H)$. 
Fortunately, $q_4(H)$ can decrease by at most~$2$ because the modification involves only two edge components of $H$.
Moreover, modifying $H$ and $M$ with an augmenting pair either increases $p_5(H)$, or does not change $p_5(H)$ but then increases $q_4(H)$.
In the former case, $q_4(H)$ may decrease, but can decrease by at most~$3$ because the modification involves at most $3$ edge components of $H$ other than the edges of $K$.
It follows that the total number of modifications with augmenting triples and augmenting pairs is at most $O(n)$,
which are done in the following Step~1.1 for modifying $H$ and $M$.
\begin{description}
\parskip=0pt
\item[Step 1.1] 
	Repeatedly modify $H$ and $M$ with an augmenting triple or an augmenting pair until neither of them exists.
	({\em Comment:} During this step, each connected component of $H$ is a $5$-path or an edge.
	If both an augmenting triple and an augmenting pair exist, then we choose the augmenting triple to modify $H$ and $M$.)
\end{description}

An illustration of a single repetition in Step~1.1 is shown in Figure~\ref{fig01}.

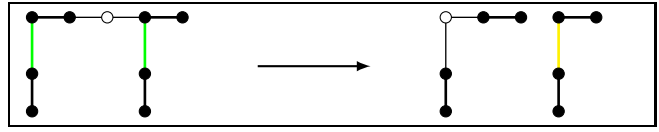
\begin{figure}[thb]
\begin{center}
\framebox{
\begin{minipage}{0.5\textwidth}
\begin{tikzpicture}[scale=0.5,transform shape]

\draw [thick, line width = 1pt] (-8, 2) -- (-7, 2);
\draw [thick, line width = 1pt] (-5, 2) -- (-4, 2);
\draw [thin, line width = 0.5pt] (-7, 2) -- (-6, 2);
\draw [thin, line width = 0.5pt] (-6, 2) -- (-5, 2);
\draw [green, thick, line width = 1pt] (-8, 2) -- (-8,0.5);
\draw [thick, line width = 1pt] (-8, 0.5) -- (-8, -0.5);
\draw [green, thick, line width = 1pt] (-5, 2) -- (-5, 0.5);
\draw [thick, line width = 1pt] (-5, 0.5) -- (-5, -0.5);
\filldraw (-8, 2) circle(.15);
\filldraw (-7, 2) circle(.15);
\filldraw[fill = white] (-6, 2) circle(.15);
\filldraw (-5, 2) circle(.15);
\filldraw (-4, 2) circle(.15);
\filldraw (-8,0.5) circle(.15);
\filldraw (-8,-0.5) circle(.15);
\filldraw (-5,0.5) circle(.15);
\filldraw (-5,-0.5) circle(.15);

\draw [-latex, thick] (-2, 0.7) to (1, 0.7);

\draw [thin, line width = 0.5pt] (3, 2) -- (4, 2);
\draw [thick, line width = 1pt] (6, 2) -- (7, 2);
\draw [thick, line width = 1pt] (4, 2) -- (5, 2);
\draw [thin, line width = 0.5pt] (3, 2) -- (3, 0.5);
\draw [thick, line width = 1pt] (3, 0.5) -- (3, -0.5);
\draw [yellow, thick, line width = 1pt] (6, 2) -- (6, 0.5);
\draw [thick, line width = 1pt] (6, 0.5) -- (6, -0.5);
\filldraw[fill = white] (3, 2) circle(.15);
\filldraw (4, 2) circle(.15);
\filldraw (5, 2) circle(.15);
\filldraw (6, 2) circle(.15);
\filldraw (7, 2) circle(.15);
\filldraw (3, 0.5) circle(.15);
\filldraw (3, -0.5) circle(.15);
\filldraw (6, 0.5) circle(.15);
\filldraw (6, -0.5) circle(.15);

\end{tikzpicture}
\end{minipage}}
\end{center}
\captionsetup{width=1.0\linewidth}
\caption{A typical case for a single repetition of Step~1.1 that involves four $2$-paths.
	The filled (respectively, blank) vertices are in (respectively, not in) $V(M)$.
	The thick (respectively, thin) edges are in the matching $M$ (respectively, in $H$ but not in $M$) and
	the yellow (respectively, green) edges are in $\mcM$ (respectively, not in $E(H)$).
	In the left-hand-side graph, we have $q_4(H) = 0$ and the left green edge and the $5$-path form an augmenting pair;
	modifying with this augmenting pair results in the right-hand-side graph for which $q_4(H) = 1$.
	The yellow edge emerges and will be added to $H$ afterwards.\label{fig01}}
\end{figure}
%

\begin{description}
\item[Step 1.2]
	Construct the auxiliary graph $\mcG$ from $H$ as in Definition~\ref{def02} and
	compute a maximum matching $\mcM$ in $\mcG$.
	Then, for each edge $\{N_1,N_2\} \in \mcM$, add an edge of $G$ to connect the two edge components of $H$ corresponding to $N_1$ and $N_2$ to form a $4$-path in $H$.

\item[Step 1.3]
	For each edge $\{u,v\} \in E(G)$ such that $u \in V(G)\setminus V(H)$ and $v$ is in a $4$-path or an edge component of $H$,
	add $u$ into a set $U$ and add the edge $\{u,v\}$ into a set $F$.
	Afterwards, add all the vertices of $U$ and all the edges of $F$ into $H$.
\end{description}

\begin{lemma}
\label{lemma02}
Steps 1.1--1.3 take $O(n^{1.5} m^2)$ time.
\end{lemma}
\begin{proof} 
Step~1.2 takes $O(\sqrt{n} m)$ time~\cite{MV80} as computing a maximum matching is dominant;
Step~1.3 takes $O(m)$ time.
As aforementioned, Step~1.1 has $O(n)$ repetitions.
Below, we show that each repetition in Step~1.1 takes $O(\sqrt{n} m^2)$ time.

First we note that an augmenting triple (see Definition~\ref{def01}) can be located, if it exists, in $O(m^2)$ time
by examining a pair of edges that interconnecting the three components in the triple. 
Next, to check if there is an augmenting pair with respect to $H$ (see Definition~\ref{def03}), it suffices to check if $E(G)\setminus E(H)$ contains an edge $\{x,y\}$ 
such that $x$ appears in an edge component $e=\{u,w\}$ of $H$, $y$ appears in a $5$-path component $K=v_1$-$v_2$-$v_3$-$v_4$-$v_5$ of $H$, and S1 and S2 hold.
There are at most $m$ such edges $\{x,y\}$ and it takes $O(1)$ (respectively, $O(\sqrt{n} m)$) time to check if S1 (respectively, $q_4(H')>q_4(H)$) holds for each of them. 
Because we prefer augmenting triples to augmenting pairs in Step~1.1, 
checking whether a new augmenting triple with respect to $H'$ appears is equivalent to checking whether there is 
an augmenting triple involving $\{v_1,v_2\}$ or $\{v_4,v_5\}$, and hence takes $O(m)$ time.
That is, locating an augmenting pair takes $O(\sqrt{n} m^2)$ time.
\end{proof}

Hereafter, $H_i$ denotes the graph $H$ obtained at the termination of Step~1.$i$, for $i = 1, 2$.
In contrast, $H$ will always mean the graph $H$ obtained after Step~1.3, 
and $M$ will always mean the maximum matching $M$ obtained at the termination of Step~1.1, which then is untouched in Steps~1.2 and~1.3.
Clearly, Invariant~\ref{inva01} holds for $H_2$ and $M$.
Moreover, each connected component of $H_2$ is a $5$-path, $4$-path, or edge.

\begin{lemma}
\label{lemma03}
For each vertex $u$ of $V(G) \setminus V(H_2)$, the following statements hold:
\begin{enumerate}
\parskip=0pt
\item
	If $\{u, v\} \in E(G)$ is added to $H_2$ in Step~1.3, then $v$ is an internal vertex of a $4$-path component or an endpoint of an edge component of $H_2$;
\item
	At most two edges incident to $u$ are added to $H_2$ in Step~1.3.
	Moreover, if exactly two edges $\{u, v\}, \{u, w\}$ are added to $H_2$ in Step~1.3, then $\{v, w\}$ is an edge component of $H_2$.
\end{enumerate}
\end{lemma}
\begin{proof} 
Note that there is no augmenting triple or pair in $H_1$ as otherwise Step~1.1 continues,
neither in $H_2$ since Step~1.2 does not generate any new edge or $5$-path component to $H_1$.
We prove the two statements separately as follows.

{\em Statement~1.} Since $\{u, v\}$ is added to $H_2$ in Step~1.3 and $u \in V(G) \setminus V(H_2)$, $v$ is in a $4$-path component or an edge component of $H$. 
Because of Step~1.1, $v$ cannot be an endpoint of a $4$-path component or else an augmenting triple would exist.
Hence, $v$ is an interval vertex of a $4$-path or an endpoint of an edge component of $H_2$.

{\em Statement~2.} We first prove that if $\{u, v\}$ and $\{u, w\}$ are added to $H_2$ in Step~1.3, then $\{v, w\}$ is an edge component of $H_2$.
For a contradiction, we assume $\{v, w\}$ is not an edge component of $H_2$.
Then, by Statement~1, either $v$ and $w$ appear in two different edge components $e_1, e_2$ of $H_2$,
or at least one of $v$ and $ w$ is an internal vertex of a $4$-path component $P$ of $H_2$.
In the first case, $(u, e_1, e_2)$ would have been an augmenting triple, a contradiction against Step~1.1. 
In the second case, without loss of generality, we assume $v$ is an internal vertex of a $4$-path $P$ and is on the edge $e_3 \in P \cap M$ by Invariant~\ref{inva01}.
If $w$ is on an edge component $e_4$ of $H_2$, then $(u, e_3, e_4)$ would have been an augmenting triple, a contradiction against Step~1.1.
Hence, we may further assume that $w$ is an internal vertex of a $4$-path $P'$ and is on the edge $e_5 \in P' \cap M$ by Invariant~\ref{inva01} again.
It follows that $(u, e_3, e_5)$ is an augmenting triple.
That is, in both cases, an augmenting triple exists and thus Step~1.1 would have not terminated.
This proves that $\{v, w\}$ is an edge component of $H_2$.

Secondly, suppose that the three edges $\{u, x\}, \{u, y\}, \{u, z\}$ incident to $u$ are added to $H_2$ in Step~1.3. 
Then, by the discussion in the last paragraph, $\{x, y\}$ and $\{x, z\}$ are edge components of $H_2$ and thus they are in $M$, a contradiction to $M$ being a matching.
\end{proof}

\begin{lemma}
\label{lemma04}
At the termination of Step~1.3, Invariant~\ref{inva01} holds.
\end{lemma}
\begin{proof}
Recall that Invariant~\ref{inva01} holds for $H_2$ and each connected component of $H_2$ is a $5$-path, $4$-path, or edge.
The maximum matching $M$ is untouched in Steps~1.2 and 1.3, and thus it remains so at the termination of Step~1.3.

For each vertex $v\in V(H_2)$, let $N(v)$ be the set of its neighbors in $V(G) \setminus V(H_2)$ that are added in Step~1.3.
By Statement~1 of Lemma~\ref{lemma03}, if $v$ is not an internal vertex of a $4$-path component or an endpoint of an edge component of $H_2$, then $N(v) = \emptyset$.

Consider an edge component $e = \{v_1, v_2\}$ of $H_2$. 
There are three cases depending on $|N(v_1)|$ and $|N(v_2)|$. 
In the first case, both $|N(v_1)|$ and $|N(v_2)|$ are~$0$, and thus $e$ remains an edge component of $H$ at the termination of Step~1.3. 
In the second case where both $|N(v_1)|$ and $|N(v_2)|$ are nonzero, $|N(v_1) \cup N(v_2)| = 1$ because otherwise a contradiction to $M$ being a maximum matching. 
Let $u$ be the unique vertex of $N(v_1) \cup N(v_2)$. 
By Statement~2 of Lemma~\ref{lemma03}, $\{u, v_1\}$ and $\{u, v_2\}$ are added to $H$ in Step~1.3 and form a triangle together with $e$.
In the last case where exactly one of $|N(v_1)|$ and $|N(v_2)|$ is~$0$, and we assume without loss of generality that $|N(v_1)|=0$ and $|N(v_2)|>0$.
For each $u \in N(v_2)$, exactly one edge $\{u, v_2\}$ is added to $H$ by Statement~2 of Lemma~\ref{lemma03},
and the edge component $e=\{v_1, v_2\}$ becomes a star with center $v_2$ at the termination of Step~1.3.

Consider a $4$-path component $P = v_1$-$v_2$-$v_3$-$v_4$ of $H_2$, for which $N(v_1) = N(v_4) = \emptyset$.
By Statement~2 of Lemma~\ref{lemma03}, $N(v_2) \cap N(v_3) = \emptyset$.
Moreover, for each $u \in N(v_2) \cup N(v_3)$, exactly one edge ($\{u, v_2\}$ or $\{u, v_3\}$) incident to $u$ is added to $H_2$ in Step~1.3,
suggesting $P$ becomes a bi-star at the termination of Step~1.3. 

In summary, Invariant~\ref{inva01} holds at the termination of Step~1.3.
\end{proof}

\subsection{Bad components and rescuing them}\label{subsec:rescue}
By Invariant~\ref{inva01}, a connected component $K$ of $H$ is an edge, a triangle, a star, a bi-star or a $5$-path.
Moreover, the above discussion in Section~\ref{subsec:k=5} states that,
if $K$ is not a $5$-path then it is impossible to form a $5^+$-path in the induced subgraph $G[V(K)]$.
This motivates the following definition of {\em bad component}.

\begin{definition}
\label{def04}
A {\em bad component} of $H$ is a connected component that is an edge, a triangle, a star, or a bi-star.
\end{definition}

In the next lemma, we focus on the edges connecting a vertex of a bad component $K$ with another vertex not in $K$, for the purpose of {\em rescuing} some vertices in $K$.
We remark that these edges are vital in the subsequent algorithm design and analysis.

\begin{lemma}
\label{lemma05}
Let $\{v, w \}$ be an edge of $G$ such that $v$ is in a bad component $K$ of $H$ but $w$ is not in $K$.
Then the following statements hold:
\begin{enumerate}
\parskip=0pt
\item If $w$ is not in a $5$-path of $H$, then $w \in V(M)$.
\item If $v \notin V(M)$ or $K$ is a triangle, then $w$ is an internal but not middle vertex of a $5$-path of $H$.
\item If $K$ is an edge or a star, then $w$ is in a bi-star or a $5$-path of $H$.
\end{enumerate}
\end{lemma}
\begin{proof} 
We prove the statements separately as follows.

{\em Statement~1.} 
Suppose that $w$ is not in a $5$-path of $H$ and $w \notin V(M)$. 
Then, $v \in V(M)$ because otherwise $M$ would have not been a maximum matching of $G$.
Recall that each $5$-path of $H$ remains untouched in Step~1.3 and hence is not a bad component.
So, $v$ is in an edge or $4$-path component of $H_2$.
Since $w \notin V(M)$ but each vertex of an edge or $4$-path component of $H_2$ is in $V(M)$, we know that $w \notin V(H_2)$. 
But then, the edge $\{v, w\}$ would have been added to $H_2$ in Step~1.3, a contradiction against the fact that $w$ is not in $K$. 
Hence, $w \in V(M)$. 

{\em Statement~2.}
Assume that $v \notin V(M)$ or $K$ is a triangle. 
If $w$ were an endpoint or the middle vertex of a $5$-path in $H$, then by Invariant~\ref{inva01} $M$ would have not been a maximum matching of $G$, a contradiction.
So, we next prove the statement by contradiction, by assuming that $w$ is not in a $5$-path of $H$.
Then, by Statement~1, $w \in V(M)$ and $w$ is in a bad component of $H$.
Consequently, since $K$ is not a $5$-path, Statement~1 implies that $v \in V(M)$ and thus $K$ is a triangle.
Since $w \in V(M)$, there exists $e_1 \in M$ incident to $w$.
By Invariant~\ref{inva01}, $K$ contains an edge $e_2 \in M$ and a vertex $u \notin V(M)$. 
One can easily verify that $(u, e_1, e_2)$ is an augmenting triple with respect to $H_1$. 
This implies that Step~1.1 would have not stopped, a contradiction. 

{\em Statement~3.}
If $v \notin V(M)$, the statement follows from Statement~2 immediately.
So, we can assume $v \in V(M)$. 
Similarly, we can further assume $w \in V(M)$ by Statement~1. 
Let $e_1$ (respectively, $e_2$) be the edge of $M$ incident to $v$ (respectively, $w$), and $K'$ be the connected component of $H$ containing $w$. 
If $K'$ were a triangle, then $(u,e_1,e_2)$ would have been an augmenting triple with respect to $H_1$, where $u$ is the vertex of $K'$ not incident to $e_2$, a contradiction. 
So, $K'$ is not a triangle. 
Thus, by Invariant~\ref{inva01}, $K'$ is an edge, star, bi-star, or $5$-path. 
Now, if $K'$ were an edge or a star, then we could have been able to include $\{N_1,N_2\}$ in $\mcM$ (cf. Definition~\ref{def02})
where $N_1$ (respectively, $N_2$) is the vertex of $\mcG$ corresponding to $e_1$ (respectively, $e_2$), a contradiction against the maximality of $\mcM$.
Therefore, $K'$ is a bi-star or $5$-path.
\end{proof}

For rescuing the maximum possible number of bad components of $H$, we give the next definition and then present Steps~2.1--2.3 of our algorithm.

\begin{definition}
\label{def05} 
Let $K$ be a bad component of $H$. 
We define the {\em weight} of $K$ to be $|E(K) \cap M|$. 
Let $F$ be a set of edges in $G$ for each of which the two endpoints belong to two different connected components of $H$. 
We say that $F$ {\em rescues} $K$ if at least one edge in $F$ is incident to a vertex of $K$.
We define the {\em weight} of $F$ to be the total weight of the bad components of $H$ rescued by $F$.
\end{definition}

By Invariant~\ref{inva01}, the weight of a bad component $K$ is~$2$ if it is a bi-star, or is~$1$ otherwise.

\begin{description}
\parskip=0pt
\item[Step 2.1] 
	Construct a spanning subgraph $G'$ of $G$ of which the edge set consists of all $\{v_1, v_2\} \in E(G)$ 
	such that $v_1$ and $v_2$ appear in two different components of $H$ of which at least one component is bad. 
\item[Step 2.2]
	Compute a maximum-weighted path-cycle cover $C$ of $G'$ (cf. the proof of Lemma~\ref{lemma06}).
\end{description}

\begin{lemma}
\label{lemma06} 
Step~2.2 takes $O(mn\log n)$ time.
\end{lemma}
\begin{proof}
The proof is a reduction to the maximum-weight $[f, g]$-factor problem. 
Given two functions $f$ and $g$ mapping each vertex $v$ of an edge-weighted graph $G_1$ to two non-negative integers $f(v), g(v)$ with $f(v) \le g(v)$.
An {\em $[f,g]$-factor} of a graph $G_1$ is a set $F$ of edges in $G_1$ such that in the spanning subgraph $(V(G_1), F)$, 
the degree of each vertex $v$ is at least $f(v)$ and at most $g(v)$. 
The {\em weight} of an $[f,g]$-factor $F$ of $G_1$ is the total weight of  edges in $F$. 
A maximum-weighted $[f,g]$-factor of $G_1$ can be computed in $O(m_1n_1\log n_1)$ time~\cite{Gab83}, where $m_1 = |E(G_1)|$ and $n_1 = |V(G_1)|$. 

Let $B_1, B_2, \ldots, B_h$ be the bad components of $H_5$.
We construct an auxiliary edge-weighted graph $G_1 = (V(G) \cup X, E(G') \cup F_1 \cup F_2)$ as follows:
\begin{itemize}
\parskip=0pt
\item $X = \{x_i, y_i, z_i \mid 1 \le i \le h \}$. 
\item $F_1 = \{\{x_i, v\}, \{y_i, v\} \mid v \in V(B_i), 1 \le i \le h\}$ and $F_2 = \{\{x_i, z_i\}, \{y_i, z_i\} \mid 1 \le i \le h\}$.
\item The weight of each edge in $E(G_1) \cup F_1$ is $0$.
\item For the edges of $F_2$, if $B_i$ is a bi-star, then the weight of each of $\{x_i, z_i\}, \{y_i, z_i\}$ is $2$;
	otherwise, the weight of each of $\{x_i, z_i\}, \{y_i, z_i\}$ is $1$.
\item For each vertex $v \in \bigcup_{i = 1}^h V(B_i)$, let $f(v) = g(v) = 2$. 
\item For each $v \in V(G) - \bigcup_{i = 1}^h V(B_i)$, let $f(v) = 0$ and $g(v) = 2$.
\item For each $i \in\{ 1, 2, \ldots, h\}$, $f(x_i) = f(y_i) = f(z_i) = 0$ and $g(x_i) = g(y_i) = |V(B_i)|$, $g(z_i) = 1$.
\end{itemize}

We next prove that the maximum weight of an $[f,g]$-factor of $G_1$ equals the maximum weight of a path-cycle cover of $G'$. 

Given a maximum-weighted path-cycle cover $C$ of $G'$, we can obtain an $[f, g]$-factor $F$ for $G_1$ as follows: 
Initially, we set $F = C$. 
Then, for each bad component $B_i$ and each vertex $v$ in $B_i$, we perform one of the following according to the degree of $v$ in the graph $(V(G), F)$. 
\begin{itemize}
\parskip=0pt
\item If the degree of $v$ in the graph $(V(G), F)$ is~$0$, then add the edges $\{v, x_i\}, \{v, y_i\}$ to $F$.
\item If the degree of $v$ in the graph $(V(G), F)$ is~$1$, then add the edge $\{v, x_i\}$ to $F$, and further add the edge $\{y_i, z_i\}$ to $F$ if it has not been added to $F$.
\item If the degree of $v$ in the graph $(V(G), F)$ is~$2$, then add the edge $\{y_i, z_i\}$ to $F$ if it has not been added to $F$.
\end{itemize}
Clearly, $F$ is an $[f,g]$-factor of $G'$. 
We claim that the weight of $F$ is no less than that of $C$.
To see this, consider a bad component $B_i$ rescued by $C$. 
Then, there exists a vertex $v$ in $B_i$ such that $C$ contains an edge incident to $v$. 
Hence, by the construction of $F$, $F$ contains $\{y_i, z_i\}$. 
If $B_i$ is a bi-star, then the weight of $\{y_i, z_i\}$ is $2$; otherwise, the weight of $\{y_i, z_i\}$ is $1$.
So, the claim holds. 

Conversely, given a maximum-weight $[f, g]$-factor $F$ of $G_1$, we obtain a subset $C$ of $E(G')$ with $C = E(G') \cap F$. 
Since $g(v) = 2$ for each vertex $v \in V(G_1)$, $C$ is a path-cycle cover of $G'$.
We claim that the weight of $C$ is no less than that of $F$. 
To see this, consider a bad component $B_i$ such that $\{x_i, z_i\}$ or $\{y_i, z_i\}$ is in $F$.
Since $g(z_i)=1$, exactly one of $\{x_i, z_i\}$ and $\{y_i, z_i\}$ is in $F$. 
Without loss of generality, we assume $\{x_i, z_i\}$ is in $F$. 
Then, there exists a vertex $v$ in $B_i$ such that the edge $\{v, x_i\}$ is not in $F$. 
Since $f(v) = g(v) = 2$, $C$ contains an edge incident to $v$ and hence $C$ rescues $B_i$. 
So, the claim holds. 

By the above two claims,  the maximum weight of an $[f,g]$-factor of $G_1$ equals the maximum weight of a path-cycle cover of $G'$. 
Now, since $|V(G_1)|\le 4n$ and $|E(G_1)|\le m+4n$, the running time is at most $O(\max\{mn, n^2\}\log n) = O(mn \log n)$ because $G$ is a connected graph and hence $m\ge n-1$.
So, the lemma holds.
\end{proof}

\begin{notation}
\label{nota02}
For the spanning subgraph graph $G'$ constructed in Step~2.1,
\begin{itemize}
\parskip=0pt
\item $C$ denotes a maximum-weighted path-cycle cover of $G'$;
\item $M_C$ denotes the set of all the edges of $M$ that appear in a $5$-path of $H$ or in a bad component of $H$ rescued by $C$. 
\end{itemize}
\end{notation}

The following lemma shows that $|V(M_C)|$ is relatively large compared with $opt(G)$.

\begin{lemma}
\label{lemma07}
$|V(M_C)| \ge \frac 45 opt(G)$.
\end{lemma}
\begin{proof}
Let $m_b$ be the number of edges of $M$ each appears in a bad component of $H$. 
Let $A_1, \ldots, A_s$ be the bi-stars of $H$ for each of which none of its vertices is incident to any edge in $OPT(G)$;
similarly, let $A_{s+1}, \ldots, A_{s+t}$ be the edge, triangle, or star components of $H$ for each of which none of its vertices is incident to any edge in $OPT(G)$.
Note that $E(G') \cap E(OPT(G))$ is a path-cycle cover of $G'$ with weight $m_b - 2s - t$.

Let $\ell$ be the number of bi-stars not rescued by $C$, and $h$ be the total number of edge, triangle, or star components not rescued by $C$.
Then, $|M| = |M_C| + 2\ell + h$ since each bi-star has two edges in $M$ and each other bad component has one edge in $M$. 
Note that the weight of the path-cycle cover $C$ is $m_b - 2\ell - h$.
Since $C$ is a maximum-weight path-cycle cover of $G'$, $m_b - 2s - t \le m_b - 2\ell - h$ and in turn $2s + t \ge 2\ell + h$. 

A crucial point is that for each bad component $A_i$, for $1\le i\le s+t$, no vertex of $A_i$ can appear in $OPT(G)$
because $opt(A_i) = 0$ and $OPT(G)$ has no edge connecting $A_i$ to the outside.
It follows that $OPT(G)$ is actually an optimal solution for the graph $G_o$ obtained from $G$ by removing the vertices of $A_i$ for every $i\in\{1,\ldots, s+t\}$. 
So, by Lemma~\ref{lemma01}, $|V(M_o)| \ge \frac 45 opt(G)$, where $M_o$ is a maximum matching in $G_o$. 

Note that $M_o\bigcup \left(\bigcup^{s+t}_{i=1} E(A_i) \cap M\right)$ is a matching of $G$,
and its size is $|M_o| + 2s + t$ because $|E(A_i) \cap M|=2$ for $1\le i\le s$ and $|E(A_{s+j}) \cap M|=1$ for $1\le j\le t$ (Invariant~\ref{inva01}).
Since $M$ is a maximum matching of $G$, $|M_o| + 2s + t \le |M| = |M_C| + 2\ell + h$.
It follows from $2s + t \ge 2\ell + h$ that $|M_o| \le |M_C|$, and hence $|V(M_C)| \ge \frac 45 opt(G)$.
\end{proof}

We will modify $C$ by several operations later on, but always maintain it to be a maximum-weighted path-cycle cover of $G'$.
During these modifications, the graph $H$ is unchanged and thus Lemma~\ref{lemma07} continues to hold for the corresponding $M_C$.

\begin{notation}
\label{nota03}
For the graph $H$ and a maximum-weighted path-cycle cover $C$ of $G'$,
\begin{itemize}
\parskip=0pt
\item
	$H+C$ denotes the spanning subgraph $(V(G), E(H)\cup C)$.
\item
	$\widehat{H+C}$ denotes the graph obtained from $H+C$ by contracting each connected component of $H$ into a single node. 
	In other words, the nodes of $\widehat{H+C}$ one-to-one correspond to the connected components of $H$ and
	two nodes are adjacent in $\widehat{H+C}$ if and only if $C$ contains an edge between the two corresponding connected components.
\item
	This way, each connected component $K$ of $H+C$ is partially contracted into a connected component 	$\widehat{K}$ of $\widehat{H+C}$.
	In the sequel, we always use $K$ to refer to a connected component of $H+C$,
	and $\widehat{K}$ denotes the connected component of $\widehat{H+C}$ corresponding to $K$. 
\end{itemize}
\end{notation}

%
\begin{description}
\parskip=0pt
\item[Step 2.3]
	Repeatedly remove an edge $e$ from $C$, that is, $C$ is updated to $C \setminus \{e\}$, if $C \setminus \{e\}$ has the same weight as $C$.
\end{description}

\begin{lemma}
\label{lemma08} 
The following statements hold at the termination of Step~2.3:
\begin{enumerate}
\parskip=0pt
\item $\widehat{K}$ is an isolated node, an edge, or a star.
\item If $\widehat{K}$ is an edge, then at least one endpoint of $\widehat{K}$ corresponds to a bad component of $H$.
\item If $\widehat{K}$ is a star, then each satellite of $\widehat{K}$ corresponds to a bad component of $H$.
\end{enumerate}
\end{lemma}
\begin{proof}
It is done by a simple contradiction as no edge of $C$ can be removed while keeping its weight.
\end{proof}

We call $K$ a {\em composite component} of $H+C$ if it contains two or more connected components of $H$.
By Lemma~\ref{lemma08}, for convenience,
if $\widehat{K}$ is an edge, we choose an endpoint corresponding to a bad component of $H$ as the {\em satellite} of $\widehat{K}$, while the other endpoint as the {\em center}.
This way, all the satellites of a composite component $\widehat{K}$ correspond to bad components of $H$.

\begin{definition}
\label{def06}
For each composite component $K$ of $H+C$, its {\em center element} denoted as $K_c$ is the component of $H$ corresponding to the center of $\widehat{K}$;
the other components of $H$ contained in $K$ are the {\em satellite-elements} of $K$.
\end{definition}

\begin{lemma}
\label{lemma09}
The following statements hold:
\begin{enumerate}
\parskip=0pt
\item A satellite-element of $K$ is an edge, a triangle, a star, or a bi-star, but not a $5$-path.
	If it is a triangle, then $K_c$ is a $5$-path.
\item $K_c$ is an edge, a star, a bi-star, or a $5$-path, but not a triangle.
	If $K_c$ is an edge or a star, then each satellite-element of $K$ is a bi-star.
\end{enumerate}
\end{lemma}
\begin{proof}
Recall that each satellite-element of $K$ is a bad component and hence cannot be a $5$-path. 
If a satellite-element of $K$ is a triangle, then $K_c$ is a $5$-path by Statement~2 of Lemma~\ref{lemma05}.
This proves the first statement.

If $K_c$ were a triangle, then by Statement~2 of Lemma~\ref{lemma05}, each satellite-element of $K$ would be a $5$-path, a contradiction against the first statement.
If $K_c$ is an edge or a star, then each satellite-element of $K$ is a bi-star because of Statement~3 of Lemma~\ref{lemma05} and the first statement. 
That is, the second statement holds.
\end{proof}

\subsection{Critical components} \label{sec:2ndAlgo}
For convenience, we define the {\em trunk} of $K$, denoted by $\widetilde{K}$, to be the subgraph of $K$ obtained by modifying $K$ as follows (cf. Figure~\ref{fig02}). 
\begin{itemize}
\parskip=0pt
\item For each satellite-element $S$ of $K$ that is a star or a bi-star, remove all vertices of $V(S) \setminus V(M \cup C)$ together with the edges incident to them.
\item If $K_c$ is a bad component of $H$, remove all vertices of $V(K_c) \setminus V(M)$ together with the edges incident to them.
\end{itemize}

\begin{figure}[thb]
\begin{center}
\framebox{
\begin{minipage}{0.5\textwidth}
\begin{tikzpicture}[scale=0.5,transform shape]
\draw [thick, line width = 1pt] (-8, 2) -- (-7, 2);
\draw [yellow, thick, line width = 1pt] (-7, 2) -- (-6, 2);
\draw [thick, line width = 1pt] (-6, 2) -- (-5, 2);
\draw [green, thick, line width = 1pt] (-6, 2) -- (-6, 0.5);
\draw [yellow, thick, line width = 1pt] (-6, 0.5) -- (-6, -0.5);
\draw [thick, line width = 1pt] (-6, 0.5) -- (-5, 0.5);
\draw [thick, line width = 1pt] (-6, -0.5) -- (-5, -0.5);
\draw [thin, line width = 0.5pt] (-6, 0.5) -- (-6.6, 0.5);
\draw [thin, line width = 0.5pt] (-6, 0.5) -- (-6.6, 0.8);
\draw [thin, line width = 0.5pt] (-6, 0.5) -- (-6.6, 0.2);
\draw [thin, line width = 0.5pt] (-7, 2) -- (-7, 1.4);
\draw [thin, line width = 0.5pt] (-7, 2) -- (-7.3, 1.4);
\draw [thin, line width = 0.5pt] (-7, 2) -- (-6.7, 1.4);
\filldraw (-8, 2) circle(.15);
\filldraw (-7, 2) circle(.15);
\filldraw (-6, 2) circle(.15);
\filldraw (-5, 2) circle(.15);
\filldraw (-6,0.5) circle(.15);
\filldraw (-6,-0.5) circle(.15);
\filldraw (-5,0.5) circle(.15);
\filldraw (-5,-0.5) circle(.15);
\filldraw[fill = white] (-7,1.4) circle(.15);
\filldraw[fill = white] (-7.3,1.4) circle(.15);
\filldraw[fill = white] (-6.7,1.4) circle(.15);
\filldraw[fill = white] (-6.6,0.2) circle(.15);
\filldraw[fill = white] (-6.6,0.5) circle(.15);
\filldraw[fill = white] (-6.6,0.8) circle(.15);
\node[font=\fontsize{18}{6}\selectfont] at (-4,0.7) {$K$};

\draw [-latex, thick] (-2, 0.7) to (1, 0.7);

\draw [thick, line width = 1pt] (3, 2) -- (4, 2);
\draw [yellow, thick, line width = 1pt] (4, 2) -- (5, 2);
\draw [thick, line width = 1pt] (5, 2) -- (6, 2);
\draw [green, thick, line width = 1pt] (5, 2) -- (5, 0.5);
\draw [yellow, thick, line width = 1pt] (5, 0.5) -- (5, -0.5);
\draw [thick, line width = 1pt] (5, 0.5) -- (6, 0.5);
\draw [thick, line width = 1pt] (5, -0.5) -- (6, -0.5);
\filldraw (3, 2) circle(.15);
\filldraw (4, 2) circle(.15);
\filldraw (5, 2) circle(.15);
\filldraw (6, 2) circle(.15);
\filldraw (5,0.5) circle(.15);
\filldraw (5,-0.5) circle(.15);
\filldraw (6,0.5) circle(.15);
\filldraw (6,-0.5) circle(.15);
\node[font=\fontsize{18}{6}\selectfont] at (7,0.7) {$\widetilde{K}$};
\end{tikzpicture}
\end{minipage}}
\end{center}
\captionsetup{width=1.0\linewidth}
\caption{An illustration of modifying $K$ into $\widetilde{K}$.
	The filled (respectively, blank) vertices are in (respectively, not in) $V(M)$.
	The thick (thin, respectively) edges are in the matching $M$ (not in $M \cup C$ but in $H$, respectively) and
	the green (yellow, respectively) edges are in $C$ (in $\mcM$, see Step~1.2, respectively).\label{fig02}}
\end{figure}
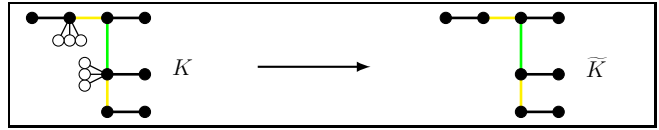

\begin{lemma}
\label{lemma10}
$C \cap E(K) = C \cap E(\widetilde{K})$ and $|V(\widetilde{K})| \le 55$. 
\end{lemma}
\begin{proof}
Clearly, $C \cap E(\widetilde{K}) \subseteq C \cap E(K)$. 
To prove that $C \cap E(K) \subseteq C \cap E(\widetilde{K})$, consider an edge $e = \{v, w\} \in C \cap E(K)$. 
Without loss of generality, we assume $v \in V(K_c)$ and $w$ is in a satellite-element $S$ of $K$.
Since $S$ is a bad component of $H$, Statement~1 of Lemma~\ref{lemma05} implies that either $K_c$ is a $5$-path or $v \in V(M)$. 
In both cases, $v \in V(\widetilde{K})$. 
By the construction of $\widetilde{K}$, $e$ remains in $\widetilde{K}$. 
Hence, $C \cap E(K) \subseteq C \cap E(\widetilde{K})$ and in turn $C \cap E(K) = C \cap E(\widetilde{K})$.  

We next prove $|V(\widetilde{K})| \le 55$. 
By Statement~2 of Lemma~\ref{lemma09}, $K_c$ is an edge, a star, a bi-star, or a $5$-path. 
If $K_c$ is an edge, a star, or a bi-star, then $|V(K_c) \cap V(M)| \le 4$ by Invariant~\ref{inva01}, and the vertices in $V(K_c) \setminus V(M)$ are not in $\widetilde{K}$.
If $K_c$ is a $5$-path, then clearly it remains in $\widetilde{K}$. 
Thus, at most five vertices of $K_c$ remain in $\widetilde{K}$.

By Statement~1 of Lemma~\ref{lemma09}, each satellite-element $S$ of $K$ is an edge, a triangle, a star, or a bi-star.
If $S$ is an edge or a triangle, then at most three vertices of $S$ remain in $\widetilde{K}$.
If $S$ is a star or a bi-star, then $|V(S) \cap V(M)| \le 4$ and $|V(S) \cap V(C)| = 1$ by Invariant~\ref{inva01},
and the vertices in $V(S) \setminus V(M \cup C)$ are not in $\widetilde{K}$. 
Thus, at most five vertices of $S$ remain in $\widetilde{K}$.

Since $C$ is a path-cycle cover of $G'$, each vertex of $K_c$ can be connected to at most two satellite-elements of $K$ by the edges of $C \cap E(K)$.
From the above $C \cap E(K) = C \cap E(\widetilde{K})$ and at most five vertices of $K_c$ remain in $\widetilde{K}$, at most ten satellite-elements of $K$ remain in $\widetilde{K}$.
Therefore, $|V(\widetilde{K})| \le 5 + 5 \times 10 = 55$.
\end{proof}

By Lemma~\ref{lemma10}, $|V(\widetilde{K})| \le 55$ and thus $OPT(\widetilde{K})$ can be computed in $O(1)$ time.
We define {\em critical components} of $H+C$ in the following.

\begin{definition}
\label{def07}
For a composite component $K$ of $H+C$, set $s(K) = |V(K)\cap V(M_C)|$ and $\eta(K) = opt(\widetilde{K})$.
Set $\alpha = \frac {15}8$. 
A {\em critical component} of $H+C$ is a connected component $K$ of $H+C$ with $\frac {s(K)}{\eta(K)} \ge \alpha$.
(Figure~\ref{fig03} shows all possible structures of a critical component.)
\end{definition}

Although $\widetilde{K}$ has a simpler structure than $K$, the exact value of $\eta(K)$ still takes time to compute.
In the remainder of this subsection, we establish a lower bound on $\eta(K)$ by finding a feasible solution for $\widetilde{K}$.
We need the following definitions and notations.

\begin{definition}
\label{def08}
For a composite component $K$ of $H+C$, a vertex in both $K_c$ and $\widetilde{K}$ is called an {\em anchor} of $K$.

For each satellite-element $S$ of $K$, the unique edge $e\in C$ connecting $S$ to an anchor $v$ in $K$ is called the {\em rescue-edge} for $S$ and 
$v$ is called the {\em supporting-anchor} for $S$. 
We say that $e$ {\em rescues} $S$ and $v$ {\em anchors} $S$.

For a positive integer $j$, an anchor $v$ is a {\em $j$-anchor} if $v$ anchors exactly $j$ satellite-elements of $K$.
\end{definition}

We now focus on an anchor $v$ of $K$.
Since $C$ is a path-cycle cover of $G'$, $v$ can anchor at most two distinct satellite-elements of $K$. 
That is, $v$ is a $0$-, $1$-, or $2$-anchor.
By Statement~1 of Lemma~\ref{lemma09}, each satellite-element of $K$ is either a triangle, an edge, a star, or a bi-star.
In the next definition, we distinguish several types of $1$-anchors and $2$-anchors.

\begin{definition}
\label{def09}
For each $i\in\{0,1\}$, $O_i$ denotes the set of all vertices $v$ in $H+C$ such that 
$v$ is a $1$-anchor of some connected component in $H+C$ and $v$ anchors exactly $i$ bi-stars. 
Similarly, for each $i\in\{0,1,2\}$, $T_i$ denotes the set of all vertices $v$ in $H+C$ such that 
$v$ is a $2$-anchor of some connected component in $H+C$ and $v$ anchors exactly $i$ bi-stars.
\end{definition}

As we mentioned before, we want to establish a lower bound on $\eta(K)$.
Consider an anchor $v$ of $K$ and a satellite-element $S$ of $K$ anchored by $v$.
Let $\{v,u\}$ be the rescue-edge for $S$. 
It is possible that a longest path in $S$ starting at $u$ can be combined with the edge $\{v,u\}$ to form a $4^+$-path starting at $v$. 
This motivates the following notation.

\begin{notation}
\label{nota04}
Let $v$ be an anchor of $K$. 
\begin{itemize}
\parskip=0pt
\item If $v$ is a $0$-anchor, then $Q_v$ denotes the $1$-path consisting of $v$ only;
\item if $v$ is a $1$- or $2$-anchor, then $Q_v$ denotes a longest path in $\widetilde{K}$ that starts at $v$, then crosses an edge $\{v,u\} \in C$, and 
	further traverses a longest path in the satellite-element $S$ of $K$ containing $u$.
\item If $v$ is a $2$-anchor, then $P_v$ denotes a longest path in $\widetilde{K}$ that contains the two edges of $C$
	incident to $v$.
\end{itemize}
\end{notation}

\begin{remark}
\label{remark01}
When $v$ is a $2$-anchor, $Q_v$ can be a portion of $P_v$.
So, if $v \in O_0 \cup T_0$, then $Q_v$ is a $3^+$-path; 
if $v\in O_1 \cup T_1 \cup T_2$, then $Q_v$ is a $4^+$-path.
Moreover, if $v \in T_0$, then $P_v$ is a $5^+$-path; 
if $v\in T_1$, then $Q_v$ is a $6^+$-path;
lastly, if $v\in T_2$, then $P_v$ is a $7^+$-path.
\end{remark}

By Invariant~\ref{inva01}, if $K_c$ is an edge or star, then exactly two vertices of $K_c$ (which are the endpoints of the unique edge in both $M$ and $K_c$) may be anchors of $K$. 
Similarly, if $K_c$ is a bi-star, then exactly four vertices of $K_c$ (which are the endpoints of the two edges in both $M$ and $K_c$) may be anchors of $K$. 
Otherwise, $K_c$ is a $5$-path and all vertices of $K_c$ may be anchors of $K$. 
For ease of presentation, we define the following notation.

\begin{notation}
\label{nota05}
If $K_c$ is an edge or star, then $v_1$-$v_2$ denotes the unique edge in both $M$ and $K_c$.
If $K_c$ is a bi-star, then $v_1$-$v_2$-$v_3$-$v_4$ denotes a $4$-path in $K_c$ such that $\{v_1, v_2\}$ and $\{v_3, v_4\}$ are in $M$.
Lastly, if $K_c$ is a $5$-path, then $v_1$-$v_2$-$v_3$-$v_4$-$v_5$ denotes the path, where $\{v_1, v_2\}$ and $\{v_4, v_5\}$ are in $M$. 

In each case, an anchor of $K$ is a $v_i$ for some $i$.
\end{notation}

\begin{lemma}
\label{lemma11}
$K$ has at most five anchors.
Moreover, the following statements hold:
\begin{enumerate}
\parskip=0pt
\item If $K_c$ is an edge or star, then no anchor of $K$ is in $O_0 \cup T_0 \cup T_1$.
\item If $K_c$ is a bi-star, then at most two anchors are in $O_0 \cup T_0 \cup T_1$. 
	Moreover, if exactly two anchors $v_i, v_j$ with $j \ge i+1$ are in $O_0 \cup T_0 \cup T_1$, then $(i, j) \in \{(1, 2), (3, 4)\}$.
\item If $K_c$ is a $5$-path, then at most two anchors are in $O_0 \cup T_0 \cup T_1$. 
	Moreover, if exactly two anchors $v_i, v_j$ with $j \ge i+1$ are in $O_0 \cup T_0 \cup T_1$, then $(i, j) \in \{(1, 2), (2, 4), (4, 5)\}$.
\end{enumerate}
\end{lemma}
\begin{proof}
The paragraph before Notation~\ref{nota05} concludes that $K$ has at most five anchors.
For each $v_j \in O_0 \cup T_0 \cup T_1$, we use $S_j$ to denote an arbitrary satellite-element of $K$ that is anchored by $v_j$ and is an edge, triangle, or star.
The existence of $S_j$ together with Statement~3 of Lemma~\ref{lemma05} implies Statement~1 of this lemma. 

We use $w_j$ to denote the unique vertex of $S_j$ with $\{v_j,w_j\}\in C$.
If $w_j \in V(M)$, then by Invariant~\ref{inva01}, a unique edge in $M$ is incident to $w_j$ and we use $e_j$ to denote this edge.
We use these notations to prove the next two statements separately.

{\em Statement~2.}
It suffices to show that if $v_1$ or $v_2$ is in $O_0 \cup T_0 \cup T_1$, then neither $v_3$ nor $v_4$ is in $O_0 \cup T_0 \cup T_1$.
For a contradiction, assume that both $v_i$ and $v_j$ are in $O_0 \cup T_0 \cup T_1$ for some $i\in\{1,2\}$ and some $j\in\{3,4\}$.
We suppose that $i=1$ and $j=3$; the other cases are similar. 
By Statement~2 of Lemma~\ref{lemma05}, both $w_1$ and $w_3$ are in $V(M)$. 
Recall that we compute a maximum matching $\mcM$ in a graph $\mcG$ and use it to create as many $4$-path components in $H$ as possible in Step~1.2. 
Indeed, $\{v_2, v_3\}$ was added to $H$ in Step~1.2 because $\mcM$ contains an edge $e$ whose endpoints correspond to the two edges $\{v_1,v_2\}$ and $\{v_3,v_4\}$ of $M$. 
But now, we would have been able to modify $\mcM$ by replacing $e$ with two other edges so that instead of adding $\{v_2,v_3\}$ to $H$ in Step~1.2,
both $\{v_1, w_1\}$ and $\{v_3, w_3\}$ are added to $H$ in Step~1.2. 
In other words, we would have been able to increase the size of $\mcM$, a contradiction against the maximality of $\mcM$.
So, the statement is proven.

{\em Statement~3.}  
We first claim that if $v_3 \in O_0 \cup T_0 \cup T_1$, then none of $v_1$, $v_2$, $v_4$, and $v_5$ is in $O_0 \cup T_0 \cup T_1$.
For a contradiction, assume that both $v_3$ and $v_i$ are in $O_0 \cup T_0 \cup T_1$ for some $i\in\{1,2,4,5\}$. 
We suppose that $i=2$; the other cases are similar.
Then, $K$ has a satellite-element $S_3$ anchored by $v_3$ that is an edge or a star. 
Let $e_3$ be the unique edge in both $M$ and $S_3$.
Now, one can easily verify that $H_1$ (i.e., the graph $H$ immediately after Step~1.1) would have contained an augmenting pair consisting of the $5$-path $K_c$ and $e_3$, a contradiction.
So, the claim holds.

We next claim that if $v_1 \in O_0 \cup T_0 \cup T_1$, then none of $v_3$, $v_4$, and $v_5$ is in $O_0 \cup T_0 \cup T_1$.
To this end, suppose $v_1\in O_0 \cup T_0 \cup T_1$. 
Then, $K$ has a satellite-element $S_1$ anchored by $v_1$ that is an edge or a star. 
Let $e_1$ be the unique edge in both $M$ and $S_1$.
Clearly, $v_3 \notin O_0 \cup T_0 \cup T_1$ because $v_1 \in O_0 \cup T_0 \cup T_1$.
If $v_4$ or $v_5$ were in $O_0 \cup T_0 \cup T_1$, then $H_1$ would have contained an augmenting pair consisting of the $5$-path $K_c$ and $e_1$, a contradiction.
Thus, the second claim holds.
By this claim and symmetry, we also know that if $v_5\in O_0 \cup T_0 \cup T_1$, none of $v_1$, $v_2$, and $v_3$ is in $O_0 \cup T_0 \cup T_1$. 
Therefore, the statement holds.
\end{proof}

We now start bounding $\eta(K)$ and $\frac {s(K)}{\eta(K)}$, for different cases of a composite component $K$.

\begin{lemma}
\label{lemma12}
Suppose that $K_c$ is a bi-star or $5$-path.
Then, the following statements hold:
\begin{enumerate}
\parskip=0pt
\item If $|(O_1\cup T_1\cup T_2)\cap V(K_c)| \ge 4$, then $\eta(K)\ge 16$.
\item If $|(O_1\cup T_1\cup T_2)\cap V(K_c)| \ge 3$ and $|(O_0\cup T_0)\cap V(K_c)| \ge 1$, then $\eta(K)\ge 15$.
\item If $|(O_1\cup T_1\cup T_2)\cap V(K_c)| \ge 3$, then $\eta(K)\ge 13$.
\item If the total number of $1$- and $2$-anchors in $K_c$ is at least two, then $\eta(K) \ge 6$.
	Moreover, if $|(O_1 \cup T_1 \cup T_2)\cap V(K_c)| \ge 1$, then $\eta(K) \ge 7$.
\item If $|(T_0 \cup T_1)\cap V(K_c)| \ge 1$ and $|(O_0 \cup O_1)\cap V(K_c)| \ge 2$, then $\eta(K) \ge 9$.
\end{enumerate}
\end{lemma}
\begin{proof} We prove the statements separately as follows.

{\em Statement~1.}
Assume that $|(O_1\cup T_1\cup T_2)\cap V(K_c)| \ge 4$.
Then, no matter $K_c$ is a bi-star or $5$-path, we can always find four anchors $v_{j_1}, v_{j_2}, v_{j_3}, v_{j_4}$ in $O_1\cup T_1\cup T_2$ such that $1 \le j_1 < j_2 < j_3 < j_4 \le 5$.
So, we can construct two vertex-disjoint $8^+$-paths by using the subpath of $K_c$ from $v_{j_i}$ to $v_{j_{i+1}}$ to connect $Q_{v_{j_i}}$ and $Q_{v_{j_{i+1}}}$ for each $j\in\{1, 3\}$.
Therefore, $\eta(K) \ge 16$ and we are done.

{\em Statement~2.}
Assume that $|(O_1\cup T_1\cup T_2)\cap V(K_c)| \ge 3$ and $|(O_0\cup T_0)\cap V(K_c)| \ge 1$.
Then, no matter $K_c$ is a bi-star or $5$-path, 
we can always find four anchors $v_{j_1}, v_{j_2}, v_{j_3}, v_{j_4}$ such that
$1 \le j_1 < j_2 < j_3 < j_4 \le 5$ and three of them are in $O_1\cup T_1\cup T_2$ and the other one is in $O_0\cup T_0$.
Then, for each $i\in\{1,3\}$, we can use the subpath of $K_c$ from $v_{j_i}$ to $v_{j_{i+1}}$ to connect $Q_{v_{j_i}}$ and $Q_{v_{j_{i+1}}}$ into a path; 
one of the paths is an $8^+$-path and the other is a $7^+$-path.
Therefore, $\eta(K) \ge 15$ and we are done.

{\em Statement~3.}
Assume that $|(O_1\cup T_1\cup T_2)\cap V(K_c)| \ge 3$.
Then, no matter $K_c$ is a bi-star or $5$-path, we can always find three anchors $v_{j_1}, v_{j_2}, v_{j_3}$ in $O_1\cup T_1\cup T_2$ such that $1 \le j_1 < j_2 < j_3 \le 5$.
First suppose that $v_{j_1}$ or $v_{j_3}$ is not of degree~1 in $K_c$.
Without loss of generality, we assume $v_{j_3}$ is not of degree~1 in $K_c$.
Then, we can use the subpath of $K_c$ from $v_{j_1}$ to $v_{j_2}$ to connect $Q_{v_{j_1}}$ and $Q_{v_{j_2}}$ into an $8^+$-path
and the edge $\{v_{j_3}, v_{j_3+1}\}$ to connect $Q_{v_{j_3}}$ into a $5^+$-path.
Therefore, $\eta(K) \ge 13$ and we are done.

Next suppose that both $v_{j_1}$ and $v_{j_3}$ are of degree~1 in $K_c$.
Then, either $j_2 \ge j_1+2$ or $j_2 \le j_3-2$.
We assume $j_2 \le j_3-2$; the other case is similar.
Since $j_2 \le j_3-2$, we can use the subpath of $K_c$ from $v_{j_1}$ to $v_{j_2}$ to connect $Q_{v_{j_1}}$ and $Q_{v_{j_2}}$ into an $8^+$-path and
use the edge $\{v_{j_3}, v_{j_3-1}\}$ to connect $Q_{v_{j_3}}$ and $Q_{v_{j_3-1}}$ into a $5^+$-path.
Therefore, $\eta(K) \ge 13$ again and we are done.

{\em Statement~4.}
Assume that the total number of $1$- and $2$-anchors is at least two.
Then, no matter $K_c$ is a bi-star or $5$-path, we can always find two anchors $v_{j_1}, v_{j_2}$ such that $1 \le j_1 < j_2 \le 5$ and
neither $v_{j_1}$ nor $v_{j_2}$ is a $0$-anchor.
So, we can use the subpath of $K_c$ from $v_{j_1}$ to $v_{j_2}$ to connect $Q_{v_{j_1}}$ and $Q_{v_{j_2}}$ into a $6^+$-path.
Thus, $\eta(K) \ge 6$ and we are done.
Similarly, we can prove $\eta(K) \ge 7$ when $|(O_1 \cup T_1 \cup T_2)\cap V(K_c)| \ge 1$.

{\em Statement~5.}
Assume that $|(T_0 \cup T_1)\cap V(K_c)| \ge 1$ and $|(O_0 \cup O_1)\cap V(K_c)| \ge 2$.
Then, no matter $K_c$ is a bi-star or $5$-path, we can always find three anchors $v_{j_1}, v_{j_2}, v_{j_3}$ such that $1 \le j_1 < j_2 < j_3 \le 5$ and
one of them is in $T_0\cup T_1$ but the other two are in $O_0\cup O_1$.
First suppose $v_{j_1} \in T_0 \cup T_1$.
Then, we can construct a $6^+$-path by using the subpath of $K_c$ between $v_{j_2}$ and $v_{j_3}$ to connect $Q_{v_{j_2}}$ and $Q_{v_{j_3}}$.
This $6^+$-path together with the $5^+$-path $P_{v_{j_1}}$ implies that $\eta(K) \ge 11$ and we are done.
Similarly, if $v_{j_3} \in T_0 \cup T_1$, then $\eta(K) \ge 11$ too.

Next suppose that $v_{j_2} \in T_0 \cup T_1$.
By Statements~2 and~3 in Lemma~\ref{lemma11}, $K_c$ has at most two anchors in $O_0 \cup T_0 \cup T_1$.
It follows that either $|O_1 \cap V(K_c)| \ge 2$ or $|O_0 \cap V(K_c)|=|O_1 \cap V(K_c)|=1$.
In the former case, we can further assume $\{v_{j_1}, v_{j_3}\} \subseteq O_1$.
Then, we can use the subpath of $K_c$ between $v_{j_1}$ and $v_{j_3}$ to connect $Q_{v_{j_1}}$ and $Q_{v_{j_3}}$ into a $9^+$-path.
It follows that $\eta(K) \ge 9$ and we are done.
In the latter case, we can further assume $v_{j_1} \in O_0$ and $v_{j_3} \in O_1$;  
the case where $v_{j_1} \in O_1$ and $v_{j_3} \in O_0$ is similar. 
If $j_3 \ge j_1+3$, then we can still use the subpath of $K_c$ between $v_{j_1}$ and $v_{j_3}$ to connect $Q_{v_{j_1}}$ and $Q_{v_{j_3}}$ into a $9^+$-path,
implying $\eta(K) \ge 9$ and we are done. 
So, we may assume  $j_3 \le j_1+2$.
Then, $j_3=j_2+1=j_1+2$.
By Statements~2 and~3 of Lemma~\ref{lemma11} again, $(j_1, j_2, j_3)=(1, 2, 3)$.
Now, we can use the edge $\{v_3, v_4\}$ to connect $Q_{v_3}$ and $Q_{v_4}$ into a $5^+$-path. 
This $5^+$-path together with the $5^+$-path $P_{v_{j_2}}$ implies that $\eta(K) \ge 10$ and we are done.
\end{proof}

\begin{lemma}
\label{lemma13}
Suppose that $K_c$ is a bi-star. Then, the following statements hold:
\begin{enumerate}
\parskip=0pt
\item If $K_c$ has no $0$-anchor, then $\eta(K) \ge 14$.
\item If $|(O_1 \cup T_2)\cap V(K_c)|\ge 1$ and $|(O_0\cup T_0 \cup T_1)\cap V(K_c)|\ge 2$, then $\eta(K)\ge 11$.
\item If $|(O_1 \cup T_1 \cup T_2)\cap V(K_c)| \ge 2$, then $\eta(K) \ge 8$.
	Moreover, $\eta(K) \ge 10$ unless either $v_1 \in O_1$ and $v_2 \in O_1 \cup T_1 \cup T_2$, or $v_4 \in O_1$ and $v_3 \in O_1 \cup T_1 \cup T_2$.
\item If $|T_2\cap V(K_c)| = 1$ and $|O_0\cap V(K_c)| = 1$, then $\eta(K) \ge 8$ unless either $v_2 \in T_2$ and $|\{v_1, v_3\} \cap O_0|=1$,
	or $v_3 \in T_2$ and $|\{v_2, v_4\} \cap O_0|=1$.
\end{enumerate}
\end{lemma}
\begin{proof} 
We prove the statements separately as follows.

{\em Statement 1.}
Assume that $K_c$ has no $0$-anchor. 
By Statement~2 of Lemma~\ref{lemma11}, there is an $i\in\{1,3\}$ such that neither $v_i$ nor $v_{i+1}$ is in $O_0 \cup T_0$. 
Without loss of generality, we may assume that $i=1$. 
Then, we can construct an $8^+$-path by connecting $Q_{v_1}$ and $Q_{v_2}$ with the edge $\{v_1, v_2\}$, 
and another $6^+$-path by connecting $Q_{v_3}$ and $Q_{v_4}$ with the edge $\{v_3, v_4\}$; 
the two paths imply that $\eta(K) \ge 14$ and we are done.

{\em Statement 2.} 
Assume that $|(O_1 \cup T_2)\cap V(K_c)|\ge 1$ and $|(O_0\cup T_0 \cup T_1)\cap V(K_c)|\ge 2$.
Then, by Statement~2 of Lemma~\ref{lemma11}, either $\{v_1, v_2\} \subseteq O_0 \cup T_0 \cup T_1$ or $\{v_3, v_4\} \subseteq O_0 \cup T_0 \cup T_1$.
Without loss of generality, we assume the former. 
Then, $v_3$ or $v_4$ is in $O_1 \cup T_2$.
For each $j\in\{1,3\}$, we can use the edge $\{v_j, v_{j+1}\}$ to connect $Q_{v_j}$ and $Q_{v_{j+1}}$ into a path. 
One of the paths is a $5^+$-path and the other is a $6^+$-path. 
So, $\eta(K) \ge 11$ and we are done.

{\em Statement 3.}
Assume that $|(O_1 \cup T_1 \cup T_2)\cap V(K_c)| \ge 2$.
Then, we can find two anchors $v_{j_1}, v_{j_2}$ in $O_1 \cup T_1 \cup T_2$ such that $1 \le j_1 < j_2 \le 4$.
We can construct an $8^+$-path by connecting $Q_{v_{j_1}}$ and $Q_{v_{j_2}}$ with the subpath of $K_c$ from $v_{j_1}$ to $v_{j_2}$, implying that $\eta(K) \ge 8$.
If $j_1 \in \{1, 2\}$ and $j_2 \in \{3, 4\}$, then we can construct a $5^+$-path by connecting $Q_{v_i}$ and $Q_{v_{i+1}}$ with the edge $\{v_i, v_{i+1}\}$ for each $i\in\{1,3\}$. 
Thus, $\eta(K) \ge 10$ and we are done.
It remains to consider the case where either $(j_1, j_2)=(1, 2)$ and $v_1 \in T_1 \cup T_2$, or $(j_1, j_2)=(3, 4)$ and $v_4 \in T_1 \cup T_2$.
By symmetry, we may assume the former case.
Then, we can use the path $v_2$-$v_3$-$v_4$ to connect $Q_{v_2}$ and $Q_{v_4}$ into a $6^+$-path.
This $6^+$-path together with the $6^+$-path $P_{v_1}$ implies $\eta(K) \ge 12$ and we are done.

{\em Statement 4.}
Assume that  $|T_2\cap V(K_c)| = 1$ and $|O_0\cap V(K_c)| = 1$.
Then, we can find two anchors $v_{j_1}, v_{j_2}$ such that $1 \le j_1 < j_2 \le 4$ and one of them is in $T_2$ but the other is in $O_0$.
If $j_2 \ge j_1+2$, then we can construct an $8^+$-path by connecting $Q_{v_{j_1}}$ and $Q_{v_{j_2}}$ with the subpath of $K_c$ from $v_{j_1}$ to $v_{j_2}$, 
implying that $\eta(K) \ge 8$ and we are done.
So, we can assume $j_2=j_1+1$ and thus $(j_1, j_2)=(1, 2)$, $(2, 3)$ or $(3, 4)$.
It suffices to show that if either $v_1 \in T_2$ and $v_2 \in O_0$, or $v_4 \in T_2$ and $v_3 \in O_0$, then $\eta(K) \ge 8$.
By symmetry, we may assume the former case.
Then, we can use the path $v_2$-$v_3$-$v_4$ to connect $Q_{v_2}$ and $Q_{v_4}$ into a $5^+$-path.
This $5^+$-path together with the $7^+$-path $P_{v_1}$ implies that $\eta(K) \ge 12$ and we are done.
\end{proof}

\begin{lemma}
\label{lemma14}
Suppose that $K_c$ is a $5$-path.
Then, the following statements hold:
\begin{enumerate}
\parskip=0pt
\item If $|(O_1 \cup T_1 \cup T_2)\cap V(K_c)| \ge 2$, then $\eta(K) \ge 8$.
	Moreover, $\eta(K) \ge 10$ unless either $v_1 \in O_1$ and $v_2 \in O_1 \cup T_1 \cup T_2$, or $v_5 \in O_1$ and $v_4 \in O_1 \cup T_1 \cup T_2$.
\item If $|T_2\cap V(K_c)| = 1$ and $|O_0\cap V(K_c)| = 1$, then $\eta(K) \ge 8$ unless (1) $v_3 \in T_2$ and $\{v_2, v_4\} \cap O_0 \ne \emptyset$,
	or (2) $v_2 \in T_2$ and $v_1 \in O_0$,
	or (3) $v_4 \in T_2$ and $v_5 \in O_0$.
\item If $|T_2\cap V(K_c)| = 1$ and $|O_0\cap V(K_c)| = 2$, then $\eta(K) \ge 13$ unless $v_3 \in T_2$ and $\{v_2, v_4\} \subseteq O_0$.
\end{enumerate}
\end{lemma}
\begin{proof} 
{\em Statement~1.}
First suppose that $v_3$ is not a $0$-anchor. 
Then, since $|(O_1 \cup T_1 \cup T_2)\cap V(K_c)| \ge 2$, $v_i \in O_1 \cup T_1 \cup T_2$ for some $i \in \{1, 2, 4, 5\}$.
Without loss of generality, we can assume $i\in\{1,2\}$; the other case is similar. 
We can construct two vertex-disjoint $5^+$-paths by connecting $Q_{v_1}$ and $Q_{v_2}$ with the edge $\{v_1, v_2\}$ and connecting $Q_{v_3}$ and $Q_{v_5}$ with the path $v_3$-$v_4$-$v_5$.
It follows that $\eta(K) \ge 10$ and we are done.
Next suppose that $v_3$ is a $0$-anchor. 
Then, at least two of $v_1$, $v_2$, $v_4$, and $v_5$ are in $O_1 \cup T_1 \cup T_2$.
Similarly to Statement~3 of Lemma~\ref{lemma13}, we can prove the statement. 

{\em Statement~2.} 
Assume that $|T_2\cap V(K_c)| = 1$ and $|O_0\cap V(K_c)| = 1$.
Then, we can find two anchors $v_{j_1}, v_{j_2}$ such that $1 \le j_1 < j_2 \le 5$ and one of them is in $T_2$ but the other is in $O_0$.
Similarly to Statement~4 of Lemma~\ref{lemma13}, we can assume $j_2=j_1+1$; otherwise, we are done. 
Therefore, $(j_1, j_2)=(1, 2)$, $(2, 3)$, $(3, 4)$, or $(4, 5)$.
It suffices to consider four cases:
$v_1 \in T_2$ and $v_2 \in O_0$, $v_2 \in T_2$ and $v_3 \in O_0$, $v_4 \in T_2$ and $v_3 \in O_0$, or $v_5 \in T_2$ and $v_4 \in O_0$.
We assume the case where $v_1 \in T_2$ and $v_2 \in O_0$; the other cases are similar.
Then, we can use the path $v_2$-$v_3$-$v_4$-$v_5$ to connect $Q_{v_2}$ and $Q_{v_5}$ into a $6^+$-path. 
This $6^+$-path together with the $7^+$-path $P_{v_1}$ implies that $\eta(K) \ge 13$ and we are done.

{\em Statement~3.}
Assume that $|T_2\cap V(K_c)| = 1$ and $|O_0\cap V(K_c)| = 2$.
By Statement~3 of Lemma~\ref{lemma11}, we can assume $\{v_1, v_2\} \subseteq O_0$ or $\{v_2, v_4\} \subseteq O_0$ or $\{v_4, v_5\} \subseteq O_0$ (but $v_3 \notin T_2$).
We prove the first case while the other two can be proven the same.
Let $v_j$ be the anchor in $T_2$ and thus $j \in \{3, 4, 5\}$.
We can use the edge $\{v_1,v_2\}$ to connect $Q_{v_1}$ and $Q_{v_2}$ into a $6^+$-path. 
This $6^+$-path together with the $7^+$-path $P_{v_j}$ implies that $\eta(K) \ge 13$ and we are done.
\end{proof}

\begin{lemma}
\label{lemma15}
Suppose that $K_c$ is a $5$-path.
Then, the following statements hold:
\begin{enumerate}
\parskip=0pt
\item If $K_c$ has at most one $0$-anchor, then $\eta(K) \ge 14$.
\item If $K_c$ has no $0$-anchor, then $\eta(K) \ge 16$.
	Moreover, if $v_1$, $v_3$, or $v_5$ is a $2$-anchor, then $\eta(K) \ge 17$.
\item If $|(O_1 \cup T_1 \cup T_2)\cap V(K_c)| = 5$, then $\eta(K) \ge 17$.
	Moreover, if $v_1$, $v_3$, or $v_5$ is a $2$-anchor, then $\eta(K) \ge 22$.
\end{enumerate}
\end{lemma}
\begin{proof}
{\em Statement~1.} 
Assume that $K_c$ has at most one $0$-anchor.
Then, we can always find four anchors $v_{j_1}, v_{j_2}, v_{j_3}, v_{j_4}$ such that $1 \le j_1 < j_2 < j_3 < j_4 \le 5$ and none of them is a $0$-anchor.
By Statement~3 of Lemma~\ref{lemma11}, at least two of $v_{j_1}, v_{j_2}, v_{j_3}, v_{j_4}$ are both in $O_1 \cup T_2$.
If $\{v_{j_1}, v_{j_2}\} \subseteq O_1 \cup T_2$ or $\{v_{j_3}, v_{j_4}\} \subseteq O_1 \cup T_2$,
then for each $i=\{1, 3\}$, we can use the subpath of $K_c$ from $v_{j_i}$ to $v_{j_{i+1}}$ to connect $Q_{v_{j_i}}$ and $Q_{v_{j_{i+1}}}$ into a path.
One of these paths is an $8^+$-path and the other is a $6^+$-path.
Thus, $\eta(K) \ge 14$ and we are done.
Lastly, we can assume exactly one of $v_{j_1}$ and $v_{j_2}$ and one of $v_{j_3}$ and $v_{j_4}$ are in $O_1 \cup T_2$.
Then, we can construct two $7^+$-paths similarly and obtain $\eta(K) \ge 14$ again.

{\em Statement 2.}  
Assume that $K_c$ has no 0-anchor. 
By Statement~3 of Lemma~\ref{lemma11}, there is an $i\in\{2,4\}$ such that
at most one vertex in $V(K_c)\setminus\{v_i\}$ is in $O_0 \cup T_0 \cup T_1$ and the other three vertices are in $O_1 \cup T_2$.
We assume the case where $i=4$; the other case where $i=2$ is similar.
Then, we can use the edge $\{v_1, v_2\}$ to connect $Q_{v_1}$ and $Q_{v_2}$ into a path, 
and use the path $v_3$-$v_4$-$v_5$ to connect $Q_{v_3}$ and $Q_{v_5}$ into a path.
Since at most one of $v_1$, $v_2$, $v_3$, and $v_5$ is in $O_0 \cup T_0 \cup T_1$, either the two paths are both $8^+$-paths,
or one of them is a $7^+$-path and the other is a $9^+$-path.
In both cases, $\eta(K) \ge 16$ and we are done. 

Next we assume that $v_1$, $v_3$, or $v_5$ is a $2$-anchor.
We assume that $v_1$ is a 2-anchor; the other cases are similar. 
Then, we can use the edge $\{v_j, v_{j+1}\}$ to connect $Q_{v_j}$ and $Q_{v_{j+1}}$ into a $6^+$-path for each $j \in \{2, 4\}$.
The two $6^+$-paths together with the $5^+$-path $P_{v_1}$ imply that $\eta(K) \ge 17$ and we are done. 

{\em Statement~3.}
Assume that $|(O_1 \cup T_1 \cup T_2)\cap V(K_c)| = 5$. 
Then, we can construct an $8^+$-path by connecting $Q_{v_1}$ and $Q_{v_2}$ with the edge $\{v_1, v_2\}$ and
another $9^+$-path by connecting $Q_{v_3}$ and $Q_{v_5}$ with the 3-path $v_3$-$v_4$-$v_5$. 
These two paths imply that $\eta(K) \ge 17$.  

Next we assume that $v_1$, $v_3$, or $v_5$ is a $2$-anchor.
We assume that $v_1$ is a 2-anchor; the other cases are similar. 
Then, we can use the edge $\{v_j, v_{j+1}\}$ to connect $Q_{v_j}$ and $Q_{v_{j+1}}$ into an $8^+$-path for each $j \in \{2, 4\}$.
These two $8^+$-paths together with the $6^+$-path $P_{v_1}$ imply that $\eta(K) \ge 22$ and we are done. 
\end{proof}

\begin{lemma}
\label{lemma16}
Suppose that $K_c$ is a $5$-path without any $2$-anchor.
Then, $\frac {s(K)}{\eta(K)} < \frac {15}8$ and hence $K$ is not critical.
\end{lemma}
\begin{proof} 
Since $K_c$ has no $2$-anchor, $s(K)\le 24$. 
Moreover, $\eta(K) \ge 5$ because $K_c$ itself is a $5$-path. 
So, if $s(K) \le 8$, then $\frac {s(K)}{\eta(K)} \le \frac 85$. 
Furthermore, if $K_c$ has at most one $0$-anchor or $|O_1\cap V(K_c)| \ge 3$,
then by Statement~1 of Lemma~\ref{lemma15} and Statement~3 of Lemma~\ref{lemma12}, $\eta(K) \ge 13$ and in turn $\frac {s(K)}{\eta(K)} \le \frac {24}{13} < \frac {15}8$.
Thus, we may assume that $s(K) > 8$, $K_c$ has at least two $0$-anchors, and $|O_1\cap V(K_c)| \le 2$.
It follows that $s(K) \le 14$ and in turn $s(K) \in \{ 10, 12, 14 \}$.

Case 1: $s(K) = 14$.
In this case, $|O_1\cap V(K_c)| = 2$.
By Statement~1 of Lemma~\ref{lemma14}, $\eta(K) \ge 8$ and in turn $\frac {s(K)}{\eta(K)} \le \frac 74< \frac {15}8$.

Case 2: $s(K) = 12$.
In this case, either $|O_1\cap V(K_c)| = 2$, or $|O_1\cap V(K_c)|=1$ and $|O_0\cap V(K_c)|=2$.
By Statement~1 of Lemma~\ref{lemma14} and Statement~4 of Lemma~\ref{lemma12}, $\eta(K) \ge 7$ and in turn $\frac {s(K)}{\eta(K)} \le \frac {12}7< \frac {15}8$.

Case 3: $s(K) = 10$.
In this case, at least two vertices of $K_c$ are not $0$-anchors. 
By Statement~4 of Lemma~\ref{lemma12}, $\frac {s(K)}{\eta(K)} \le \frac 53$ and we are done. 
\end{proof}

\begin{lemma}
\label{lemma17}
Suppose that $K_c$ is not a $5$-path and has no $2$-anchor.
Then, $\frac {s(K)}{\eta(K)} < \frac {15}8$ and hence $K$ is not critical.
\end{lemma}
\begin{proof}
By Statement~2 of Lemma~\ref{lemma09}, $K_c$ is an edge, star, or bi-star.
If $K_c$ is an edge or star, then by Statement~1 of Lemma~\ref{lemma11}, either $\{v_1, v_2\}\subseteq O_1$, or one of $v_1$ and $v_2$ is in $O_1$ and the other is a $0$-anchor.
In the former case, $s(K) = 10$, and $\eta(K)\ge 8$ because we can construct an $8^+$-path by connecting $Q_{v_1}$ and $Q_{v_2}$ with the edge $\{v_1, v_2\}$. 
Similarly, in the latter case, $s(K) = 6$ and $\eta(K)\ge 5$. 
So, in both cases, $\frac {s(K)}{\eta(K)} \le \frac 54$ and thus $K$ is not critical. 
Thus, we may assume that $K_c$ is a bi-star. 
Then, $6 \le s(K) \le 20$. 
Moreover, since $K$ is a composite component of $H+C$, $K$ always contains a $5^+$-path.
Hence, if $s(K) \le 8$, then $\frac {s(K)}{\eta(K)} \le \frac 85$. 
Consequently, we may assume that $s(K) > 8$.
Furthermore, if $|O_1\cap V(K_c)| \ge 3$ or $K_c$ has no $0$-anchor,
then by Statement~3 of Lemma~\ref{lemma12} and Statement~1 of Lemma~\ref{lemma13}, $\eta(K) \ge 13$ and in turn $\frac {s(K)}{\eta(K)} \le \frac {20}{13}$.
Therefore, we may assume that $|O_1\cap V(K_c)| \le 2$ and $K_c$ has at least one $0$-anchor. 
It follows that $s(K) \le 14$ and in turn $s(K) \in \{ 10, 12, 14 \}$.

If $s(K) = 14$, then $|O_1 \cap V(K_c)| = 2$ and in turn $\eta(K) \ge 8$ by Statement~3 of Lemma~\ref{lemma13}, implying that $\frac {s(K)}{\eta(K)} \le \frac 74$.  
If $s(K) = 12$, then either $|O_1\cap V(K_c)| = 2$, or $|O_1\cap V(K_c)| = 1$ and $|O_0\cap V(K_c)|= 2$; 
by Statements~2 and~3 in Lemma~\ref{lemma13}, $\eta(K) \ge 8$ and in turn $\frac {s(K)}{\eta(K)} \le \frac 32$. 
If $s(K) = 10$, then $|O_0\cap V(K_c)| \le 2$ by Statement~2 of Lemma~\ref{lemma11}, and in turn $|O_0\cap V(K_c)| = |O_1\cap V(K_c)|= 1$;
by Statement~4 of Lemma~\ref{lemma12}, $\eta(K) \ge 6$ and in turn $\frac {s(K)}{\eta(K)} \le \frac 53$.
\end{proof}

\begin{notation}
\label{nota06}
We define the notation $\preceq$ as follows.
Suppose that $K$ is a connected component of $H+C$, and $a, b, c, d$ are all positive integers.
\begin{itemize}
\parskip=0pt
\item
	For a rational fraction $\frac ab$, we write $\frac{s(K)}{\eta(K)} \preceq \frac ab$, whenever $s(K) \le a$ and $\eta(K) \ge b$.
	
	({\em Comment:} If $\frac {s(K)}{\eta(K)}\preceq \frac ab$, then $\frac {s(K)}{\eta(K)}\le \frac ab$, but not vice versa.)
\item
	For a set $\mcF$ of rational fractions, we write $\frac{s(K)}{\eta(K)} \preceq \mcF$ if $\frac{s(K)}{\eta(K)} \preceq \frac ab$ for some $\frac ab \in \mcF$.
\item
	For a rational fraction $\frac ab$ and a set $\mcF$ of rational fractions, we write $\mcF \preceq \frac ab$ if $\frac cd \le \frac ab$ for all $\frac cd \in \mcF$.
\end{itemize}
\end{notation}

The reason why we need Notation~\ref{nota06} is that we not only focus on the ratio of $\frac {s(K)}{\eta(K)}$ but sometimes also on the exact values of $s(K)$ and $\eta(K)$. 
This notation will help us distinguish the structures of critical components and be useful in the performance analysis of the algorithm.

\begin{lemma}
\label{lemma18}
Suppose that $K_c$ has at least three $2$-anchors. 
Then, the following statements hold and hence $K$ is not critical.
\begin{enumerate}
\parskip=0pt
\item If $K_c$ has five $2$-anchors, then $\frac {s(K)}{\eta(K)} \preceq \left\{\frac {40}{25}, \frac {44}{33}\right\}$.
\item If $K_c$ has exactly four $2$-anchors, then $\frac {s(K)}{\eta(K)} \preceq \left\{\frac {32}{20}, \frac {38}{24}, \frac {40}{28}\right\}$.
\item If $K_c$ has exactly three $2$-anchors, then $\frac {s(K)}{\eta(K)} \preceq \frac {2h+24}{h+15}$ for some $h = 0, 1, \ldots, 6$.
\end{enumerate}
\end{lemma}
\begin{proof} 
{\em Statement~1.} 
Assume that $K_c$ has five $2$-anchors. 
Then, $K_c$ is a $5$-path, and $\eta(K) \ge 5\cdot 5 = 25$ because $P_{v_i}$ is a $5^+$-path for each $v_i \in K_c$. 
If $s(K)\le 40$, then $\frac {s(K)}{\eta(K)} \preceq \frac {40}{25}$. 
So, we may assume that $s(K)\ge 41$. 
Then, since $s(K) \le 44$ is even, $s(K)\in \{42, 44\}$ and in turn at least four vertices of $K_c$ are in $T_2$. 
Consequently, $P_{v_1}$, \ldots, $P_{v_5}$ are $5^+$-paths and at least four of them are indeed $7^+$-paths. 
Therefore, $\eta(K)\ge 5 + 4\cdot 7 =  33$ and in turn $\frac {s(K)}{\eta(K)} \preceq \frac {44}{33}$.

{\em Statement~2.} 
Assume that $K_c$ has exactly four $2$-anchors. 
Then, $\eta(K) \ge 20$ because $P_{v_i}$ is a $5^+$-path for each $2$-anchor $v_i$ in $K_c$. 
Moreover, if $s(K) \le 32$, then $\frac {s(K)}{\eta(K)} \preceq \frac {32}{20}$.
So, we may assume $s(K)\ge 34$ because $s(K)$ is even. 
Then, either at least two anchors of $K_c$ are in $T_2$, or one anchor of $K_c$ is in $T_2$ but the other three are in $T_1$;
it follows that either way $\eta(K) \ge 24$.
Now, if $s(K) \le 38$, then $\frac {s(K)}{\eta(K)} \preceq \frac {38}{24}$. 
Thus, we may assume $s(K) = 40$. 
Then, $K_c$ is a $5$-path, $|T_2\cap V(K_c)|=4$, and four paths among $P_{v_1}$, \ldots, $P_{v_5}$ are $7^+$-paths. 
Consequently, $\eta(K) \ge 4\cdot 7 = 28$ and in turn $\frac {s(K)}{\eta(K)} \preceq \frac {40}{28}$.

{\em Statement~3.}
Assume that $K_c$ has exactly three $2$-anchors. 
Then, $K_c$ is a bi-star or $5$-path. 
Moreover, $s(K) \le \sum_{i=0}^2(4+2i)|T_i\cap V(K_c)| + 2\cdot 4 + 4$. 
Furthermore, since $P_{v_i}$ is a $5^+$-path for each  $2$-anchor $v_i$, $\eta(K)\ge \sum_{i=0}^2 (5+i)|T_i\cap V(K_c)|$.
Now, since $\sum_{i=0}^2 |T_i\cap V(K_c)|=3$, we have 
\[
\frac {s(K)}{\eta(K)} \preceq \frac {\sum_{i=0}^2(4+2i)|T_i\cap V(K_c)| +12}{\sum_{i=0}^2 (5+i)|T_i\cap V(K_c)|} 
= \frac {2|T_1\cap V(K_c)|+4|T_2\cap V(K_c)|+24}{|T_1\cap V(K_c)|+2|T_2\cap V(K_c)|+15} 
\]
Note that $|T_1\cap V(K_c)|+2|T_2\cap V(K_c)|\in \{0,\ldots,6\}$. 
Consequently, the lemma holds.
\end{proof}

\begin{lemma}
\label{lemma19}
Suppose that $K_c$ has exactly two $2$-anchors.
Then, the following statements hold:
\begin{enumerate}
\parskip=0pt
\item If $K_c$ is an edge or star, then $\frac {s(K)}{\eta(K)} \preceq \frac {18}{14}$ and hence $K$ is not critical.
\item If $K_c$ is a bi-star, then $\frac {s(K)}{\eta(K)} \preceq \frac {2h+16}{h+10}$ for some $h = 0, 1, \ldots, 6$, and hence $K$ is not critical.
\end{enumerate}
\end{lemma}
\begin{proof}
{\em Statement~1.}
Assume that $K_c$ is an edge or star. 
Then, each vertex of $K_c$ is in $T_2$ by Statement~1 of Lemma~\ref{lemma11} and in turn $s(K)=18$. 
Moreover, $\eta(K) \ge 14$ because both $P_{v_1}$ and $P_{v_2}$ are $7^+$-paths. 
Therefore, $\frac {s(K)}{\eta(K)} \preceq \frac {18}{14}$ and hence $K$ is not critical.

{\em Statement~2.}
Assume that $K_c$ is a bi-star. 
Then, $s(K) \le 2\cdot 8 + 2 \cdot 6 = 28$ because $K_c$ has exactly two $2$-anchors. 
Moreover, if $s(K) = 28$ (i.e., $h = 6$), then $|T_2\cap V(K_c)|=|O_1\cap V(K_c)|=2$ and in turn Statement~1 of Lemma~\ref{lemma12} implies that $\eta(K) \ge 16$ and we are done.
Similarly, if $s(K) = 26$ (i.e., $h = 5$), then either $|T_2\cap V(K_c)| = |T_1\cap V(K_c)| = 1$ and $|O_1\cap V(K_c)|=2$, 
or $|T_2\cap V(K_c)| = 2$ and $|O_1\cap V(K_c)|=|O_0\cap V(K_c)|=1$; 
in both cases, Statements~1 and~2 in Lemma~\ref{lemma12} imply that $\eta(K) \ge 15$ and we are done.

Now, we may assume that $s(K) \le 24$. 
If $K_c$ has no $0$-anchor, then $\eta(K) \ge 14$ by Statement~1 of Lemma~\ref{lemma13}, and in turn $\frac {s(K)}{\eta(K)} \preceq \frac {24}{14}$ and we are done.
So, we may further assume that $K_c$ has at least one $0$-anchor.
Then, $s(K) \le \sum_{i=0}^2 (4+2i)|T_i\cap V(K_c)| + 8$. 
Furthermore, since $P_{v_i}$ is a $5^+$-path for each $2$-anchor $v_i$, $\eta(K)\ge \sum_{i=0}^2 (5+i)|T_i\cap V(K_c)|$.
Now, since $\sum_{i=0}^2 |T_i\cap V(K_c)|=2$, we have 
\[
\frac {s(K)}{\eta(K)} \preceq \frac {\sum_{i=0}^2 (4+2i)|T_i\cap V(K_c)| + 8}{\sum_{i=0}^2 (5+i)|T_i\cap V(K_c)|}
= \frac {2|T_1\cap V(K_c)| + 4|T_2\cap V(K_c)| + 16}{|T_1\cap V(K_c)| + 2|T_2\cap V(K_c)| + 10}
\]
Note that $|T_1\cap V(K_c)|+2|T_2\cap V(K_c)|\in \{0,\ldots,4\}$. 
Consequently, the lemma holds.
\end{proof}

\begin{lemma}
\label{lemma20}
Suppose that $K_c$ has exactly two $2$-anchors and $K_c$ is a $5$-path. 
Then, $\{v_2, v_4\} \subseteq T_2$ if $K$ is critical. 
Moreover, $\frac {s(K)}{\eta(K)} = \frac {30}{16}$ if $s(K)=30$, 
while $\frac {s(K)}{\eta(K)} = \frac {32}{17}$ if $s(K)=32$.
Furthermore, if $K$ is not critical, then $\frac {s(K)}{\eta(K)} \preceq \frac {2h+16}{h+10}$ for some $h = 0, 1, \ldots, 8$.
\end{lemma}
\begin{proof}
Since $K_c$ has exactly two $2$-anchors, $s(K) \le 32$.

Case 1: $s(K) = 32$.
In this case, $|T_2\cap V(K_c)| = 2$ and $|O_1\cap V(K_c)| = 3$. 
If $v_1$, $v_3$, or $v_5$ is in $T_2$, then $\eta(K) \ge 22$ by Statement~3 of Lemma~\ref{lemma15}, and hence $K$ is not critical. 
Otherwise, $\{v_2, v_4\} \subseteq T_2$, and $\eta(K) \ge 17$ again by Statement~3 of Lemma~\ref{lemma15}.
Hence, $\frac {s(K)}{\eta(K)} \le \frac {32}{17}$ in any case. 
In summary, if $K$ is critical, then $\{v_2, v_4\} \subseteq T_2$ and $\frac {s(K)}{\eta(K)} = \frac {32}{17}$.

Case 2: $s(K) = 30$.
In this case, either $|T_2\cap V(K_c)| = |T_1\cap V(K_c)| = 1$ and $|O_1\cap V(K_c)| = 3$, or $|T_2\cap V(K_c)|=|O_1\cap V(K_c)|=2$ and $|O_0\cap V(K_c)|=1$. 
In the former case, by Statement~3 of Lemma~\ref{lemma15}, $\eta(K) \ge 17$ and in turn $\frac {s(K)}{\eta(K)} \le \frac {30}{17}$, implying that $K$ is not critical. 
In the latter case, if $v_1$, $v_3$, or $v_5$ is a $2$-anchor, then $\eta(K) \ge 17$ by Statement~2 of Lemma~\ref{lemma15}; 
otherwise, both $v_2$ and $v_4$ are in $T_2$, and $\eta(K) \ge 16$ by Statement~2 of Lemma~\ref{lemma15}, implying that $\frac {s(K)}{\eta(K)} \le \frac {30}{16}$.
In summary, if $K$ is critical, then $\{v_2, v_4\} \subseteq T_2$ and $\frac {s(K)}{\eta(K)} = \frac {30}{16}$.

Case 3: $s(K) \le 28$.
In this case, if $K_c$ has no $0$-anchor, then by Statement~2 of Lemma~\ref{lemma15}, $\eta(K) \ge 16$ and we are done. 
So, we may assume that $K_c$ has at least one $0$-anchor. 
In the first case where $s(K) = 28$, $|T_2\cap V(K_c)| = |O_1\cap V(K_c)| = 2$ because $K_c$ has at least one $0$-anchor and exactly two $2$-anchors; 
hence, by Statement~1 of Lemma~\ref{lemma12}, $\eta(K) \ge 16$ and we are done.
In the second case where $s(K) = 26$, either $|T_2\cap V(K_c)| = |T_1\cap V(K_c)| = 1$ and $|O_1\cap V(K_c)| = 2$, 
or $|T_2\cap V(K_c)| = 2$ and $|O_1\cap V(K_c)| = |O_0\cap V(K_c)| = 1$; 
hence, by Statements~1 and~2 in Lemma~\ref{lemma12}, $\eta(K) \ge 15$ and we are done. 
Now, we assume the last case where $s(K) \le 24$.
If $K_c$ has at most one $0$-anchor, then by Statement~1 of Lemma~\ref{lemma15}, $\eta(K) \ge 14$ and we are done. 
Thus, we may further assume that $K_c$ has at least two $0$-anchors. 
Then, $s(K) \le \sum_{i=0}^2 (4+2i)|T_i\cap V(K_c)| + 8$. 
Furthermore, since $P_{v_i}$ is a $5^+$-path for each  $2$-anchor $v_i$, $\eta(K)\ge \sum_{i=0}^2 (5+i)|T_i\cap V(K_c)|$. 
Therefore, the same discussion as at the end of the proof of Lemma~\ref{lemma19} applies here too.
\end{proof}

\begin{lemma}
\label{lemma21}
Assume that $K_c$ has exactly one $2$-anchor.
Then, the following statements hold:
\begin{enumerate}
\parskip=0pt
\item If $K_c$ is an edge or star, then $\frac {s(K)}{\eta(K)} \preceq \{\frac {12}7, \frac {14}8\}$ and hence $K$ is not critical.
\item Suppose that $K_c$ is a bi-star. Then, exactly one of $v_2$ and $v_3$ is in $T_2$ if $K$ is critical. 
	Moreover, $\frac {s(K)}{\eta(K)}=\frac {14}7$ if $s(K) = 14$; $\frac {s(K)}{\eta(K)}=\frac {16}8$ if $s(K) = 16$; 
	$\frac {s(K)}{\eta(K)} \in \{\frac {18}8, \frac {18}9\}$ if $s(K) = 18$.
	Furthermore, if $K$ is not critical, then $\frac {s(K)}{\eta(K)} \preceq \frac {2h+8}{h+5}$ for some $h = 0, 1, \ldots, 8$.

	(cf. Figure \ref{fig03} for all possible structures of critical components.)
\end{enumerate}
\end{lemma}
\begin{proof}
{\em Statement~1.}
Assume that $K_c$ is an edge or star.
Without loss of generality, we may assume that $v_1$ is a $2$-anchor.
By Statement~1 of Lemma~\ref{lemma11}, $v_1 \in T_2$ and $v_2 \notin O_0$.  
If $v_2$ is a $0$-anchor, then $s(K)=12$, and $\eta(K)\ge 7$ because $P_{v_1}$ is a $7^+$-path.
Otherwise, $v_2\in O_1$, $s(K)=14$, and $\eta(K)\ge 8$ because we can connect $Q_{v_1}$ and $Q_{v_2}$ into an $8^+$-path with the edge $\{v_1, v_2\}$. 
In summary, $\frac {s(K)}{\eta(K)}\preceq \{\frac {12}7, \frac {14}8\}$, and hence $K$ is not critical.

{\em Statement~2.}
Since $K_c$ is a bi-star and has exactly one $2$-anchor, $s(K) \le 24$.
Moreover, $\eta(K) \ge 5$ because $P_{v_i}$ is a $5^+$-path for the $2$-anchor $v_i$. 
Consequently, if $s(K) \le 8$, then we are done. 
So, we may assume $10 \le s(K) \le 24$.
If $K_c$ has no $0$-anchor, then by Statement~1 of Lemma~\ref{lemma13}, $\eta(K) \ge 14$ and in turn $\frac {s(K)}{\eta(K)} \preceq \frac {24}{14}$. 
Hence, we may assume that $K_c$ has at least one $0$-anchor.
It follows that $s(K) \le 20$. 
In particular, if $s(K)=20$, then $|T_2\cap V(K_c)|=1$ and $|O_1\cap V(K_c)|=2$; by Statement~3 of Lemma~\ref{lemma12}, $\eta(K) \ge 13$ and we are done. 
Thus, it suffices to distinguish five cases below, where $i$ denotes the unique integer in $\{0,1,2\}$ with $|T_i\cap V(K_c)|=1$.

Case 1: $s(K) = 18$.
In this case, either $|T_1\cap V(K_c)| = 1$ and $|O_1\cap V(K_c)| = 2$, or $|T_2\cap V(K_c)| = |O_1\cap V(K_c)| = |O_0\cap V(K_c)| = 1$.
In the former case, by Statement~3 of Lemma~\ref{lemma12}, $\eta(K) \ge 13$ and hence $K$ is not critical.
In the latter case, since $|(T_2 \cup O_1)\cap V(K_c)| = 2$, Statement~3 of Lemma~\ref{lemma13} implies that $\eta(K) \ge 8$
and in turn $\frac {s(K)}{\eta(K)} \preceq \frac {18}8$, and $\eta(K) \ge 10$ unless either $v_1 \in O_1$ and $v_2 \in T_2$, or $v_4 \in O_1$ and $v_3 \in T_2$.
Therefore, if $K$ is critical, then $\frac {s(K)}{\eta(K)} \preceq \{ \frac {18}8, \frac {18}9 \}$
and either $v_1 \in O_1$ and $v_2 \in T_2$, or $v_4 \in O_1$ and $v_3 \in T_2$.

Case 2: $s(K)=16$.
If $i=0$, then $|O_1\cap V(K_c)| = 2$ and in turn Statement~5 of Lemma~\ref{lemma12} implies that $\eta(K) \ge 9$ and hence $K$ is not critical.
If $i=1$, then $|O_1\cap V(K_c)| = |O_0\cap V(K_c)| = 1$ and in turn Statement~2 of Lemma~\ref{lemma13} implies that $\eta(K) \ge 11$ and hence $K$ is not critical.
So, we may assume that $i=2$. 
Then, either $|O_0\cap V(K_c)| = 2$ or $|O_1\cap V(K_c)| = 1$. 
In the former case, Statement~2 of Lemma~\ref{lemma13} implies that $\eta(K) \ge 11$ and hence $K$ is not critical.
In the latter case, Statement~3 of Lemma~\ref{lemma13} implies that if $K$ is critical,
then $\frac {s(K)}{\eta(K)} = \frac {16}8$ and either $v_1 \in O_1$ and $v_2 \in T_2$, or $v_4 \in O_1$ and $v_3 \in T_2$.

Case 3: $s(K)=14$.
If $i=0$, then $|O_1\cap V(K_c)| = |O_0\cap V(K_c)| = 1$ and in turn Statement~2 of Lemma~\ref{lemma13} implies that $\eta(K) \ge 11$ and we are done.
If $i = 1$, then Statement~2 of Lemma~\ref{lemma11} implies that $|(O_0 \cup T_0 \cup T_1)\cap V(K_c)| \le 2$ and hence $|O_1\cap V(K_c)| = 1$, 
and in turn Statement~3 of Lemma~\ref{lemma13} implies that $\eta(K) \ge 8$ and we are done.
So, we may assume that $i=2$.
Then, $|O_0\cap V(K_c)| = 1$ and $\eta(K) \ge 7$ because $P_{v_j}$ is a $7^+$-path for the unique $2$-anchor $v_j$ in $K_c$.
Thus, $\frac {s(K)}{\eta(K)} \le \frac {14}7$.
Clearly, $K$ is critical, then $\frac {s(K)}{\eta(K)} = \frac {14}7$. 
Moreover, if $K$ is critical, then by Statement~4 of Lemma~\ref{lemma13},
either $v_2 \in T_2$ and exactly one of $v_1$ and $v_3$ is in $O_0$, or $v_3 \in T_2$ and exactly one of $v_2$ and $v_4$ is in $O_0$.

Case 4: $s(K) = 12$.
If $i=2$, then $\eta(K) \ge 7$ because $P_{v_j}$ is a $7^+$-path for the unique $2$-anchor $v_j$ in $K_c$, and hence $K$ is not critical. 
So, we may assume that $i=0$ or~$1$. 
It follows from Statement~2 of Lemma~\ref{lemma11} that either $|T_1 \cap V(K_c)| = |O_0 \cap V(K_c)| = 1$ or $|T_0 \cap V(K_c)| = |O_1 \cap V(K_c)| = 1$. 
By Statement~4 of Lemma~\ref{lemma12}, $\eta(K) \ge 7$ and hence $K$ is not critical.

Case 5: $s(K)=10$.
In this case, $i=0$ or~$1$. 
If $i=1$, then $\eta(K) \ge 6$ because $P_{v_j}$ is a $6^+$-path for the unique $2$-anchor $v_j$ in $K_c$, implying that $K$ is not critical. 
So, we may assume that $i=0$.
Then, $|T_0\cap V(K_c)| = 1$ and in turn $|O_0\cap V(K_c)| = 1$.
By Statement~4 of Lemma~\ref{lemma12}, $\eta(K) \ge 6$ and $K$ is not critical.
\end{proof}

\begin{lemma}
\label{lemma22}
Suppose that $K_c$ has exactly one $2$-anchor and is a $5$-path. 
Then, exactly one of $v_2$, $v_3$, and $v_4$ is in $T_2$ if $K$ is critical. 
Moreover, $\frac {s(K)}{\eta(K)}=\frac {14}7$ if $s(K) = 14$; 
$\frac {s(K)}{\eta(K)}=\frac {16}8$ or $\frac {16}7$ if $s(K) = 16$; 
$\frac {s(K)}{\eta(K)} = \frac {18}9$ if $s(K) = 18$. 
Furthermore, if $K$ is not critical, then $\frac {s(K)}{\eta(K)} \preceq \frac {2h+8}{h+5}$ for some $h = 0, 1, \ldots, 10$.

(cf. Figure \ref{fig03} for all possible structures of critical components.)
\end{lemma}
\begin{proof}
Clearly, $\eta(K)\ge 5$ because $K_c$ is a $5$-path. 
So, if $s(K)\le 8$, then $\frac {s(K)}{\eta(K)}\le \frac 85$ and hence $K$ is not critical.
Thus, we may assume that $s(K) \ge 10$. 
Indeed, $s(K) \le 28$ because $K$ has exactly one $2$-anchor. 
Moreover, if $K_c$ has no $0$-anchor, then Statement~2 of Lemma~\ref{lemma15} implies that $\eta(K) \ge 16$ and we are done.
Hence, we may assume that $K_c$ has at least one $0$-anchor and in turn $s(K) \le 24$.
If $K_c$ has at most one $0$-anchor, then Statement~1 of Lemma~\ref{lemma15} implies that $\eta(K) \ge 14$ and in turn $\frac {s(K)}{\eta(K)} \le \frac {24}{14}$ and we are done.
Therefore, we may further assume that $K_c$ has at least two $0$-anchors and in turn $10 \le s(K) \le 20$.
Indeed, if $s(K)=20$, then $|T_2\cap V(K_c)| = 1$ and $|O_1\cap V(K_c)| = 2$, and hence Statement~3 of Lemma~\ref{lemma12} implies that $\eta(K) \ge 13$ and we are done. 
Thus, it suffices to distinguish five cases below, where $i$ denotes the unique integer in $\{0,1,2\}$ with $|T_i\cap V(K_c)|=1$.

Case 1: $s(K) = 18$.
In this case, either $|T_1\cap V(K_c)| = 1$ and $|O_1\cap V(K_c)| = 2$, or $|T_2\cap V(K_c)| = |O_1\cap V(K_c)| = 1$ and $|O_0\cap V(K_c)| = 1$.
In the former case, Statement~3 of Lemma~\ref{lemma12} implies that $\eta(K) \ge 13$ and we are done. 
So, we may assume the latter case. 
Then, by Statement~1 of Lemma~\ref{lemma14}, $\eta(K) \ge 10$ and we are done unless either $v_1 \in O_1$ and $v_2 \in T_2$, or $v_5 \in O_1$ and $v_4 \in T_2$. 
Thus, by symmetry, we may assume that $v_1 \in O_1$ and $v_2 \in T_2$. 
Now, if $v_3$ or $v_5$ is in $O_0$, then we can construct an $8^+$-path by connecting $Q_{v_1}$ and $Q_{v_2}$ with the edge $\{v_1, v_2\}$,
and another $5^+$-path by connecting $Q_{v_3}$ and $Q_{v_5}$ with the 3-path $v_3$-$v_4$-$v_5$, implying that $\eta(K) \ge 13$ and we are done.
Hence, we may further assume that $v_4 \in O_0$. Then, we can construct a $9^+$-path by connecting $Q_{v_1}$ and $Q_{v_4}$ with the $4$-path $v_1$-$v_2$-$v_3$-$v_4$. 
Therefore, $\eta(K) \ge 9$. 
Moreover, if $K$ is critical, then $\frac {s(K)}{\eta(K)} = \frac {18}9$.

Case 2: $s(K) = 16$.
In this case, $|O_1\cap V(K_c)| = 2$ if $i=0$, while $|O_1\cap V(K_c)| = |O_0\cap V(K_c)| = 1$ if $i = 1$. 
So, if $i=0$ or~$1$, then Statement~5 of Lemma~\ref{lemma12} implies that $\eta(K) \ge 9$ and we are done.
Thus, we may assume $i=2$. 
Then, either $|O_1\cap V(K_c)| = 1$ or $|O_0\cap V(K_c)| = 2$.
In the former case, Statement~1 of Lemma~\ref{lemma14} implies that $\eta(K) \ge 8$ and in turn $\frac {s(K)}{\eta(K)} \le \frac {16}8$; 
moreover, if $K$ is critical, then $\frac {s(K)}{\eta(K)} = \frac {16}8$ and either $v_1 \in O_1$ and $v_2 \in T_2$, or $v_5 \in O_1$ and $v_4 \in T_2$.
In the latter case, $\eta(K) \ge 7$ and in turn $\frac {s(K)}{\eta(K)} \le \frac {16}7$ because $P_{v_j}$ is a $7^+$-path for the unique $2$-anchor $v_j$ in $K_c$; 
moreover, if $K$ is critical, then $\frac {s(K)}{\eta(K)} = \frac {16}8$ or $\frac {16}7$ and Statement~3 of Lemma~\ref{lemma14} implies that $\{v_2, v_4\} \subseteq O_0$ and $v_3 \in T_2$.

Case 3: $s(K) = 14$.
In this case, if $i = 0$, then $|O_1\cap V(K_c)| = |O_0\cap V(K_c)| = 1$ and in turn Statement~5 of Lemma~\ref{lemma12} implies that $\eta(K) \ge 9$ and hence $K$ is not critical.
Similarly, if $i = 1$, then either $|O_1\cap V(K_c)| = 1$ or $|O_0\cap V(K_c)| = 2$. In the former case, $\eta(K) \ge 8$ by Statement~1 of Lemma~\ref{lemma14}. 
In the latter case, $\eta(K) \ge 9$ by Statement~5 of Lemma~\ref{lemma12}. 
So, in both cases, $K$ is not critical. 
Thus, we may assume that $i=2$. 
Then, $|O_0\cap V(K_c)| = 1$, and $\eta(K) \ge 7$ because $P_{v_j}$ is a $7^+$-path for the unique $2$-anchor $v_j$ in $K_c$. 
Hence, if $K$ is critical, then $\frac {s(K)}{\eta(K)} = \frac {14}7$. 
Moreover, by Statement~2 of Lemma~\ref{lemma14},
if $K$ is critical, then (1) $v_3 \in T_2$ and $\{v_2, v_4\} \cap O_0 \ne \emptyset$, or (2) $v_2 \in T_2$ and $v_1 \in O_0$, or (3) $v_4 \in T_2$ and $v_5 \in O_0$. 

Case 4: $s(K) = 12$.
In this case, if $i = 2$, then $\eta(K) \ge 7$ because $P_{v_j}$ is a $7^+$-path for the unique $2$-anchor $v_j$ in $K_c$. 
Similarly, if $i=1$, then $|O_0\cap V(K_c)| = 1$ and in turn $\eta(K) \ge 7$ by Statement~4 of Lemma~\ref{lemma12}. 
Moreover, if $i=0$, then either $|O_1\cap V(K_c)| = 1$ or $|O_0\cap V(K_c)| = 2$; 
in both cases, by Statements~4 and~5 in Lemma~\ref{lemma12}, $\eta(K) \ge 7$. 
Therefore, $K$ is always not critical.

Case 5: $s(K) = 10$.
In this case, $i = 0$ or~$1$.
If $i=1$, then $\eta(K) \ge 6$ because $P_{v_j}$ is a $6^+$-path for the unique $2$-anchor $v_j$ in $K_c$. 
So, we may assume $i=0$. Then, $|O_0 \cap V(K_c)| = |T_0 \cap V(K_c)| = 1$.
By Statement~4 of Lemma~\ref{lemma12}, $\eta(K) \ge 6$.
Thus, we can conclude $\eta(K) \ge 6$ and we are done.
\end{proof}

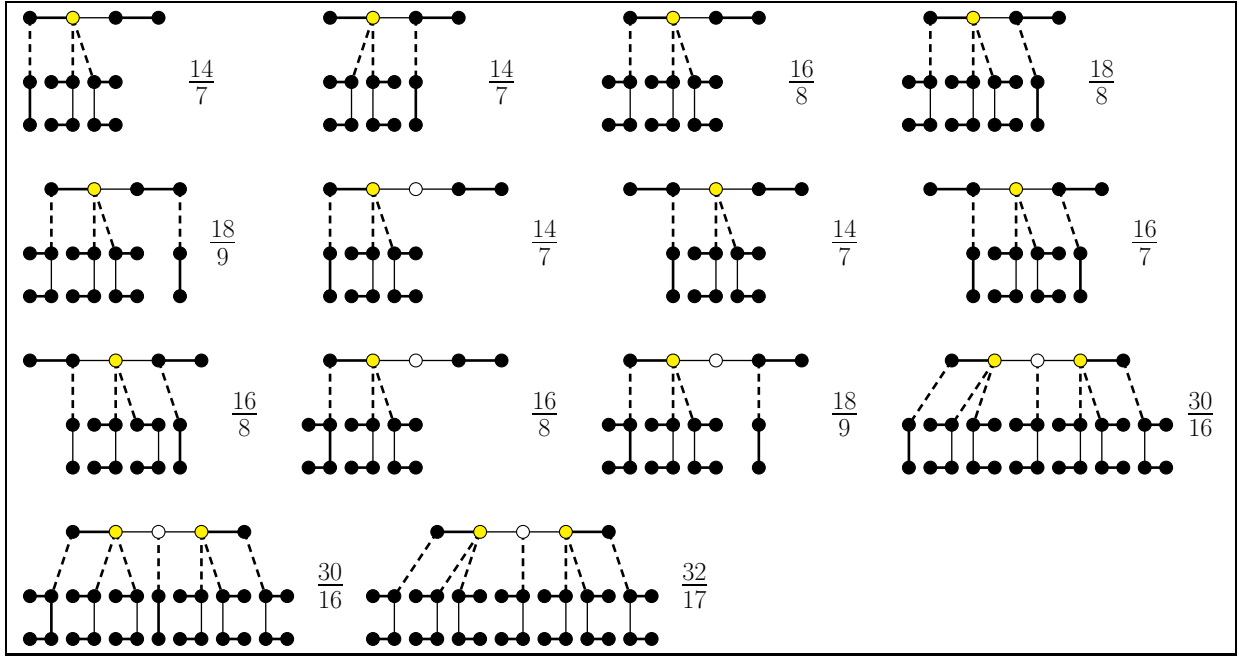
\begin{figure}[thb]
\begin{center}
\framebox{
\begin{minipage}{0.97\textwidth}
\begin{tikzpicture}[scale=0.57,transform shape]
\draw [thick, line width = 1pt] (-8, 8) -- (-7, 8);
\draw [thick, line width = 1pt] (-6, 8) -- (-5, 8);
\draw [thin, line width = 0.5pt] (-7, 8) -- (-6, 8);
\draw [densely dashed, line width = 1pt] (-8, 8) -- (-8, 6.5);
\draw [thick, line width = 1pt] (-8, 6.5) -- (-8, 5.5);
\draw [densely dashed, line width = 1pt] (-7, 8) -- (-7, 6.5);
\draw [thin, line width = 0.5pt] (-7, 6.5) -- (-7, 5.5);
\draw [thick, line width = 1pt] (-7, 6.5) -- (-7.5, 6.5);
\draw [thick, line width = 1pt] (-7, 5.5) -- (-7.5, 5.5);
\draw [densely dashed, line width = 1pt] (-7, 8) -- (-6.5, 6.5);
\draw [thin, line width = 0.5pt] (-6.5, 6.5) -- (-6.5, 5.5);
\draw [thick, line width = 1pt] (-6.5, 6.5) -- (-6, 6.5);
\draw [thick, line width = 1pt] (-6.5, 5.5) -- (-6, 5.5);
\filldraw (-8, 8) circle(.15);
\filldraw[fill = yellow] (-7, 8) circle(.15);
\filldraw (-6, 8) circle(.15);
\filldraw (-5, 8) circle(.15);
\filldraw (-6.5, 6.5) circle(.15);
\filldraw (-6.5, 5.5) circle(.15);
\filldraw (-7, 6.5) circle(.15);
\filldraw (-7, 5.5) circle(.15);
\filldraw (-7.5, 6.5) circle(.15);
\filldraw (-7.5, 5.5) circle(.15);
\filldraw (-6, 6.5) circle(.15);
\filldraw (-6, 5.5) circle(.15);
\filldraw (-8, 6.5) circle(.15);
\filldraw (-8, 5.5) circle(.15);
\node[font=\fontsize{25}{6}\selectfont] at (-4, 6.5) {$\frac {14}7$};

\draw [thick, line width = 1pt] (-1, 8) -- (0, 8);
\draw [thick, line width = 1pt] (1, 8) -- (2, 8);
\draw [thin, line width = 0.5pt] (0, 8) -- (1, 8);
\draw [densely dashed, line width = 1pt] (1, 8) -- (1, 6.5);
\draw [thick, line width = 1pt] (1, 6.5) -- (1, 5.5);
\draw [densely dashed, line width = 1pt] (0, 8) -- (0, 6.5);
\draw [thin, line width = 0.5pt] (0, 6.5) -- (0, 5.5);
\draw [thick, line width = 1pt] (0, 6.5) -- (0.5, 6.5);
\draw [thick, line width = 1pt] (0, 5.5) -- (0.5, 5.5);
\draw [densely dashed, line width = 1pt] (0, 8) -- (-0.5, 6.5);
\draw [thin, line width = 0.5pt] (-0.5, 6.5) -- (-0.5, 5.5);
\draw [thick, line width = 1pt] (-0.5, 6.5) -- (-1, 6.5);
\draw [thick, line width = 1pt] (-0.5, 5.5) -- (-1, 5.5);
\filldraw (-1, 8) circle(.15);
\filldraw[fill = yellow] (0, 8) circle(.15);
\filldraw (1, 8) circle(.15);
\filldraw (2, 8) circle(.15);
\filldraw (1, 6.5) circle(.15);
\filldraw (1, 5.5) circle(.15);
\filldraw (0, 6.5) circle(.15);
\filldraw (0, 5.5) circle(.15);
\filldraw (0.5, 6.5) circle(.15);
\filldraw (0.5, 5.5) circle(.15);
\filldraw (-1, 6.5) circle(.15);
\filldraw (-1, 5.5) circle(.15);
\filldraw (-0.5, 6.5) circle(.15);
\filldraw (-0.5, 5.5) circle(.15);
\node[font=\fontsize{25}{6}\selectfont] at (3,6.5) {$\frac {14}7$};

\draw [thick, line width = 1pt] (6, 8) -- (7, 8);
\draw [thick, line width = 1pt] (8, 8) -- (9, 8);
\draw [thin, line width = 0.5pt] (7, 8) -- (8, 8);
\draw [densely dashed, line width = 1pt] (7, 8) -- (7, 6.5);
\draw [thin, line width = 0.5pt] (7, 6.5) -- (7, 5.5);
\draw [thick, line width = 1pt] (7, 6.5) -- (6.5, 6.5);
\draw [thick, line width = 1pt] (7, 5.5) -- (6.5, 5.5);
\draw [densely dashed, line width = 1pt] (7, 8) -- (7.5, 6.5);
\draw [thin, line width = 0.5pt] (7.5, 6.5) -- (7.5, 5.5);
\draw [thick, line width = 1pt] (7.5, 6.5) -- (8, 6.5);
\draw [thick, line width = 1pt] (7.5, 5.5) -- (8, 5.5);
\draw [densely dashed, line width = 1pt] (6, 8) -- (6, 6.5);
\draw [thin, line width = 0.5pt] (6, 6.5) -- (6, 5.5);
\draw [thick, line width = 1pt] (6, 6.5) -- (5.5, 6.5);
\draw [thick, line width = 1pt] (6, 5.5) -- (5.5, 5.5);
\filldraw (6, 8) circle(.15);
\filldraw[fill = yellow] (7, 8) circle(.15);
\filldraw (8, 8) circle(.15);
\filldraw (9, 8) circle(.15);
\filldraw (7, 6.5) circle(.15);
\filldraw (7, 5.5) circle(.15);
\filldraw (8, 6.5) circle(.15);
\filldraw (8, 5.5) circle(.15);
\filldraw (7.5, 6.5) circle(.15);
\filldraw (7.5, 5.5) circle(.15);
\filldraw (6, 6.5) circle(.15);
\filldraw (6, 5.5) circle(.15);
\filldraw (5.5, 6.5) circle(.15);
\filldraw (5.5, 5.5) circle(.15);
\filldraw (6.5, 6.5) circle(.15);
\filldraw (6.5, 5.5) circle(.15);
\node[font=\fontsize{25}{6}\selectfont] at (10,6.5) {$\frac {16}8$};

\draw [thick, line width = 1pt] (13, 8) -- (14, 8);
\draw [thick, line width = 1pt] (15, 8) -- (16, 8);
\draw [thin, line width = 0.5pt] (14, 8) -- (15, 8);
\draw [densely dashed, line width = 1pt] (13, 8) -- (13, 6.5);
\draw [thin, line width = 0.5pt] (13, 6.5) -- (13, 5.5);
\draw [thick, line width = 1pt] (13, 6.5) -- (12.5, 6.5);
\draw [thick, line width = 1pt] (13, 5.5) -- (12.5, 5.5);
\draw [densely dashed, line width = 1pt] (14, 8) -- (14, 6.5);
\draw [thin, line width = 0.5pt] (14, 6.5) -- (14, 5.5);
\draw [thick, line width = 1pt] (14, 6.5) -- (13.5, 6.5);
\draw [thick, line width = 1pt] (14, 5.5) -- (13.5, 5.5);
\draw [densely dashed, line width = 1pt] (14, 8) -- (14.5, 6.5);
\draw [thin, line width = 0.5pt] (14.5, 6.5) -- (14.5, 5.5);
\draw [thick, line width = 1pt] (14.5, 6.5) -- (15, 6.5);
\draw [thick, line width = 1pt] (14.5, 5.5) -- (15, 5.5);
\draw [densely dashed, line width = 1pt] (15, 8) -- (15.5, 6.5);
\draw [thick, line width = 1pt] (15.5, 6.5) -- (15.5, 5.5);
\filldraw (13, 8) circle(.15);
\filldraw[fill = yellow] (14, 8) circle(.15);
\filldraw (15, 8) circle(.15);
\filldraw (16, 8) circle(.15);
\filldraw (14.5, 6.5) circle(.15);
\filldraw (14.5, 5.5) circle(.15);
\filldraw (14, 6.5) circle(.15);
\filldraw (14, 5.5) circle(.15);
\filldraw (13.5, 6.5) circle(.15);
\filldraw (13.5, 5.5) circle(.15);
\filldraw (15, 6.5) circle(.15);
\filldraw (15, 5.5) circle(.15);
\filldraw (13, 6.5) circle(.15);
\filldraw (13, 5.5) circle(.15);
\filldraw (12.5, 6.5) circle(.15);
\filldraw (12.5, 5.5) circle(.15);
\filldraw (15.5, 6.5) circle(.15);
\filldraw (15.5, 5.5) circle(.15);
\node[font=\fontsize{25}{6}\selectfont] at (17,6.5) {$\frac {18}8$};

\draw [thick, line width = 1pt] (-7.5, 4) -- (-6.5, 4);
\draw [thick, line width = 1pt] (-5.5, 4) -- (-4.5, 4);
\draw [thin, line width = 0.5pt] (-6.5, 4) -- (-5.5, 4);
\draw [densely dashed, line width = 1pt] (-7.5, 4) -- (-7.5, 2.5);
\draw [thin, line width = 0.5pt] (-7.5, 2.5) -- (-7.5, 1.5);
\draw [thick, line width = 1pt] (-7.5, 2.5) -- (-8, 2.5);
\draw [thick, line width = 1pt] (-7.5, 1.5) -- (-8, 1.5);
\draw [densely dashed, line width = 1pt] (-6.5, 4) -- (-6.5, 2.5);
\draw [thin, line width = 0.5pt] (-6.5, 2.5) -- (-6.5, 1.5);
\draw [thick, line width = 1pt] (-6.5, 2.5) -- (-7, 2.5);
\draw [thick, line width = 1pt] (-6.5, 1.5) -- (-7, 1.5);
\draw [densely dashed, line width = 1pt] (-6.5, 4) -- (-6, 2.5);
\draw [thin, line width = 0.5pt] (-6, 2.5) -- (-6, 1.5);
\draw [thick, line width = 1pt] (-6, 2.5) -- (-5.5, 2.5);
\draw [thick, line width = 1pt] (-6, 1.5) -- (-5.5, 1.5);
\draw [densely dashed, line width = 1pt] (-4.5, 4) -- (-4.5, 2.5);
\draw [thick, line width = 1pt] (-4.5, 2.5) -- (-4.5, 1.5);
\filldraw (-7.5, 4) circle(.15);
\filldraw[fill = yellow] (-6.5, 4) circle(.15);
\filldraw (-5.5, 4) circle(.15);
\filldraw (-4.5, 4) circle(.15);
\filldraw (-6, 2.5) circle(.15);
\filldraw (-6, 1.5) circle(.15);
\filldraw (-6.5, 2.5) circle(.15);
\filldraw (-6.5, 1.5) circle(.15);
\filldraw (-7, 2.5) circle(.15);
\filldraw (-7, 1.5) circle(.15);
\filldraw (-5.5, 2.5) circle(.15);
\filldraw (-5.5, 1.5) circle(.15);
\filldraw (-7.5, 2.5) circle(.15);
\filldraw (-7.5, 1.5) circle(.15);
\filldraw (-8, 2.5) circle(.15);
\filldraw (-8, 1.5) circle(.15);
\filldraw (-4.5, 2.5) circle(.15);
\filldraw (-4.5, 1.5) circle(.15);
\node[font=\fontsize{25}{6}\selectfont] at (-3.5,2.75) {$\frac {18}9$};

\draw [thick, line width = 1pt] (-1, 4) -- (0, 4);
\draw [thick, line width = 1pt] (2, 4) -- (3, 4);
\draw [thin, line width = 0.5pt] (0, 4) -- (1, 4) -- (2,4);
\draw [densely dashed, line width = 1pt] (-1, 4) -- (-1, 2.5);
\draw [thick, line width = 1pt] (-1, 2.5) -- (-1, 1.5);
\draw [densely dashed, line width = 1pt] (0, 4) -- (0, 2.5);
\draw [thin, line width = 0.5pt] (0, 2.5) -- (0, 1.5);
\draw [thick, line width = 1pt] (0, 2.5) -- (-0.5, 2.5);
\draw [thick, line width = 1pt] (0, 1.5) -- (-0.5, 1.5);
\draw [densely dashed, line width = 1pt] (0, 4) -- (0.5, 2.5);
\draw [thin, line width = 0.5pt] (0.5, 2.5) -- (0.5, 1.5);
\draw [thick, line width = 1pt] (0.5, 2.5) -- (1, 2.5);
\draw [thick, line width = 1pt] (0.5, 1.5) -- (1, 1.5);
\filldraw (-1, 4) circle(.15);
\filldraw[fill = yellow] (0, 4) circle(.15);
\filldraw[fill = white] (1, 4) circle(.15);
\filldraw (2, 4) circle(.15);
\filldraw (3, 4) circle(.15);
\filldraw (-1, 2.5) circle(.15);
\filldraw (-1, 1.5) circle(.15);
\filldraw (0, 2.5) circle(.15);
\filldraw (0, 1.5) circle(.15);
\filldraw (-0.5, 2.5) circle(.15);
\filldraw (-0.5, 1.5) circle(.15);
\filldraw (0.5, 2.5) circle(.15);
\filldraw (0.5, 1.5) circle(.15);
\filldraw (1, 2.5) circle(.15);
\filldraw (1, 1.5) circle(.15);
\node[font=\fontsize{25}{6}\selectfont] at (4, 2.75) {$\frac {14}7$};

\draw [thick, line width = 1pt] (6, 4) -- (7, 4);
\draw [thick, line width = 1pt] (9, 4) -- (10, 4);
\draw [thin, line width = 0.5pt] (7, 4) -- (8, 4) -- (9,4);
\draw [densely dashed, line width = 1pt] (7, 4) -- (7, 2.5);
\draw [thick, line width = 1pt] (7, 2.5) -- (7, 1.5);
\draw [densely dashed, line width = 1pt] (8, 4) -- (8, 2.5);
\draw [thin, line width = 0.5pt] (8, 2.5) -- (8, 1.5);
\draw [thick, line width = 1pt] (8, 2.5) -- (7.5, 2.5);
\draw [thick, line width = 1pt] (8, 1.5) -- (7.5, 1.5);
\draw [densely dashed, line width = 1pt] (8, 4) -- (8.5, 2.5);
\draw [thin, line width = 0.5pt] (8.5, 2.5) -- (8.5, 1.5);
\draw [thick, line width = 1pt] (8.5, 2.5) -- (9, 2.5);
\draw [thick, line width = 1pt] (8.5, 1.5) -- (9, 1.5);
\filldraw (6, 4) circle(.15);
\filldraw (7, 4) circle(.15);
\filldraw[fill = yellow] (8, 4) circle(.15);
\filldraw (9, 4) circle(.15);
\filldraw (10, 4) circle(.15);
\filldraw (9, 2.5) circle(.15);
\filldraw (9, 1.5) circle(.15);
\filldraw (7, 2.5) circle(.15);
\filldraw (7, 1.5) circle(.15);
\filldraw (8.5, 2.5) circle(.15);
\filldraw (8.5, 1.5) circle(.15);
\filldraw (7.5, 2.5) circle(.15);
\filldraw (7.5, 1.5) circle(.15);
\filldraw (8, 2.5) circle(.15);
\filldraw (8, 1.5) circle(.15);
\node[font=\fontsize{25}{6}\selectfont] at (11,2.75) {$\frac {14}7$};

\draw [thick, line width = 1pt] (13, 4) -- (14, 4);
\draw [thick, line width = 1pt] (16, 4) -- (17, 4);
\draw [thin, line width = 0.5pt] (14, 4) -- (15, 4) -- (16,4);
\draw [densely dashed, line width = 1pt] (14, 4) -- (14, 2.5);
\draw [thick, line width = 1pt] (14, 2.5) -- (14, 1.5);
\draw [densely dashed, line width = 1pt] (15, 4) -- (15, 2.5);
\draw [thin, line width = 0.5pt] (15, 2.5) -- (15, 1.5);
\draw [thick, line width = 1pt] (15, 2.5) -- (14.5, 2.5);
\draw [thick, line width = 1pt] (15, 1.5) -- (14.5, 1.5);
\draw [densely dashed, line width = 1pt] (15, 4) -- (15.5, 2.5);
\draw [thin, line width = 0.5pt] (15.5, 2.5) -- (15.5, 1.5);
\draw [thick, line width = 1pt] (15.5, 2.5) -- (16, 2.5);
\draw [thick, line width = 1pt] (15.5, 1.5) -- (16, 1.5);
\draw [densely dashed, line width = 1pt] (16, 4) -- (16.5, 2.5);
\draw [thick, line width = 1pt] (16.5, 2.5) -- (16.5, 1.5);
\filldraw (13, 4) circle(.15);
\filldraw (14, 4) circle(.15);
\filldraw[fill = yellow] (15, 4) circle(.15);
\filldraw (16, 4) circle(.15);
\filldraw (17, 4) circle(.15);
\filldraw (15.5, 2.5) circle(.15);
\filldraw (15.5, 1.5) circle(.15);
\filldraw (15, 2.5) circle(.15);
\filldraw (15, 1.5) circle(.15);
\filldraw (14.5, 2.5) circle(.15);
\filldraw (14.5, 1.5) circle(.15);
\filldraw (16.5, 2.5) circle(.15);
\filldraw (16.5, 1.5) circle(.15);
\filldraw (16, 2.5) circle(.15);
\filldraw (16, 1.5) circle(.15);
\filldraw (14, 2.5) circle(.15);
\filldraw (14, 1.5) circle(.15);
\node[font=\fontsize{25}{6}\selectfont] at (18, 2.75) {$\frac {16}7$};

\draw [thick, line width = 1pt] (-8, 0) -- (-7, 0);
\draw [thick, line width = 1pt] (-5, 0) -- (-4, 0);
\draw [thin, line width = 0.5pt] (-7, 0) -- (-6, 0) -- (-5,0);
\draw [densely dashed, line width = 1pt] (-7, 0) -- (-7, -1.5);
\draw [thin, line width = 0.5pt] (-7, -1.5) -- (-7, -2.5);
\draw [densely dashed, line width = 1pt] (-6, 0) -- (-6, -1.5);
\draw [thin, line width = 0.5pt] (-6, -1.5) -- (-6, -2.5);
\draw [thick, line width = 1pt] (-6, -1.5) -- (-6.5, -1.5);
\draw [thick, line width = 1pt] (-6, -2.5) -- (-6.5, -2.5);
\draw [densely dashed, line width = 1pt] (-6, 0) -- (-5.5, -1.5);
\draw [thin, line width = 0.5pt] (-5, -1.5) -- (-5, -2.5);
\draw [thick, line width = 1pt] (-5.5, -1.5) -- (-5, -1.5);
\draw [thick, line width = 1pt] (-5.5, -2.5) -- (-5, -2.5);
\draw [densely dashed, line width = 1pt] (-5, 0) -- (-4.5, -1.5);
\draw [thick, line width = 1pt] (-4.5, -1.5) -- (-4.5, -2.5);
\filldraw (-8, 0) circle(.15);
\filldraw (-7, 0) circle(.15);
\filldraw[fill = yellow] (-6, 0) circle(.15);
\filldraw (-5, 0) circle(.15);
\filldraw (-4, 0) circle(.15);
\filldraw (-6.5, -1.5) circle(.15);
\filldraw (-6.5, -2.5) circle(.15);
\filldraw (-7, -1.5) circle(.15);
\filldraw (-7, -2.5) circle(.15);
\filldraw (-5.5, -1.5) circle(.15);
\filldraw (-5.5, -2.5) circle(.15);
\filldraw (-6, -1.5) circle(.15);
\filldraw (-6, -2.5) circle(.15);
\filldraw (-5, -1.5) circle(.15);
\filldraw (-5, -2.5) circle(.15);
\filldraw (-4.5, -1.5) circle(.15);
\filldraw (-4.5, -2.5) circle(.15);
\node[font=\fontsize{25}{6}\selectfont] at (-3, -1.25) {$\frac {16}8$};

\draw [thick, line width = 1pt] (-1, 0) -- (0, 0);
\draw [thick, line width = 1pt] (2, 0) -- (3, 0);
\draw [thin, line width = 0.5pt] (0, 0) -- (1, 0) -- (2,0);
\draw [densely dashed, line width = 1pt] (-1, 0) -- (-1, -1.5);
\draw [thick, line width = 1pt] (-1, -1.5) -- (-1, -2.5);
\draw [thick, line width = 1pt] (-1, -1.5) -- (-1.5, -1.5);
\draw [thick, line width = 1pt] (-1, -2.5) -- (-1.5, -2.5);
\draw [densely dashed, line width = 1pt] (0, 0) -- (0, -1.5);
\draw [thin, line width = 0.5pt] (0, -1.5) -- (0, -2.5);
\draw [thick, line width = 1pt] (0, -1.5) -- (-0.5, -1.5);
\draw [thick, line width = 1pt] (0, -2.5) -- (-0.5, -2.5);
\draw [densely dashed, line width = 1pt] (0, 0) -- (0.5, -1.5);
\draw [thin, line width = 0.5pt] (0.5, -1.5) -- (0.5, -2.5);
\draw [thick, line width = 1pt] (0.5, -1.5) -- (1, -1.5);
\draw [thick, line width = 1pt] (0.5, -2.5) -- (1, -2.5);
\filldraw (-1, 0) circle(.15);
\filldraw[fill = yellow] (0, 0) circle(.15);
\filldraw[fill = white] (1, 0) circle(.15);
\filldraw (2, 0) circle(.15);
\filldraw (3, 0) circle(.15);
\filldraw (-1, -1.5) circle(.15);
\filldraw (-1, -2.5) circle(.15);
\filldraw (-1.5, -1.5) circle(.15);
\filldraw (-1.5, -2.5) circle(.15);
\filldraw (0, -1.5) circle(.15);
\filldraw (0, -2.5) circle(.15);
\filldraw (-0.5, -1.5) circle(.15);
\filldraw (-0.5, -2.5) circle(.15);
\filldraw (0.5, -1.5) circle(.15);
\filldraw (0.5, -2.5) circle(.15);
\filldraw (1, -1.5) circle(.15);
\filldraw (1, -2.5) circle(.15);
\node[font=\fontsize{25}{6}\selectfont] at (4,-1.25) {$\frac {16}8$};

\draw [thick, line width = 1pt] (6, 0) -- (7, 0);
\draw [thick, line width = 1pt] (9, 0) -- (10, 0);
\draw [thin, line width = 0.5pt] (7, 0) -- (8, 0) -- (9, 0);
\draw [densely dashed, line width = 1pt] (6, 0) -- (6, -1.5);
\draw [thick, line width = 1pt] (6, -1.5) -- (6, -2.5);
\draw [thick, line width = 1pt] (6, -1.5) -- (5.5, -1.5);
\draw [thick, line width = 1pt] (6, -2.5) -- (5.5, -2.5);
\draw [densely dashed, line width = 1pt] (7, 0) -- (7, -1.5);
\draw [thin, line width = 0.5pt] (7, -1.5) -- (7, -2.5);
\draw [thick, line width = 1pt] (7, -1.5) -- (6.5, -1.5);
\draw [thick, line width = 1pt] (7, -2.5) -- (6.5, -2.5);
\draw [densely dashed, line width = 1pt] (7, 0) -- (7.5, -1.5);
\draw [thin, line width = 0.5pt] (7.5, -1.5) -- (7.5, -2.5);
\draw [thick, line width = 1pt] (7.5, -1.5) -- (8, -1.5);
\draw [thick, line width = 1pt] (7.5, -2.5) -- (8, -2.5);
\draw [densely dashed, line width = 1pt] (9, 0) -- (9, -1.5);
\draw [thick, line width = 1pt] (9, -1.5) -- (9, -2.5);
\filldraw (6, 0) circle(.15);
\filldraw[fill = yellow] (7, 0) circle(.15);
\filldraw[fill = white] (8, 0) circle(.15);
\filldraw (9, 0) circle(.15);
\filldraw (10, 0) circle(.15);
\filldraw (6, -1.5) circle(.15);
\filldraw (6, -2.5) circle(.15);
\filldraw (5.5, -1.5) circle(.15);
\filldraw (5.5, -2.5) circle(.15);
\filldraw (7, -1.5) circle(.15);
\filldraw (7, -2.5) circle(.15);
\filldraw (6.5, -1.5) circle(.15);
\filldraw (6.5, -2.5) circle(.15);
\filldraw (7.5, -1.5) circle(.15);
\filldraw (7.5, -2.5) circle(.15);
\filldraw (8, -1.5) circle(.15);
\filldraw (8, -2.5) circle(.15);
\filldraw (9, -1.5) circle(.15);
\filldraw (9, -2.5) circle(.15);
\node[font=\fontsize{25}{6}\selectfont] at (11,-1.25) {$\frac {18}9$};

\draw [thick, line width = 1pt] (13.5, 0) -- (14.5, 0);
\draw [thick, line width = 1pt] (16.5, 0) -- (17.5, 0);
\draw [thin, line width = 0.5pt] (14.5, 0) -- (15.5, 0) -- (16.5, 0);
\draw [densely dashed, line width = 1pt] (13.5, 0) -- (12.5, -1.5);
\draw [thick, line width = 1pt] (12.5, -1.5) -- (12.5, -2.5);
\draw [densely dashed, line width = 1pt] (14.5, 0) -- (13.5, -1.5);
\draw [thin, line width = 0.5pt] (13.5, -1.5) -- (13.5, -2.5);
\draw [thick, line width = 1pt] (13.5, -1.5) -- (13, -1.5);
\draw [thick, line width = 1pt] (13.5, -2.5) -- (13, -2.5);
\draw [densely dashed, line width = 1pt] (14.5, 0) -- (14, -1.5);
\draw [thin, line width = 0.5pt] (14, -1.5) -- (14, -2.5);
\draw [thick, line width = 1pt] (14, -1.5) -- (14.5, -1.5);
\draw [thick, line width = 1pt] (14, -2.5) -- (14.5, -2.5);
\draw [densely dashed, line width = 1pt] (15.5, 0) -- (15.5, -1.5);
\draw [thin, line width = 0.5pt] (15.5, -1.5) -- (15.5, -2.5);
\draw [thick, line width = 1pt] (15.5, -1.5) -- (15, -1.5);
\draw [thick, line width = 1pt] (15.5, -2.5) -- (15, -2.5);
\draw [densely dashed, line width = 1pt] (16.5, 0) -- (16.5, -1.5);
\draw [thin, line width = 0.5pt] (16.5, -1.5) -- (16.5, -2.5);
\draw [thick, line width = 1pt] (16.5, -1.5) -- (16, -1.5);
\draw [thick, line width = 1pt] (16.5, -2.5) -- (16, -2.5);
\draw [densely dashed, line width = 1pt] (16.5, 0) -- (17, -1.5);
\draw [thin, line width = 0.5pt] (17, -1.5) -- (17, -2.5);
\draw [thick, line width = 1pt] (17, -1.5) -- (17.5, -1.5);
\draw [thick, line width = 1pt] (17, -2.5) -- (17.5, -2.5);
\draw [densely dashed, line width = 1pt] (17.5, 0) -- (18, -1.5);
\draw [thin, line width = 0.5pt] (18, -1.5) -- (18, -2.5);
\draw [thick, line width = 1pt] (18, -1.5) -- (18.5, -1.5);
\draw [thick, line width = 1pt] (18, -2.5) -- (18.5, -2.5);
\filldraw (13.5, 0) circle(.15);
\filldraw[fill = yellow] (14.5, 0) circle(.15);
\filldraw[fill = white] (15.5, 0) circle(.15);
\filldraw[fill = yellow] (16.5, 0) circle(.15);
\filldraw (17.5, 0) circle(.15);
\filldraw (13.5, -1.5) circle(.15);
\filldraw (13.5, -2.5) circle(.15);
\filldraw (13, -1.5) circle(.15);
\filldraw (13, -2.5) circle(.15);
\filldraw (14.5, -1.5) circle(.15);
\filldraw (14.5, -2.5) circle(.15);
\filldraw (14, -1.5) circle(.15);
\filldraw (14, -2.5) circle(.15);
\filldraw (15, -1.5) circle(.15);
\filldraw (15, -2.5) circle(.15);
\filldraw (16, -1.5) circle(.15);
\filldraw (16, -2.5) circle(.15);
\filldraw (15.5, -1.5) circle(.15);
\filldraw (15.5, -2.5) circle(.15);
\filldraw (16.5, -1.5) circle(.15);
\filldraw (16.5, -2.5) circle(.15);
\filldraw (17, -1.5) circle(.15);
\filldraw (17, -2.5) circle(.15);
\filldraw (18, -1.5) circle(.15);
\filldraw (18, -2.5) circle(.15);
\filldraw (18.5, -1.5) circle(.15);
\filldraw (18.5, -2.5) circle(.15);
\filldraw (17.5, -1.5) circle(.15);
\filldraw (17.5, -2.5) circle(.15);
\filldraw (12.5, -1.5) circle(.15);
\filldraw (12.5, -2.5) circle(.15);
\node[font=\fontsize{25}{6}\selectfont] at (19.3,-1.25) {$\frac {30}{16}$};

\draw [thick, line width = 1pt] (-7, -4) -- (-6, -4);
\draw [thick, line width = 1pt] (-4, -4) -- (-3, -4);
\draw [thin, line width = 0.5pt] (-6, -4) -- (-5, -4) -- (-4,-4);
\draw [densely dashed, line width = 1pt] (-7, -4) -- (-7.5, -5.5);
\draw [thick, line width = 1pt] (-7.5, -5.5) -- (-7.5, -6.5);
\draw [thick, line width = 1pt] (-7.5, -5.5) -- (-8, -5.5);
\draw [thick, line width = 1pt] (-7.5, -6.5) -- (-8, -6.5);
\draw [densely dashed, line width = 1pt] (-6, -4) -- (-6.5, -5.5);
\draw [thin, line width = 0.5pt] (-6.5, -5.5) -- (-6.5, -6.5);
\draw [thick, line width = 1pt] (-6.5, -5.5) -- (-7, -5.5);
\draw [thick, line width = 1pt] (-6.5, -6.5) -- (-7, -6.5);
\draw [densely dashed, line width = 1pt] (-6, -4) -- (-5.5, -5.5);
\draw [thin, line width = 0.5pt] (-5.5, -5.5) -- (-5.5, -6.5);
\draw [thick, line width = 1pt] (-5.5, -5.5) -- (-6, -5.5);
\draw [thick, line width = 1pt] (-5.5, -6.5) -- (-6, -6.5);
\draw [densely dashed, line width = 1pt] (-5, -4) -- (-5, -5.5);
\draw [thick, line width = 1pt] (-5, -5.5) -- (-5, -6.5);
\draw [densely dashed, line width = 1pt] (-4, -4) -- (-4, -5.5);
\draw [thin, line width = 0.5pt] (-4, -5.5) -- (-4, -6.5);
\draw [thick, line width = 1pt] (-4, -5.5) -- (-4.5, -5.5);
\draw [thick, line width = 1pt] (-4, -6.5) -- (-4.5, -6.5);
\draw [densely dashed, line width = 1pt] (-4, -4) -- (-3.5, -5.5);
\draw [thin, line width = 0.5pt] (-3.5, -5.5) -- (-3.5, -6.5);
\draw [thick, line width = 1pt] (-3.5, -5.5) -- (-3, -5.5);
\draw [thick, line width = 1pt] (-3.5, -6.5) -- (-3, -6.5);
\draw [densely dashed, line width = 1pt] (-3, -4) -- (-2.5, -5.5);
\draw [thin, line width = 0.5pt] (-2.5, -5.5) -- (-2.5, -6.5);
\draw [thick, line width = 1pt] (-2.5, -5.5) -- (-2, -5.5);
\draw [thick, line width = 1pt] (-2.5, -6.5) -- (-2, -6.5);
\filldraw (-7, -4) circle(.15);
\filldraw[fill = yellow] (-6, -4) circle(.15);
\filldraw[fill = white] (-5, -4) circle(.15);
\filldraw[fill = yellow] (-4, -4) circle(.15);
\filldraw (-3, -4) circle(.15);
\filldraw (-7, -5.5) circle(.15);
\filldraw (-7, -6.5) circle(.15);
\filldraw (-7.5, -5.5) circle(.15);
\filldraw (-7.5, -6.5) circle(.15);
\filldraw (-6, -5.5) circle(.15);
\filldraw (-6, -6.5) circle(.15);
\filldraw (-6.5, -5.5) circle(.15);
\filldraw (-6.5, -6.5) circle(.15);
\filldraw (-5.5, -5.5) circle(.15);
\filldraw (-5.5, -6.5) circle(.15);
\filldraw (-4.5, -5.5) circle(.15);
\filldraw (-4.5, -6.5) circle(.15);
\filldraw (-5, -5.5) circle(.15);
\filldraw (-5, -6.5) circle(.15);
\filldraw (-4, -5.5) circle(.15);
\filldraw (-4, -6.5) circle(.15);
\filldraw (-3.5, -5.5) circle(.15);
\filldraw (-3.5, -6.5) circle(.15);
\filldraw (-2.5, -5.5) circle(.15);
\filldraw (-2.5, -6.5) circle(.15);
\filldraw (-2, -5.5) circle(.15);
\filldraw (-2, -6.5) circle(.15);
\filldraw (-3, -5.5) circle(.15);
\filldraw (-3, -6.5) circle(.15);
\filldraw (-8, -5.5) circle(.15);
\filldraw (-8, -6.5) circle(.15);
\node[font=\fontsize{25}{6}\selectfont] at (-1, -5.25) {$\frac {30}{16}$};

\draw [thick, line width = 1pt] (1.5, -4) -- (2.5, -4);
\draw [thick, line width = 1pt] (4.5, -4) -- (5.5, -4);
\draw [thin, line width = 0.5pt] (2.5, -4) -- (3.5, -4) -- (4.5, -4);
\draw [densely dashed, line width = 1pt] (1.5, -4) -- (0.5, -5.5);
\draw [thin, line width = 0.5pt] (0.5, -5.5) -- (0.5, -6.5);
\draw [thick, line width = 1pt] (0.5, -5.5) -- (0, -5.5);
\draw [thick, line width = 1pt] (0.5, -6.5) -- (0, -6.5);
\draw [densely dashed, line width = 1pt] (2.5, -4) -- (1.5, -5.5);
\draw [thin, line width = 0.5pt] (1.5, -5.5) -- (1.5, -6.5);
\draw [thick, line width = 1pt] (1.5, -5.5) -- (1, -5.5);
\draw [thick, line width = 1pt] (1.5, -6.5) -- (1, -6.5);
\draw [densely dashed, line width = 1pt] (2.5, -4) -- (2, -5.5);
\draw [thin, line width = 0.5pt] (2, -5.5) -- (2, -6.5);
\draw [thick, line width = 1pt] (2, -5.5) -- (2.5, -5.5);
\draw [thick, line width = 1pt] (2, -6.5) -- (2.5, -6.5);
\draw [densely dashed, line width = 1pt] (3.5, -4) -- (3.5, -5.5);
\draw [thin, line width = 0.5pt] (3.5, -5.5) -- (3.5, -6.5);
\draw [thick, line width = 1pt] (3.5, -5.5) -- (3, -5.5);
\draw [thick, line width = 1pt] (3.5, -6.5) -- (3, -6.5);
\draw [densely dashed, line width = 1pt] (4.5, -4) -- (4.5, -5.5);
\draw [thin, line width = 0.5pt] (4.5, -5.5) -- (4.5, -6.5);
\draw [thick, line width = 1pt] (4.5, -5.5) -- (4, -5.5);
\draw [thick, line width = 1pt] (4.5, -6.5) -- (4, -6.5);
\draw [densely dashed, line width = 1pt] (4.5, -4) -- (5, -5.5);
\draw [thin, line width = 0.5pt] (5, -5.5) -- (5, -6.5);
\draw [thick, line width = 1pt] (5, -5.5) -- (5.5, -5.5);
\draw [thick, line width = 1pt] (5, -6.5) -- (5.5, -6.5);
\draw [densely dashed, line width = 1pt] (5.5, -4) -- (6, -5.5);
\draw [thin, line width = 0.5pt] (6, -5.5) -- (6, -6.5);
\draw [thick, line width = 1pt] (6, -5.5) -- (6.5, -5.5);
\draw [thick, line width = 1pt] (6, -6.5) -- (6.5, -6.5);
\filldraw (1.5, -4) circle(.15);
\filldraw[fill = yellow] (2.5, -4) circle(.15);
\filldraw[fill = white] (3.5, -4) circle(.15);
\filldraw[fill = yellow] (4.5, -4) circle(.15);
\filldraw (5.5, -4) circle(.15);
\filldraw (0.5, -5.5) circle(.15);
\filldraw (0.5, -6.5) circle(.15);
\filldraw (1.5, -5.5) circle(.15);
\filldraw (1.5, -6.5) circle(.15);
\filldraw (2.5, -5.5) circle(.15);
\filldraw (2.5, -6.5) circle(.15);
\filldraw (2, -5.5) circle(.15);
\filldraw (2, -6.5) circle(.15);
\filldraw (3, -5.5) circle(.15);
\filldraw (3, -6.5) circle(.15);
\filldraw (4, -5.5) circle(.15);
\filldraw (4, -6.5) circle(.15);
\filldraw (3.5, -5.5) circle(.15);
\filldraw (3.5, -6.5) circle(.15);
\filldraw (4.5, -5.5) circle(.15);
\filldraw (4.5, -6.5) circle(.15);
\filldraw (5, -5.5) circle(.15);
\filldraw (5, -6.5) circle(.15);
\filldraw (6, -5.5) circle(.15);
\filldraw (6, -6.5) circle(.15);
\filldraw (5.5, -5.5) circle(.15);
\filldraw (5.5, -6.5) circle(.15);
\filldraw (6.5, -5.5) circle(.15);
\filldraw (6.5, -6.5) circle(.15);
\filldraw (1, -5.5) circle(.15);
\filldraw (1, -6.5) circle(.15);
\filldraw (0, -5.5) circle(.15);
\filldraw (0, -6.5) circle(.15);
\node[font=\fontsize{25}{6}\selectfont] at (7.5,-5.25) {$\frac {32}{17}$};
\end{tikzpicture}
\end{minipage}}
\end{center}
\captionsetup{width=1.0\linewidth}
\caption{All possible structures of a critical component of $H + C$. 
	In these structures, the yellow vertices are $2$-anchors and the filled (white, respectively) vertices are in (not in, respectively) $V(M)$.\label{fig03}}
\end{figure}

From Lemmas~\ref{lemma16}--\ref{lemma22} that estimate the quantity $\frac {s(K)}{\eta(K)}$, a critical component $K$ has exactly one $2$-anchor or two $2$-anchors.
Moreover, those $2$-anchors are all in $T_2$. 
This motivates the next definition and Fact~\ref{fact01} follows immediately.

\begin{definition}
\label{def10}
An anchor is {\em critical} if it is a $2$-anchor in a critical component.
A satellite-element is {\em critical} if its rescue-anchor is a critical anchor.
\end{definition}

\begin{fact}
\label{fact01}
A critical anchor is in $T_2$ and a critical satellite-element is a bi-star component of $H$.
\end{fact}

For a critical component $K$, $\frac {s(K)}{\eta(K)}$ is relatively large and this challenges the performance ratio of an approximation algorithm.
We design three operations for decreasing the number of critical components in the next subsection.
Before that, we give the definition of a basic {\em moving} operation.

\begin{definition}
\label{def11}
Suppose that $v \in V(H)$ and $S$ is a satellite-element in $H+C$ such that $S$ has a vertex $w$ with $\{v,w\}\in E(G)$. 
Then, {\em moving $S$ to $v$ in $H+C$} is the operation of modifying $C$ by replacing the rescue-edge of $S$ with the edge $\{v,w\}$. 
\end{definition}

We prove an important fact below.

\begin{fact}
\label{fact02}
For each critical component $K$ of $H+C$ and each critical satellite-element $S$ of $K$, the following hold:
\begin{enumerate}
\parskip=0pt
\item $K$ has exactly one or two critical anchors and $s(K) \in \{ 14, 16, 18, 30, 32 \}$.
\item After modifying $H+C$ by moving $S$ to a vertex not in $K$, $K$ is no longer critical and remains not isolated in $H+C$.
\item Moving $S$ to a $0$- or $1$-anchor of $K$ makes $K$ no longer critical.
\end{enumerate}
\end{fact}
\begin{proof}
Statement~1 follows from Lemmas~\ref{lemma16}--\ref{lemma22} immediately.

{\em Statement~2.}
By Statement~1, each critical component of $H+C$ has exactly one or two $2$-anchors.
So, if $K$ has exactly one $2$-anchor, then after $S$ is moved outside $K$, $K$ has no $2$-anchor and hence cannot be critical any more. 
Thus, we may assume that $K$ has exactly two $2$-anchors. 
Then, by Lemmas~\ref{lemma19} and~\ref{lemma20}, $s(K) \in \{ 30, 32 \}$. 
Now, since $S$ is a bi-star component of $H$, we know that after moving $S$ outside $K$, $s(K) \in \{ 26, 28 \}$ and hence $K$ is not critical by Statement~1.

{\em Statement~3.}
First suppose that $v$ is a $0$-anchor of $K$. 
If $K$ has exactly one $2$-anchor, then after moving $S$ to $v$, $K$ has no $2$-anchor and hence is not critical. 
So, we may assume that $K$ has exactly two $2$-anchors. 
Then, by Lemmas~\ref{lemma19} and~\ref{lemma20}, $K_c$ is a $5$-path and both $v_2$ and $v_4$ are in $T_2$.
After moving $S$ to $v$, $v_2$ or $v_4$ is not in $T_2$ but $s(K)$ remains unchanged, implying that $K$ is not critical.

Next suppose that $v$ is a $1$-anchor of $K$.
If $v \in O_0$, then after moving $S$ to $v$, $v \in T_1$ and hence $K$ is not critical. 
So, we may assume that $v \in O_1$. 
Then, using Figure~\ref{fig03}, one can verify that $K$ is not critical after moving $S$ to $v$.
\end{proof}

By Statement~2 of Fact~\ref{fact02}, $K$ is no longer critical after moving one critical satellite-element $S$ of $K$ to a vertex $v$ appearing in a component $K' \ne K$. 
Moreover, if $v \in O_0$, then after the move of $S$, $v \in T_1$ and hence $K'$ does not become critical. 
However, $K'$ may become critical when $v$ is a $0$-anchor or $v \in O_1$. 
This motivates the next definition.

\begin{definition}
\label{def12}
Let $K$ be a composite component of $H+C$. 
If $K$ has a $1$-anchor $v \in O_1$ such that $G$ has an edge between $v$ and some critical satellite-element $S$ of $H+C$ and moving $S$ to $v$ in $H+C$ makes $K$ critical in $H+C$, 
then we call $K$ a {\em responsible component} of $H+C$ and call $v$ a {\em responsible $1$-anchor} of $H+C$.
\end{definition}

By Statement~2 of Fact~\ref{fact02}, if $K$ is critical, then $K$ is no longer critical after moving a satellite-element $S$ not in $K$ to a 1-anchor of $K$. 
This together with Statement~3 of Fact~\ref{fact02} implies that no component $K$ of $H+C$ can be both critical and responsible.
Moreover, the possible structures of responsible components can be obtained by removing exactly one critical satellite-element of a critical component of $H+C$.
All possible structures of a responsible component are listed in Figure~\ref{fig04}.

\begin{figure}[thb]
\begin{center}
\framebox{
\begin{minipage}{0.97\textwidth}
\begin{tikzpicture}[scale=0.58,transform shape]
\draw [thick, line width = 1pt] (-8, 8) -- (-7, 8);
\draw [thick, line width = 1pt] (-6, 8) -- (-5, 8);
\draw [thin, line width = 0.5pt] (-7, 8) -- (-6, 8);
\draw [densely dashed, line width = 1pt] (-8, 8) -- (-8, 6.5);
\draw [thick, line width = 1pt] (-8, 6.5) -- (-8, 5.5);
\draw [densely dashed, line width = 1pt] (-7, 8) -- (-7, 6.5);
\draw [thin, line width = 0.5pt] (-7, 6.5) -- (-7, 5.5);
\draw [thick, line width = 1pt] (-7, 6.5) -- (-7.5, 6.5);
\draw [thick, line width = 1pt] (-7, 5.5) -- (-7.5, 5.5);
\filldraw (-8, 8) circle(.15);
\filldraw[fill = red] (-7, 8) circle(.15);
\filldraw (-6, 8) circle(.15);
\filldraw (-5, 8) circle(.15);
\filldraw (-7, 6.5) circle(.15);
\filldraw (-7, 5.5) circle(.15);
\filldraw (-7.5, 6.5) circle(.15);
\filldraw (-7.5, 5.5) circle(.15);
\filldraw (-8, 6.5) circle(.15);
\filldraw (-8, 5.5) circle(.15);
\node[font=\fontsize{25}{6}\selectfont] at (-4, 6.5) {$\frac {10}7$};

\draw [thick, line width = 1pt] (-1, 8) -- (0, 8);
\draw [thick, line width = 1pt] (1, 8) -- (2, 8);
\draw [thin, line width = 0.5pt] (0, 8) -- (1, 8);
\draw [densely dashed, line width = 1pt] (1, 8) -- (1, 6.5);
\draw [thick, line width = 1pt] (1, 6.5) -- (1, 5.5);
\draw [densely dashed, line width = 1pt] (0, 8) -- (0, 6.5);
\draw [thin, line width = 0.5pt] (0, 6.5) -- (0, 5.5);
\draw [thick, line width = 1pt] (0, 6.5) -- (0.5, 6.5);
\draw [thick, line width = 1pt] (0, 5.5) -- (0.5, 5.5);
\filldraw (-1, 8) circle(.15);
\filldraw[fill = red] (0, 8) circle(.15);
\filldraw (1, 8) circle(.15);
\filldraw (2, 8) circle(.15);
\filldraw (1, 6.5) circle(.15);
\filldraw (1, 5.5) circle(.15);
\filldraw (0, 6.5) circle(.15);
\filldraw (0, 5.5) circle(.15);
\filldraw (0.5, 6.5) circle(.15);
\filldraw (0.5, 5.5) circle(.15);
\node[font=\fontsize{25}{6}\selectfont] at (3,6.5) {$\frac {10}7$};

\draw [thick, line width = 1pt] (6, 8) -- (7, 8);
\draw [thick, line width = 1pt] (8, 8) -- (9, 8);
\draw [thin, line width = 0.5pt] (7, 8) -- (8, 8);
\draw [densely dashed, line width = 1pt] (7, 8) -- (7, 6.5);
\draw [thin, line width = 0.5pt] (7, 6.5) -- (7, 5.5);
\draw [thick, line width = 1pt] (7, 6.5) -- (6.5, 6.5);
\draw [thick, line width = 1pt] (7, 5.5) -- (6.5, 5.5);
\draw [densely dashed, line width = 1pt] (6, 8) -- (6, 6.5);
\draw [thin, line width = 0.5pt] (6, 6.5) -- (6, 5.5);
\draw [thick, line width = 1pt] (6, 6.5) -- (5.5, 6.5);
\draw [thick, line width = 1pt] (6, 5.5) -- (5.5, 5.5);
\filldraw (6, 8) circle(.15);
\filldraw[fill = red] (7, 8) circle(.15);
\filldraw (8, 8) circle(.15);
\filldraw (9, 8) circle(.15);
\filldraw (7, 6.5) circle(.15);
\filldraw (7, 5.5) circle(.15);
\filldraw (6, 6.5) circle(.15);
\filldraw (6, 5.5) circle(.15);
\filldraw (5.5, 6.5) circle(.15);
\filldraw (5.5, 5.5) circle(.15);
\filldraw (6.5, 6.5) circle(.15);
\filldraw (6.5, 5.5) circle(.15);
\node[font=\fontsize{25}{6}\selectfont] at (10,6.5) {$\frac {12}8$};

\draw [thick, line width = 1pt] (13, 8) -- (14, 8);
\draw [thick, line width = 1pt] (15, 8) -- (16, 8);
\draw [thin, line width = 0.5pt] (14, 8) -- (15, 8);
\draw [densely dashed, line width = 1pt] (13, 8) -- (13, 6.5);
\draw [thin, line width = 0.5pt] (13, 6.5) -- (13, 5.5);
\draw [thick, line width = 1pt] (13, 6.5) -- (12.5, 6.5);
\draw [thick, line width = 1pt] (13, 5.5) -- (12.5, 5.5);
\draw [densely dashed, line width = 1pt] (14, 8) -- (14, 6.5);
\draw [thin, line width = 0.5pt] (14, 6.5) -- (14, 5.5);
\draw [thick, line width = 1pt] (14, 6.5) -- (13.5, 6.5);
\draw [thick, line width = 1pt] (14, 5.5) -- (13.5, 5.5);
\draw [densely dashed, line width = 1pt] (15, 8) -- (15, 6.5);
\draw [thick, line width = 1pt] (15, 6.5) -- (15, 5.5);
\filldraw (13, 8) circle(.15);
\filldraw[fill = red] (14, 8) circle(.15);
\filldraw (15, 8) circle(.15);
\filldraw (16, 8) circle(.15);
\filldraw (14, 6.5) circle(.15);
\filldraw (14, 5.5) circle(.15);
\filldraw (13.5, 6.5) circle(.15);
\filldraw (13.5, 5.5) circle(.15);
\filldraw (15, 6.5) circle(.15);
\filldraw (15, 5.5) circle(.15);
\filldraw (13, 6.5) circle(.15);
\filldraw (13, 5.5) circle(.15);
\filldraw (12.5, 6.5) circle(.15);
\filldraw (12.5, 5.5) circle(.15);
\node[font=\fontsize{25}{6}\selectfont] at (17,6.5) {$\frac {14}8$};

\draw [thick, line width = 1pt] (-7.5, 4) -- (-6.5, 4);
\draw [thick, line width = 1pt] (-5.5, 4) -- (-4.5, 4);
\draw [thin, line width = 0.5pt] (-6.5, 4) -- (-5.5, 4);
\draw [densely dashed, line width = 1pt] (-7.5, 4) -- (-7.5, 2.5);
\draw [thin, line width = 0.5pt] (-7.5, 2.5) -- (-7.5, 1.5);
\draw [thick, line width = 1pt] (-7.5, 2.5) -- (-8, 2.5);
\draw [thick, line width = 1pt] (-7.5, 1.5) -- (-8, 1.5);
\draw [densely dashed, line width = 1pt] (-6.5, 4) -- (-6.5, 2.5);
\draw [thin, line width = 0.5pt] (-6.5, 2.5) -- (-6.5, 1.5);
\draw [thick, line width = 1pt] (-6.5, 2.5) -- (-7, 2.5);
\draw [thick, line width = 1pt] (-6.5, 1.5) -- (-7, 1.5);
\draw [densely dashed, line width = 1pt] (-4.5, 4) -- (-4.5, 2.5);
\draw [thick, line width = 1pt] (-4.5, 2.5) -- (-4.5, 1.5);
\filldraw (-7.5, 4) circle(.15);
\filldraw[fill = red] (-6.5, 4) circle(.15);
\filldraw (-5.5, 4) circle(.15);
\filldraw (-4.5, 4) circle(.15);
\filldraw (-6.5, 2.5) circle(.15);
\filldraw (-6.5, 1.5) circle(.15);
\filldraw (-7, 2.5) circle(.15);
\filldraw (-7, 1.5) circle(.15);
\filldraw (-7.5, 2.5) circle(.15);
\filldraw (-7.5, 1.5) circle(.15);
\filldraw (-8, 2.5) circle(.15);
\filldraw (-8, 1.5) circle(.15);
\filldraw (-4.5, 2.5) circle(.15);
\filldraw (-4.5, 1.5) circle(.15);
\node[font=\fontsize{25}{6}\selectfont] at (-3.5,2.75) {$\frac {14}9$};

\draw [thick, line width = 1pt] (-1, 4) -- (0, 4);
\draw [thick, line width = 1pt] (2, 4) -- (3, 4);
\draw [thin, line width = 0.5pt] (0, 4) -- (1, 4) -- (2,4);
\draw [densely dashed, line width = 1pt] (-1, 4) -- (-1, 2.5);
\draw [thick, line width = 1pt] (-1, 2.5) -- (-1, 1.5);
\draw [densely dashed, line width = 1pt] (0, 4) -- (0, 2.5);
\draw [thin, line width = 0.5pt] (0, 2.5) -- (0, 1.5);
\draw [thick, line width = 1pt] (0, 2.5) -- (-0.5, 2.5);
\draw [thick, line width = 1pt] (0, 1.5) -- (-0.5, 1.5);
\filldraw (-1, 4) circle(.15);
\filldraw[fill = red] (0, 4) circle(.15);
\filldraw[fill = white] (1, 4) circle(.15);
\filldraw (2, 4) circle(.15);
\filldraw (3, 4) circle(.15);
\filldraw (-1, 2.5) circle(.15);
\filldraw (-1, 1.5) circle(.15);
\filldraw (0, 2.5) circle(.15);
\filldraw (0, 1.5) circle(.15);
\filldraw (-0.5, 2.5) circle(.15);
\filldraw (-0.5, 1.5) circle(.15);
\node[font=\fontsize{25}{6}\selectfont] at (4, 2.75) {$\frac {10}7$};

\draw [thick, line width = 1pt] (6, 4) -- (7, 4);
\draw [thick, line width = 1pt] (9, 4) -- (10, 4);
\draw [thin, line width = 0.5pt] (7, 4) -- (8, 4) -- (9,4);
\draw [densely dashed, line width = 1pt] (7, 4) -- (7, 2.5);
\draw [thick, line width = 1pt] (7, 2.5) -- (7, 1.5);
\draw [densely dashed, line width = 1pt] (8, 4) -- (8, 2.5);
\draw [thin, line width = 0.5pt] (8, 2.5) -- (8, 1.5);
\draw [thick, line width = 1pt] (8, 2.5) -- (7.5, 2.5);
\draw [thick, line width = 1pt] (8, 1.5) -- (7.5, 1.5);
\filldraw (6, 4) circle(.15);
\filldraw (7, 4) circle(.15);
\filldraw[fill = red] (8, 4) circle(.15);
\filldraw (9, 4) circle(.15);
\filldraw (10, 4) circle(.15);
\filldraw (7, 2.5) circle(.15);
\filldraw (7, 1.5) circle(.15);
\filldraw (7.5, 2.5) circle(.15);
\filldraw (7.5, 1.5) circle(.15);
\filldraw (8, 2.5) circle(.15);
\filldraw (8, 1.5) circle(.15);
\node[font=\fontsize{25}{6}\selectfont] at (11,2.75) {$\frac {10}7$};

\draw [thick, line width = 1pt] (13, 4) -- (14, 4);
\draw [thick, line width = 1pt] (16, 4) -- (17, 4);
\draw [thin, line width = 0.5pt] (14, 4) -- (15, 4) -- (16,4);
\draw [densely dashed, line width = 1pt] (14, 4) -- (14, 2.5);
\draw [thick, line width = 1pt] (14, 2.5) -- (14, 1.5);
\draw [densely dashed, line width = 1pt] (15, 4) -- (15, 2.5);
\draw [thin, line width = 0.5pt] (15, 2.5) -- (15, 1.5);
\draw [thick, line width = 1pt] (15, 2.5) -- (14.5, 2.5);
\draw [thick, line width = 1pt] (15, 1.5) -- (14.5, 1.5);
\draw [densely dashed, line width = 1pt] (16, 4) -- (16, 2.5);
\draw [thick, line width = 1pt] (16, 2.5) -- (16, 1.5);
\filldraw (13, 4) circle(.15);
\filldraw (14, 4) circle(.15);
\filldraw[fill = red] (15, 4) circle(.15);
\filldraw (16, 4) circle(.15);
\filldraw (17, 4) circle(.15);
\filldraw (15, 2.5) circle(.15);
\filldraw (15, 1.5) circle(.15);
\filldraw (14.5, 2.5) circle(.15);
\filldraw (14.5, 1.5) circle(.15);
\filldraw (16, 2.5) circle(.15);
\filldraw (16, 1.5) circle(.15);
\filldraw (14, 2.5) circle(.15);
\filldraw (14, 1.5) circle(.15);
\node[font=\fontsize{25}{6}\selectfont] at (18, 2.75) {$\frac {12}7$};

\draw [thick, line width = 1pt] (-8, 0) -- (-7, 0);
\draw [thick, line width = 1pt] (-5, 0) -- (-4, 0);
\draw [thin, line width = 0.5pt] (-7, 0) -- (-6, 0) -- (-5,0);
\draw [densely dashed, line width = 1pt] (-7, 0) -- (-7, -1.5);
\draw [thin, line width = 0.5pt] (-7, -1.5) -- (-7, -2.5);
\draw [densely dashed, line width = 1pt] (-6, 0) -- (-6, -1.5);
\draw [thin, line width = 0.5pt] (-5.5, -1.5) -- (-5.5, -2.5);
\draw [thick, line width = 1pt] (-6, -1.5) -- (-5.5, -1.5);
\draw [thick, line width = 1pt] (-6, -2.5) -- (-5.5, -2.5);
\draw [densely dashed, line width = 1pt] (-5, 0) -- (-5, -1.5);
\draw [thick, line width = 1pt] (-5, -1.5) -- (-5, -2.5);
\filldraw (-8, 0) circle(.15);
\filldraw (-7, 0) circle(.15);
\filldraw[fill = red] (-6, 0) circle(.15);
\filldraw (-5, 0) circle(.15);
\filldraw (-4, 0) circle(.15);
\filldraw (-7, -1.5) circle(.15);
\filldraw (-7, -2.5) circle(.15);
\filldraw (-5.5, -1.5) circle(.15);
\filldraw (-5.5, -2.5) circle(.15);
\filldraw (-6, -1.5) circle(.15);
\filldraw (-6, -2.5) circle(.15);
\filldraw (-5, -1.5) circle(.15);
\filldraw (-5, -2.5) circle(.15);
\node[font=\fontsize{25}{6}\selectfont] at (-3, -1.25) {$\frac {12}8$};

\draw [thick, line width = 1pt] (-1, 0) -- (0, 0);
\draw [thick, line width = 1pt] (2, 0) -- (3, 0);
\draw [thin, line width = 0.5pt] (0, 0) -- (1, 0) -- (2,0);
\draw [densely dashed, line width = 1pt] (-1, 0) -- (-1, -1.5);
\draw [thick, line width = 1pt] (-1, -1.5) -- (-1, -2.5);
\draw [thick, line width = 1pt] (-1, -1.5) -- (-1.5, -1.5);
\draw [thick, line width = 1pt] (-1, -2.5) -- (-1.5, -2.5);
\draw [densely dashed, line width = 1pt] (0, 0) -- (0, -1.5);
\draw [thin, line width = 0.5pt] (0, -1.5) -- (0, -2.5);
\draw [thick, line width = 1pt] (0, -1.5) -- (-0.5, -1.5);
\draw [thick, line width = 1pt] (0, -2.5) -- (-0.5, -2.5);
\filldraw (-1, 0) circle(.15);
\filldraw[fill = red] (0, 0) circle(.15);
\filldraw[fill = white] (1, 0) circle(.15);
\filldraw (2, 0) circle(.15);
\filldraw (3, 0) circle(.15);
\filldraw (-1, -1.5) circle(.15);
\filldraw (-1, -2.5) circle(.15);
\filldraw (-1.5, -1.5) circle(.15);
\filldraw (-1.5, -2.5) circle(.15);
\filldraw (0, -1.5) circle(.15);
\filldraw (0, -2.5) circle(.15);
\filldraw (-0.5, -1.5) circle(.15);
\filldraw (-0.5, -2.5) circle(.15);
\node[font=\fontsize{25}{6}\selectfont] at (4,-1.25) {$\frac {12}8$};

\draw [thick, line width = 1pt] (6, 0) -- (7, 0);
\draw [thick, line width = 1pt] (9, 0) -- (10, 0);
\draw [thin, line width = 0.5pt] (7, 0) -- (8, 0) -- (9, 0);
\draw [densely dashed, line width = 1pt] (6, 0) -- (6, -1.5);
\draw [thick, line width = 1pt] (6, -1.5) -- (6, -2.5);
\draw [thick, line width = 1pt] (6, -1.5) -- (5.5, -1.5);
\draw [thick, line width = 1pt] (6, -2.5) -- (5.5, -2.5);
\draw [densely dashed, line width = 1pt] (7, 0) -- (7, -1.5);
\draw [thin, line width = 0.5pt] (7, -1.5) -- (7, -2.5);
\draw [thick, line width = 1pt] (7, -1.5) -- (6.5, -1.5);
\draw [thick, line width = 1pt] (7, -2.5) -- (6.5, -2.5);
\draw [densely dashed, line width = 1pt] (9, 0) -- (9, -1.5);
\draw [thick, line width = 1pt] (9, -1.5) -- (9, -2.5);
\filldraw (6, 0) circle(.15);
\filldraw[fill = red] (7, 0) circle(.15);
\filldraw[fill = white] (8, 0) circle(.15);
\filldraw (9, 0) circle(.15);
\filldraw (10, 0) circle(.15);
\filldraw (6, -1.5) circle(.15);
\filldraw (6, -2.5) circle(.15);
\filldraw (5.5, -1.5) circle(.15);
\filldraw (5.5, -2.5) circle(.15);
\filldraw (7, -1.5) circle(.15);
\filldraw (7, -2.5) circle(.15);
\filldraw (6.5, -1.5) circle(.15);
\filldraw (6.5, -2.5) circle(.15);
\filldraw (9, -1.5) circle(.15);
\filldraw (9, -2.5) circle(.15);
\node[font=\fontsize{25}{6}\selectfont] at (11,-1.25) {$\frac {14}9$};

\draw [thick, line width = 1pt] (13.5, 0) -- (14.5, 0);
\draw [thick, line width = 1pt] (16.5, 0) -- (17.5, 0);
\draw [thin, line width = 0.5pt] (14.5, 0) -- (15.5, 0) -- (16.5, 0);
\draw [densely dashed, line width = 1pt] (13.5, 0) -- (12.5, -1.5);
\draw [thick, line width = 1pt] (12.5, -1.5) -- (12.5, -2.5);
\draw [densely dashed, line width = 1pt] (14.5, 0) -- (13.5, -1.5);
\draw [thin, line width = 0.5pt] (13.5, -1.5) -- (13.5, -2.5);
\draw [thick, line width = 1pt] (13.5, -1.5) -- (13, -1.5);
\draw [thick, line width = 1pt] (13.5, -2.5) -- (13, -2.5);
\draw [densely dashed, line width = 1pt] (14.5, 0) -- (14, -1.5);
\draw [thin, line width = 0.5pt] (14, -1.5) -- (14, -2.5);
\draw [thick, line width = 1pt] (14, -1.5) -- (14.5, -1.5);
\draw [thick, line width = 1pt] (14, -2.5) -- (14.5, -2.5);
\draw [densely dashed, line width = 1pt] (15.5, 0) -- (15.5, -1.5);
\draw [thin, line width = 0.5pt] (15.5, -1.5) -- (15.5, -2.5);
\draw [thick, line width = 1pt] (15.5, -1.5) -- (15, -1.5);
\draw [thick, line width = 1pt] (15.5, -2.5) -- (15, -2.5);
\draw [densely dashed, line width = 1pt] (16.5, 0) -- (16.5, -1.5);
\draw [thin, line width = 0.5pt] (16.5, -1.5) -- (16.5, -2.5);
\draw [thick, line width = 1pt] (16.5, -1.5) -- (16, -1.5);
\draw [thick, line width = 1pt] (16.5, -2.5) -- (16, -2.5);
\draw [densely dashed, line width = 1pt] (17.5, 0) -- (17.5, -1.5);
\draw [thin, line width = 0.5pt] (17.5, -1.5) -- (17.5, -2.5);
\draw [thick, line width = 1pt] (17.5, -1.5) -- (17, -1.5);
\draw [thick, line width = 1pt] (17.5, -2.5) -- (17, -2.5);
\filldraw (13.5, 0) circle(.15);
\filldraw[fill = yellow] (14.5, 0) circle(.15);
\filldraw[fill = white] (15.5, 0) circle(.15);
\filldraw[fill = red] (16.5, 0) circle(.15);
\filldraw (17.5, 0) circle(.15);
\filldraw (13.5, -1.5) circle(.15);
\filldraw (13.5, -2.5) circle(.15);
\filldraw (13, -1.5) circle(.15);
\filldraw (13, -2.5) circle(.15);
\filldraw (14.5, -1.5) circle(.15);
\filldraw (14.5, -2.5) circle(.15);
\filldraw (14, -1.5) circle(.15);
\filldraw (14, -2.5) circle(.15);
\filldraw (15, -1.5) circle(.15);
\filldraw (15, -2.5) circle(.15);
\filldraw (16, -1.5) circle(.15);
\filldraw (16, -2.5) circle(.15);
\filldraw (15.5, -1.5) circle(.15);
\filldraw (15.5, -2.5) circle(.15);
\filldraw (16.5, -1.5) circle(.15);
\filldraw (16.5, -2.5) circle(.15);
\filldraw (17, -1.5) circle(.15);
\filldraw (17, -2.5) circle(.15);
\filldraw (17.5, -1.5) circle(.15);
\filldraw (17.5, -2.5) circle(.15);
\filldraw (12.5, -1.5) circle(.15);
\filldraw (12.5, -2.5) circle(.15);
\node[font=\fontsize{25}{6}\selectfont] at (18.5,-1.25) {$\frac {26}{16}$};

\draw [thick, line width = 1pt] (-8, -4) -- (-7, -4);
\draw [thick, line width = 1pt] (-5, -4) -- (-4, -4);
\draw [thin, line width = 0.5pt] (-7, -4) -- (-6, -4) -- (-5, -4);
\draw [densely dashed, line width = 1pt] (-8, -4) -- (-8, -5.5);
\draw [thick, line width = 1pt] (-8, -5.5) -- (-8, -6.5);
\draw [densely dashed, line width = 1pt] (-7, -4) -- (-7, -5.5);
\draw [thin, line width = 0.5pt] (-7, -5.5) -- (-7, -6.5);
\draw [thick, line width = 1pt] (-7.5, -5.5) -- (-7, -5.5);
\draw [thick, line width = 1pt] (-7.5, -6.5) -- (-7, -6.5);
\draw [densely dashed, line width = 1pt] (-6, -4) -- (-6, -5.5);
\draw [thin, line width = 0.5pt] (-6, -5.5) -- (-6, -6.5);
\draw [thick, line width = 1pt] (-6, -5.5) -- (-6.5, -5.5);
\draw [thick, line width = 1pt] (-6, -6.5) -- (-6.5, -6.5);
\draw [densely dashed, line width = 1pt] (-5, -4) -- (-5, -5.5);
\draw [thin, line width = 0.5pt] (-5, -5.5) -- (-5, -6.5);
\draw [thick, line width = 1pt] (-5, -5.5) -- (-5.5, -5.5);
\draw [thick, line width = 1pt] (-5, -6.5) -- (-5.5, -6.5);
\draw [densely dashed, line width = 1pt] (-5, -4) -- (-4.5, -5.5);
\draw [thin, line width = 0.5pt] (-4.5, -5.5) -- (-4.5, -6.5);
\draw [thick, line width = 1pt] (-4.5, -5.5) -- (-4, -5.5);
\draw [thick, line width = 1pt] (-4.5, -6.5) -- (-4, -6.5);
\draw [densely dashed, line width = 1pt] (-4, -4) -- (-3.5, -5.5);
\draw [thin, line width = 0.5pt] (-3.5, -5.5) -- (-3.5, -6.5);
\draw [thick, line width = 1pt] (-3.5, -5.5) -- (-3, -5.5);
\draw [thick, line width = 1pt] (-3.5, -6.5) -- (-3, -6.5);
\filldraw (-8, -4) circle(.15);
\filldraw[fill = red] (-7, -4) circle(.15);
\filldraw[fill = white] (-6, -4) circle(.15);
\filldraw[fill = yellow] (-5, -4) circle(.15);
\filldraw (-4, -4) circle(.15);
\filldraw (-8, -5.5) circle(.15);
\filldraw (-8, -6.5) circle(.15);
\filldraw (-7, -5.5) circle(.15);
\filldraw (-7, -6.5) circle(.15);
\filldraw (-7.5, -5.5) circle(.15);
\filldraw (-7.5, -6.5) circle(.15);
\filldraw (-6.5, -5.5) circle(.15);
\filldraw (-6.5, -6.5) circle(.15);
\filldraw (-5.5, -5.5) circle(.15);
\filldraw (-5.5, -6.5) circle(.15);
\filldraw (-6, -5.5) circle(.15);
\filldraw (-6, -6.5) circle(.15);
\filldraw (-5, -5.5) circle(.15);
\filldraw (-5, -6.5) circle(.15);
\filldraw (-4.5, -5.5) circle(.15);
\filldraw (-4.5, -6.5) circle(.15);
\filldraw (-3.5, -5.5) circle(.15);
\filldraw (-3.5, -6.5) circle(.15);
\filldraw (-3, -5.5) circle(.15);
\filldraw (-3, -6.5) circle(.15);
\filldraw (-4, -5.5) circle(.15);
\filldraw (-4, -6.5) circle(.15);
\node[font=\fontsize{25}{6}\selectfont] at (-2.3,-5.25) {$\frac {26}{16}$};

\draw [thick, line width = 1pt] (0, -4) -- (1, -4);
\draw [thick, line width = 1pt] (3, -4) -- (4, -4);
\draw [thin, line width = 0.5pt] (1, -4) -- (2, -4) -- (3,-4);
\draw [densely dashed, line width = 1pt] (0, -4) -- (-0.5, -5.5);
\draw [thick, line width = 1pt] (-0.5, -5.5) -- (-0.5, -6.5);
\draw [thick, line width = 1pt] (-0.5, -5.5) -- (-1, -5.5);
\draw [thick, line width = 1pt] (-0.5, -6.5) -- (-1, -6.5);
\draw [densely dashed, line width = 1pt] (1, -4) -- (0.5, -5.5);
\draw [thin, line width = 0.5pt] (0.5, -5.5) -- (0.5, -6.5);
\draw [thick, line width = 1pt] (0.5, -5.5) -- (0, -5.5);
\draw [thick, line width = 1pt] (0.5, -6.5) -- (0, -6.5);
\draw [densely dashed, line width = 1pt] (1, -4) -- (1.5, -5.5);
\draw [thin, line width = 0.5pt] (1.5, -5.5) -- (1.5, -6.5);
\draw [thick, line width = 1pt] (1.5, -5.5) -- (1, -5.5);
\draw [thick, line width = 1pt] (1.5, -6.5) -- (1, -6.5);
\draw [densely dashed, line width = 1pt] (2, -4) -- (2, -5.5);
\draw [thick, line width = 1pt] (2, -5.5) -- (2, -6.5);
\draw [densely dashed, line width = 1pt] (3, -4) -- (3, -5.5);
\draw [thin, line width = 0.5pt] (3, -5.5) -- (3, -6.5);
\draw [thick, line width = 1pt] (3, -5.5) -- (2.5, -5.5);
\draw [thick, line width = 1pt] (3, -6.5) -- (2.5, -6.5);
\draw [densely dashed, line width = 1pt] (4, -4) -- (4, -5.5);
\draw [thin, line width = 0.5pt] (4, -5.5) -- (4, -6.5);
\draw [thick, line width = 1pt] (4, -5.5) -- (3.5, -5.5);
\draw [thick, line width = 1pt] (4, -6.5) -- (3.5, -6.5);
\filldraw (0, -4) circle(.15);
\filldraw[fill = yellow] (1, -4) circle(.15);
\filldraw[fill = white] (2, -4) circle(.15);
\filldraw[fill = red] (3, -4) circle(.15);
\filldraw (4, -4) circle(.15);
\filldraw (0, -5.5) circle(.15);
\filldraw (0, -6.5) circle(.15);
\filldraw (-0.5, -5.5) circle(.15);
\filldraw (-0.5, -6.5) circle(.15);
\filldraw (1, -5.5) circle(.15);
\filldraw (1, -6.5) circle(.15);
\filldraw (0.5, -5.5) circle(.15);
\filldraw (0.5, -6.5) circle(.15);
\filldraw (1.5, -5.5) circle(.15);
\filldraw (1.5, -6.5) circle(.15);
\filldraw (2.5, -5.5) circle(.15);
\filldraw (2.5, -6.5) circle(.15);
\filldraw (2, -5.5) circle(.15);
\filldraw (2, -6.5) circle(.15);
\filldraw (3, -5.5) circle(.15);
\filldraw (3, -6.5) circle(.15);
\filldraw (3.5, -5.5) circle(.15);
\filldraw (3.5, -6.5) circle(.15);
\filldraw (4, -5.5) circle(.15);
\filldraw (4, -6.5) circle(.15);
\filldraw (-1, -5.5) circle(.15);
\filldraw (-1, -6.5) circle(.15);
\node[font=\fontsize{25}{6}\selectfont] at (5.3, -5.25) {$\frac {26}{16}$};

\draw [thick, line width = 1pt] (8.5, -4) -- (9.5, -4);
\draw [thick, line width = 1pt] (11.5, -4) -- (12.5, -4);
\draw [thin, line width = 0.5pt] (9.5, -4) -- (10.5, -4) -- (11.5, -4);
\draw [densely dashed, line width = 1pt] (8.5, -4) -- (7.5, -5.5);
\draw [thin, line width = 0.5pt] (7.5, -5.5) -- (7.5, -6.5);
\draw [thick, line width = 1pt] (7.5, -5.5) -- (7, -5.5);
\draw [thick, line width = 1pt] (7.5, -6.5) -- (7, -6.5);
\draw [densely dashed, line width = 1pt] (9.5, -4) -- (8.5, -5.5);
\draw [thin, line width = 0.5pt] (8.5, -5.5) -- (8.5, -6.5);
\draw [thick, line width = 1pt] (8.5, -5.5) -- (8, -5.5);
\draw [thick, line width = 1pt] (8.5, -6.5) -- (8, -6.5);
\draw [densely dashed, line width = 1pt] (9.5, -4) -- (9, -5.5);
\draw [thin, line width = 0.5pt] (9, -5.5) -- (9, -6.5);
\draw [thick, line width = 1pt] (9, -5.5) -- (9.5, -5.5);
\draw [thick, line width = 1pt] (9, -6.5) -- (9.5, -6.5);
\draw [densely dashed, line width = 1pt] (10.5, -4) -- (10.5, -5.5);
\draw [thin, line width = 0.5pt] (10.5, -5.5) -- (10.5, -6.5);
\draw [thick, line width = 1pt] (10.5, -5.5) -- (10, -5.5);
\draw [thick, line width = 1pt] (10.5, -6.5) -- (10, -6.5);
\draw [densely dashed, line width = 1pt] (11.5, -4) -- (11.5, -5.5);
\draw [thin, line width = 0.5pt] (11.5, -5.5) -- (11.5, -6.5);
\draw [thick, line width = 1pt] (11.5, -5.5) -- (11, -5.5);
\draw [thick, line width = 1pt] (11.5, -6.5) -- (11, -6.5);
\draw [densely dashed, line width = 1pt] (12.5, -4) -- (12.5, -5.5);
\draw [thin, line width = 0.5pt] (12.5, -5.5) -- (12.5, -6.5);
\draw [thick, line width = 1pt] (12.5, -5.5) -- (12, -5.5);
\draw [thick, line width = 1pt] (12.5, -6.5) -- (12, -6.5);
\filldraw (8.5, -4) circle(.15);
\filldraw[fill = yellow] (9.5, -4) circle(.15);
\filldraw[fill = white] (10.5, -4) circle(.15);
\filldraw[fill = red] (11.5, -4) circle(.15);
\filldraw (12.5, -4) circle(.15);
\filldraw (7.5, -5.5) circle(.15);
\filldraw (7.5, -6.5) circle(.15);
\filldraw (8.5, -5.5) circle(.15);
\filldraw (8.5, -6.5) circle(.15);
\filldraw (9.5, -5.5) circle(.15);
\filldraw (9.5, -6.5) circle(.15);
\filldraw (9, -5.5) circle(.15);
\filldraw (9, -6.5) circle(.15);
\filldraw (10, -5.5) circle(.15);
\filldraw (10, -6.5) circle(.15);
\filldraw (11, -5.5) circle(.15);
\filldraw (11, -6.5) circle(.15);
\filldraw (10.5, -5.5) circle(.15);
\filldraw (10.5, -6.5) circle(.15);
\filldraw (11.5, -5.5) circle(.15);
\filldraw (11.5, -6.5) circle(.15);
\filldraw (12, -5.5) circle(.15);
\filldraw (12, -6.5) circle(.15);
\filldraw (12.5, -5.5) circle(.15);
\filldraw (12.5, -6.5) circle(.15);
\filldraw (8, -5.5) circle(.15);
\filldraw (8, -6.5) circle(.15);
\filldraw (7, -5.5) circle(.15);
\filldraw (7, -6.5) circle(.15);
\node[font=\fontsize{25}{6}\selectfont] at (13.7,-5.25) {$\frac {28}{17}$};
\end{tikzpicture}
\end{minipage}}
\end{center}
\captionsetup{width=1.0\linewidth}
\caption{All possible structures of a responsible components of $H + C$.
	The filled (white, respectively) vertices are in (not in, respectively) $V(M)$. 
	Moreover, the yellow vertices are $2$-anchors while the red vertices are responsible $1$-anchors.\label{fig04}}
\end{figure}
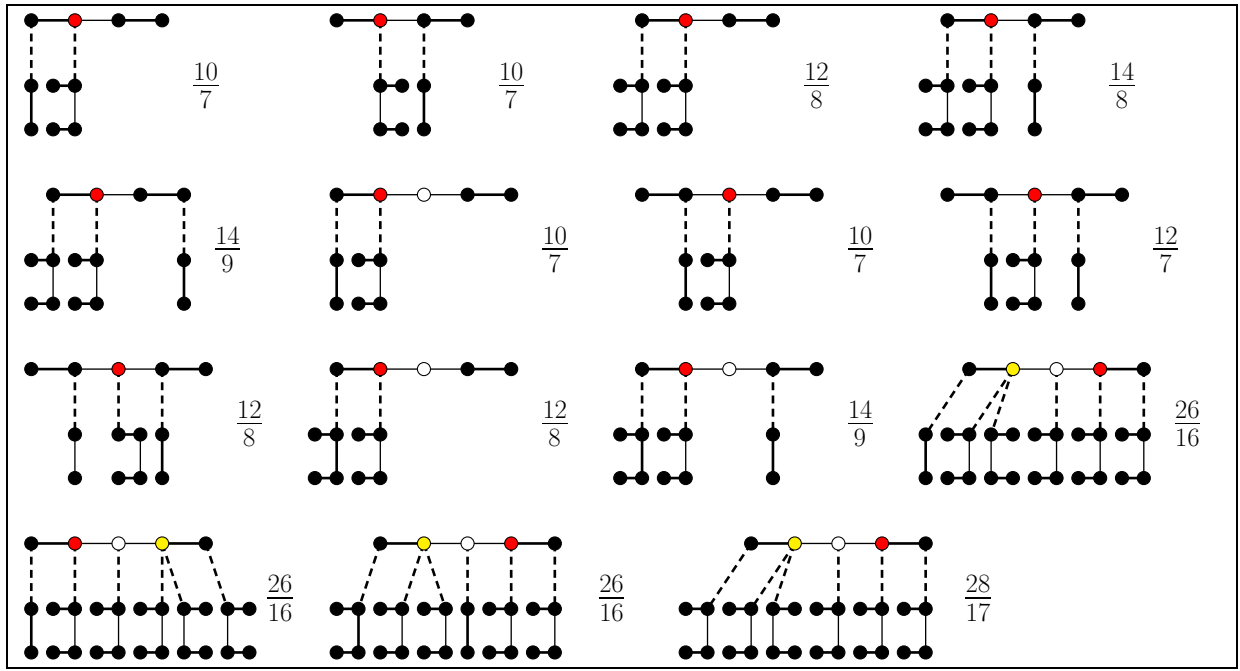

\subsection{Three operations}\label{subsec:op}
In this subsection, we design three operations for reducing the number of critical components of $H+C$. 
We remind the reader that the three operations are applied to a connected component $K$ of $H+C$ rather than its trunk $\widetilde{K}$.
But we still use $\widetilde{K}$ to compute a feasible solution of $K$; 
Lemma~\ref{lemma10} implies that computing a feasible solution takes $O(1)$ time.

\begin{notation}
In the remainder of this subsection, $K_1$ denotes a critical component of $H+C$, $K_2$ denotes a component of $H+C$ (possibly $K_1 = K_2$),
$S_1$ denotes a critical satellite-element of $K_1$, and $\{v_1,v_2\}$ denotes an edge in $E(G)\setminus C$ with $v_1 \in V(S_1)$ and $v_2 \in V(K_2)$. 
Moreover, 
\begin{itemize}
\parskip=0pt
\item $n_0$ denotes the total number of $0$-anchors in $H+C$; 
\item $n_c$ denotes the total number of critical components in $H+C$; 
\item $n_{cc}$ denotes the total number of connected components of $H+C$; 
\item and $g = n_0 + 5n_c - 6n_{cc}$.
\end{itemize}
\end{notation}

Each of the three operations aims to reduce $g$ by at least~$1$.

\begin{operation}
\label{op01}
Suppose that $v_2$ is a $0$-anchor of $K_2$ or a non-responsible $1$-anchor.
Then, the operation moves $S_1$ to $v_2$ by replacing the rescue-edge of $S_1$ with $\{v_1, v_2\}$ (see for an illustration in Figure~\ref{fig05}).
\end{operation}

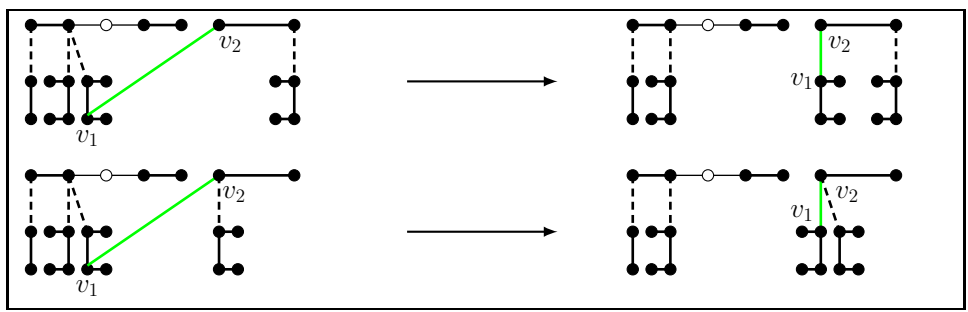
\begin{figure}[thb]
\begin{center}
\framebox{
\begin{minipage}{0.75\textwidth}
\begin{tikzpicture}[scale=0.5,transform shape]
\draw [thick, line width = 1pt] (-8, -6) -- (-7, -6);
\draw [thick, line width = 1pt] (-5, -6) -- (-4, -6);
\draw [thin, line width = 0.5pt] (-7, -6) -- (-6, -6);
\draw [thin, line width = 0.5pt] (-6, -6) -- (-5, -6);
\draw [densely dashed, line width = 1pt] (-8, -6) -- (-8,-7.5);
\draw [thick, line width = 1pt] (-8, -7.5) -- (-8, -8.5);
\draw [densely dashed, line width = 1pt] (-7, -6) -- (-7, -7.5);
\draw [thick, line width = 1pt] (-7, -7.5) -- (-7.5, -7.5);
\draw [thick, line width = 1pt] (-7, -7.5) -- (-7, -8.5);
\draw [thick, line width = 1pt] (-7, -8.5) -- (-7.5, -8.5);
\draw [densely dashed, line width = 1pt] (-7, -6) -- (-6.5,-7.5);
\draw [thick, line width = 1pt] (-6.5,-7.5) -- (-6, -7.5);
\draw [thick, line width = 1pt] (-6.5,-8.5) -- (-6, -8.5);
\draw [thick, line width = 1pt] (-6.5, -7.5) -- (-6.5, -8.5);
\filldraw (-8, -6) circle(.15);
\filldraw (-7, -6) circle(.15);
\filldraw[fill = white] (-6, -6) circle(.15);
\filldraw (-5, -6) circle(.15);
\filldraw (-4, -6) circle(.15);
\filldraw (-8,-7.5) circle(.15);
\filldraw (-8,-8.5) circle(.15);
\filldraw (-7,-7.5) circle(.15);
\filldraw (-7,-8.5) circle(.15);
\filldraw (-7.5,-7.5) circle(.15);
\filldraw (-7.5,-8.5) circle(.15);
\filldraw (-6.5,-7.5) circle(.15);
\filldraw (-6.5,-8.5) circle(.15);
\filldraw (-6,-7.5) circle(.15);
\filldraw (-6,-8.5) circle(.15);
\node[font=\fontsize{20}{6}\selectfont] at (-6.5, -9) {$v_1$};

\draw [thick, line width = 1pt, color=green] (-6.5, -8.4) -- (-3, -6);

\node[font=\fontsize{20}{6}\selectfont] at (-2.7, -6.5) {$v_2$};
\draw [thick, line width = 1pt] (-3, -6) -- (-1, -6);
\draw [densely dashed, line width = 1pt] (-1, -6) -- (-1,-7.5);
\draw [thick, line width = 1pt] (-1, -7.5) -- (-1, -8.5);
\draw [thick, line width = 1pt] (-1, -7.5) -- (-1.5, -7.5);
\draw [thick, line width = 1pt] (-1, -8.5) -- (-1.5, -8.5);
\filldraw (-3, -6) circle(.15);
\filldraw (-1, -6) circle(.15);
\filldraw (-1,-7.5) circle(.15);
\filldraw (-1,-8.5) circle(.15);
\filldraw (-1.5,-7.5) circle(.15);
\filldraw (-1.5,-8.5) circle(.15);

\draw [-latex, thick] (2, -7.5) to (6, -7.5);

\draw [thick, line width = 1pt] (8, -6) -- (9, -6);
\draw [thick, line width = 1pt] (11, -6) -- (12,-6);
\draw [thin, line width = 0.5pt] (9, -6) -- (10, -6);
\draw [thin, line width = 0.5pt] (10, -6) -- (11, -6);
\draw [densely dashed, line width = 1pt] (8, -6) -- (8, -7.5);
\draw [thick, line width = 1pt] (8, -7.5) -- (8, -8.5);
\draw [densely dashed, line width = 1pt] (9, -6) -- (9, -7.5);
\draw [thick, line width = 1pt] (9, -7.5) -- (9, -8.5);
\draw [thick, line width = 1pt] (9, -7.5) -- (8.5, -7.5);
\draw [thick, line width = 1pt] (9, -8.5) -- (8.5, -8.5);
\filldraw (8, -6) circle(.15);
\filldraw (9, -6) circle(.15);
\filldraw[fill = white] (10, -6) circle(.15);
\filldraw (11, -6) circle(.15);
\filldraw (12, -6) circle(.15);
\filldraw (8, -7.5) circle(.15);
\filldraw (8, -8.5) circle(.15);
\filldraw (8.5, -7.5) circle(.15);
\filldraw (8.5, -8.5) circle(.15);
\filldraw (9, -7.5) circle(.15);
\filldraw (9, -8.5) circle(.15);

\node[font=\fontsize{20}{6}\selectfont] at (13.5, -6.5) {$v_2$};
\node[font=\fontsize{20}{6}\selectfont] at (12.5, -7.5) {$v_1$};

\draw [thick, line width = 1pt] (13, -6) -- (15, -6);
\draw [thick, line width = 1pt, color=green] (13, -6) -- (13,-7.5);
\draw [thick, line width = 1pt] (13, -7.5) -- (13, -8.5);
\draw [thick, line width = 1pt] (13, -7.5) -- (13.5, -7.5);
\draw [thick, line width = 1pt] (13, -8.5) -- (13.5, -8.5);
\draw [densely dashed, line width = 1pt] (15, -6) -- (15,-7.5);
\draw [thick, line width = 1pt] (15, -7.5) -- (15, -8.5);
\draw [thick, line width = 1pt] (15, -8.5) -- (14.5, -8.5);
\draw [thick, line width = 1pt] (15, -7.5) -- (14.5, -7.5);
\filldraw (13, -6) circle(.15);
\filldraw (15, -6) circle(.15);
\filldraw (13,-7.5) circle(.15);
\filldraw (13,-8.5) circle(.15);
\filldraw (13.5,-7.5) circle(.15);
\filldraw (13.5,-8.5) circle(.15);
\filldraw (14.5,-7.5) circle(.15);
\filldraw (14.5,-8.5) circle(.15);
\filldraw (15,-7.5) circle(.15);
\filldraw (15,-8.5) circle(.15);

\draw [thick, line width = 1pt] (-8, -10) -- (-7, -10);
\draw [thick, line width = 1pt] (-5, -10) -- (-4, -10);
\draw [thin, line width = 0.5pt] (-7, -10) -- (-6, -10);
\draw [thin, line width = 0.5pt] (-6, -10) -- (-5, -10);
\draw [densely dashed, line width = 1pt] (-8, -10) -- (-8,-11.5);
\draw [thick, line width = 1pt] (-8, -11.5) -- (-8, -12.5);
\draw [densely dashed, line width = 1pt] (-7, -10) -- (-7, -11.5);
\draw [thick, line width = 1pt] (-7, -11.5) -- (-7, -12.5);
\draw [thick, line width = 1pt] (-7, -11.5) -- (-7.5, -11.5);
\draw [thick, line width = 1pt] (-7, -12.5) -- (-7.5, -12.5);
\draw [densely dashed, line width = 1pt] (-7, -10) -- (-6.5,-11.5);
\draw [thick, line width = 1pt] (-6.5, -11.5) -- (-6.5, -12.5);
\draw [thick, line width = 1pt] (-6.5, -11.5) -- (-6, -11.5);
\draw [thick, line width = 1pt] (-6.5, -12.5) -- (-6, -12.5);
\filldraw (-8, -10) circle(.15);
\filldraw (-7, -10) circle(.15);
\filldraw[fill = white] (-6, -10) circle(.15);
\filldraw (-5, -10) circle(.15);
\filldraw (-4, -10) circle(.15);
\filldraw (-8,-11.5) circle(.15);
\filldraw (-8,-12.5) circle(.15);
\filldraw (-7,-11.5) circle(.15);
\filldraw (-7,-12.5) circle(.15);
\filldraw (-7.5,-11.5) circle(.15);
\filldraw (-7.5,-12.5) circle(.15);
\filldraw (-6.5,-11.5) circle(.15);
\filldraw (-6.5,-12.5) circle(.15);
\filldraw (-6,-11.5) circle(.15);
\filldraw (-6,-12.5) circle(.15);

\node[font=\fontsize{20}{6}\selectfont] at (-6.5, -13) {$v_1$};

\draw [thick, line width = 1pt, color=green] (-6.5, -12.4) -- (-3, -10);

\node[font=\fontsize{20}{6}\selectfont] at (-2.6, -10.5) {$v_2$};
\draw [thick, line width = 1pt] (-3, -10) -- (-1, -10);
\draw [densely dashed, line width = 1pt] (-3, -10) -- (-3,-11.5);
\draw [thick, line width = 1pt] (-3, -11.5) -- (-3, -12.5);
\draw [thick, line width = 1pt] (-3, -11.5) -- (-2.5, -11.5);
\draw [thick, line width = 1pt] (-3, -12.5) -- (-2.5, -12.5);
\filldraw (-3, -10) circle(.15);
\filldraw (-1, -10) circle(.15);
\filldraw (-3,-11.5) circle(.15);
\filldraw (-3,-12.5) circle(.15);
\filldraw (-2.5,-11.5) circle(.15);
\filldraw (-2.5,-12.5) circle(.15);

\draw [-latex, thick] (2, -11.5) to (6, -11.5);

\draw [thick, line width = 1pt] (8, -10) -- (9, -10);
\draw [thick, line width = 1pt] (11, -10) -- (12,-10);
\draw [thin, line width = 0.5pt] (9, -10) -- (10, -10);
\draw [thin, line width = 0.5pt] (10, -10) -- (11, -10);
\draw [densely dashed, line width = 1pt] (8, -10) -- (8, -11.5);
\draw [thick, line width = 1pt] (8, -11.5) -- (8, -12.5);
\draw [densely dashed, line width = 1pt] (9, -10) -- (9, -11.5);
\draw [thick, line width = 1pt] (9, -11.5) -- (9, -12.5);
\draw [thick, line width = 1pt] (9, -11.5) -- (8.5, -11.5);
\draw [thick, line width = 1pt] (9, -12.5) -- (8.5, -12.5);
\filldraw (8, -10) circle(.15);
\filldraw (9, -10) circle(.15);
\filldraw[fill = white] (10, -10) circle(.15);
\filldraw (11, -10) circle(.15);
\filldraw (12, -10) circle(.15);
\filldraw (8, -11.5) circle(.15);
\filldraw (8, -12.5) circle(.15);
\filldraw (9, -11.5) circle(.15);
\filldraw (9, -12.5) circle(.15);
\filldraw (8.5, -11.5) circle(.15);
\filldraw (8.5, -12.5) circle(.15);

\node[font=\fontsize{20}{6}\selectfont] at (13.7, -10.5) {$v_2$};
\node[font=\fontsize{20}{6}\selectfont] at (12.5, -11) {$v_1$};

\draw [thick, line width = 1pt] (13, -10) -- (15, -10);
\draw [thick, line width = 1pt, color=green] (13, -10) -- (13,-11.5);
\draw [thick, line width = 1pt] (13, -11.5) -- (13, -12.5);
\draw [thick, line width = 1pt] (13, -11.5) -- (12.5, -11.5);
\draw [thick, line width = 1pt] (13, -12.5) -- (12.5, -12.5);
\draw [densely dashed, line width = 1pt] (13, -10) -- (13.5,-11.5);
\draw [thick, line width = 1pt] (13.5, -11.5) -- (13.5, -12.5);
\draw [thick, line width = 1pt] (13.5, -11.5) -- (14, -11.5);
\draw [thick, line width = 1pt] (13.5, -12.5) -- (14, -12.5);
\filldraw (13, -10) circle(.15);
\filldraw (15, -10) circle(.15);
\filldraw (12.5,-11.5) circle(.15);
\filldraw (12.5,-12.5) circle(.15);
\filldraw (13,-11.5) circle(.15);
\filldraw (13,-12.5) circle(.15);
\filldraw (13.5,-11.5) circle(.15);
\filldraw (13.5,-12.5) circle(.15);
\filldraw (14,-11.5) circle(.15);
\filldraw (14,-12.5) circle(.15);
\end{tikzpicture}
\end{minipage}}
\end{center}
\captionsetup{width=1.0\linewidth}
\caption{Illustration of two representative possible cases in Operation~\ref{op01}, 
	in which $v_2$ is a $0$-anchor and a $1$-anchor, respectively, and the green edge is the edge $\{v_1, v_2\}$.\label{fig05}}
\end{figure}

\begin{lemma}\label{lemma23} 
Operation~1 reduces the value of $g$ by at least~$1$.
\end{lemma}
\begin{proof}
Clearly, Operation~1 never changes $n_{cc}$ and never increases $n_0$. 
For convenience, let $K'_i$ denote the modified $K_i$ after the operation for each $i\in\{1,2\}$. 
If $K_1 = K_2$, then by Statement~3 of Fact~\ref{fact02}, $n_c$ is reduced by $1$, and in turn  $g$ is reduced by at least~$5$. 
So, we may assume that $K_1 \ne K_2$. 
Then, by Statement~2 of Fact~\ref{fact02}, $K'_1$ is not critical, and hence $n_c$ does not increase. 
Now, if $v_2$ is a $0$-anchor in $K_2$, then because $v_2$ becomes a $1$-anchor in $K'_2$, $n_0$ decreases by~$1$ and hence $g$ decreases by at least~$1$. 
Otherwise, $v_2$ is a non-responsible $1$-anchor and hence the operation reduces $n_c$ by~$1$ but does not change $n_0$, implying that $g$ decreases by at least~$5$.
\end{proof}

\begin{operation}   
\label{op02}
Suppose that the center element of $K_2$ is an edge, star, or bi-star, $K_2$ has exactly one satellite-element $S_2$ and $v_2$ is in $S_2$. 
Then, the operation moves $S_1$ to $v_2$ by replacing the rescue-edge of $S_1$ with $\{v_1, v_2\}$, and updates $S_2$ to be the center element of the new $K_2$
(see for an illustration in Figure~\ref{fig06}).
\end{operation}

\begin{figure}[thb]
\begin{center}
\framebox{
\begin{minipage}{0.75\textwidth}
\begin{tikzpicture}[scale=0.5,transform shape]
\draw [thick, line width = 1pt] (-8, -10) -- (-7, -10);
\draw [thick, line width = 1pt] (-5, -10) -- (-4, -10);
\draw [thin, line width = 0.5pt] (-7, -10) -- (-6, -10);
\draw [thin, line width = 0.5pt] (-6, -10) -- (-5, -10);
\draw [densely dashed, line width = 1pt] (-8, -10) -- (-8,-11.5);
\draw [thick, line width = 1pt] (-8, -11.5) -- (-8, -12.5);
\draw [densely dashed, line width = 1pt] (-7, -10) -- (-7, -11.5);
\draw [thick, line width = 1pt] (-7, -11.5) -- (-7, -12.5);
\draw [thick, line width = 1pt] (-7, -11.5) -- (-7.5, -11.5);
\draw [thick, line width = 1pt] (-7, -12.5) -- (-7.5, -12.5);
\draw [densely dashed, line width = 1pt] (-7, -10) -- (-6.5,-11.5);
\draw [thick, line width = 1pt] (-6.5, -11.5) -- (-6.5, -12.5);
\draw [thick, line width = 1pt] (-6.5, -11.5) -- (-6, -11.5);
\draw [thick, line width = 1pt] (-6.5, -12.5) -- (-6, -12.5);
\filldraw (-8, -10) circle(.15);
\filldraw (-7, -10) circle(.15);
\filldraw[fill = white] (-6, -10) circle(.15);
\filldraw (-5, -10) circle(.15);
\filldraw (-4, -10) circle(.15);
\filldraw (-8,-11.5) circle(.15);
\filldraw (-8,-12.5) circle(.15);
\filldraw (-7,-11.5) circle(.15);
\filldraw (-7,-12.5) circle(.15);
\filldraw (-7.5,-11.5) circle(.15);
\filldraw (-7.5,-12.5) circle(.15);
\filldraw (-6.5,-11.5) circle(.15);
\filldraw (-6.5,-12.5) circle(.15);
\filldraw (-6,-11.5) circle(.15);
\filldraw (-6,-12.5) circle(.15);

\node[font=\fontsize{20}{6}\selectfont] at (-3.5, -11) {$v_2$};
\node[font=\fontsize{20}{6}\selectfont] at (-6.5, -13) {$v_1$};

\draw [thick, line width = 1pt, color=green] (-6.5, -12.4) -- (-3, -11.5);

\draw [thick, line width = 1pt] (-3, -10) -- (-1, -10);
\draw [densely dashed, line width = 1pt] (-3, -10) -- (-3,-11.5);
\draw [thick, line width = 1pt] (-3, -11.5) -- (-3, -12.5);
\draw [thick, line width = 1pt] (-3, -11.5) -- (-2.5, -11.5);
\draw [thick, line width = 1pt] (-3, -12.5) -- (-2.5, -12.5);
\filldraw (-3, -10) circle(.15);
\filldraw (-1, -10) circle(.15);
\filldraw (-3,-11.5) circle(.15);
\filldraw (-3,-12.5) circle(.15);
\filldraw (-2.5,-11.5) circle(.15);
\filldraw (-2.5,-12.5) circle(.15);

\draw [-latex, thick] (2, -11.5) to (6, -11.5);

\draw [thick, line width = 1pt] (8, -10) -- (9, -10);
\draw [thick, line width = 1pt] (11, -10) -- (12,-10);
\draw [thin, line width = 0.5pt] (9, -10) -- (10, -10);
\draw [thin, line width = 0.5pt] (10, -10) -- (11, -10);
\draw [densely dashed, line width = 1pt] (8, -10) -- (8, -11.5);
\draw [thick, line width = 1pt] (8, -11.5) -- (8, -12.5);
\draw [densely dashed, line width = 1pt] (9, -10) -- (9, -11.5);
\draw [thick, line width = 1pt] (9, -11.5) -- (9, -12.5);
\draw [thick, line width = 1pt] (9, -11.5) -- (8.5, -11.5);
\draw [thick, line width = 1pt] (9, -12.5) -- (8.5, -12.5);
\filldraw (8, -10) circle(.15);
\filldraw (9, -10) circle(.15);
\filldraw[fill = white] (10, -10) circle(.15);
\filldraw (11, -10) circle(.15);
\filldraw (12, -10) circle(.15);
\filldraw (8, -11.5) circle(.15);
\filldraw (8, -12.5) circle(.15);
\filldraw (9, -11.5) circle(.15);
\filldraw (9, -12.5) circle(.15);
\filldraw (8.5, -11.5) circle(.15);
\filldraw (8.5, -12.5) circle(.15);

\node[font=\fontsize{20}{6}\selectfont] at (13.5, -10.5) {$v_2$};
\node[font=\fontsize{20}{6}\selectfont] at (14.7, -11) {$v_1$};

\draw [thick, line width = 1pt] (13, -10) -- (14, -10);
\draw [thin, line width = 0.5pt] (14, -10) -- (15, -10);
\draw [thick, line width = 1pt] (15, -10) -- (16, -10);
\draw [densely dashed, line width = 1pt] (14, -10) -- (14,-11.5);
\draw [thick, line width = 1pt] (14, -11.5) -- (14, -12.5);
\draw [thick, line width = 1pt, color = green] (14, -10) -- (14.5,-11.5);
\draw [thick, line width = 1pt] (14.5, -11.5) -- (14.5, -12.5);
\draw [thick, line width = 1pt] (14.5, -11.5) -- (15, -11.5);
\draw [thick, line width = 1pt] (14.5, -12.5) -- (15, -12.5);
\filldraw (13, -10) circle(.15);
\filldraw (14, -10) circle(.15);
\filldraw (15, -10) circle(.15);
\filldraw (16, -10) circle(.15);
\filldraw (14,-11.5) circle(.15);
\filldraw (14,-12.5) circle(.15);
\filldraw (14.5,-11.5) circle(.15);
\filldraw (14.5,-12.5) circle(.15);
\filldraw (15,-11.5) circle(.15);
\filldraw (15,-12.5) circle(.15);
\end{tikzpicture}
\end{minipage}}
\end{center}
\captionsetup{width=1.0\linewidth}
\caption{An illustration of a representative possible case in Operation~\ref{op02},
	in which the green edge is $\{v_1, v_2\}$ and $v_2$ is in a satellite (respectively, center) element of $K_2$ before (respectively, after) the operation.\label{fig06}}
\end{figure}
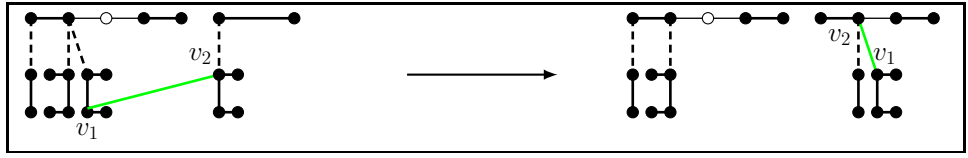

\begin{lemma}\label{lemma24} 
Operation~2 reduces the value of $g$ by at least~$1$.
\end{lemma}
\begin{proof}
Clearly, Operation~2 never changes $n_{cc}$, but could increase $n_0$ by at most~$4$ because $S_2$ is a bad component of $H$ and hence $|V(S_2) \cap V(M)|\le 4$.
For convenience, let $K'_i$ denote the modified $K_i$ after the operation for each $i\in\{1,2\}$. 
If the center element of $K_2$ is an edge or star, then by Statement~2 of Lemma~\ref{lemma09}, $S_2$ is bi-star and in turn $s(K_2) = 6$. 
On the other hand, if the center element of $K_2$ is a bi-star,
then by Statement~1 of Lemma~\ref{lemma09}, $S_2$ may be an edge, star, or bi-star, and hence $s(K_2) = 6$ or $8$.
So, in summary, $s(K_2) \in \{ 6, 8 \}$. 
Hence, by Statement~1 of Fact~\ref{fact02}, $K_2$ is not critical. 
Consequently, $K_1 \ne K_2$, and in turn $K'_1$ is not critical. 
Moreover, since $S_1$ is a bi-star, $s(K'_2) \in \{ 10, 12 \}$; hence, by Statement~1 of Fact~\ref{fact02} again, $K'_2$ is not critical. 
Therefore, the operation decreases $n_c$ by at least~$1$ and hence decreases $g$ by at least~$1$.
\end{proof}

\begin{operation}   
\label{op03}     
Suppose that the center element of $K_2$ is a $5$-path or $K_2$ has two or more satellite-elements, and $v_2$ is in a satellite-element $S_2$ of $K_2$. 
Then, the operation modifies $C$ by replacing the rescue-edges of $S_1$ and $S_2$ with the edge $\{v_1, v_2\}$,
and updates $S_2$ to be the center element of the new component $K_3$ of $H+C$ formed by $S_1$ and $S_2$ together (see for an illustration in Figure~\ref{fig07}).
\end{operation}

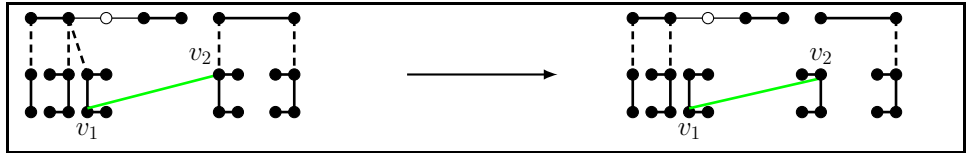
\begin{figure}[thb]
\begin{center}
\framebox{
\begin{minipage}{0.75\textwidth}
\begin{tikzpicture}[scale=0.5,transform shape]
\draw [thick, line width = 1pt] (-8, -10) -- (-7, -10);
\draw [thick, line width = 1pt] (-5, -10) -- (-4, -10);
\draw [thin, line width = 0.5pt] (-7, -10) -- (-6, -10);
\draw [thin, line width = 0.5pt] (-6, -10) -- (-5, -10);
\draw [densely dashed, line width = 1pt] (-8, -10) -- (-8,-11.5);
\draw [thick, line width = 1pt] (-8, -11.5) -- (-8, -12.5);
\draw [densely dashed, line width = 1pt] (-7, -10) -- (-7, -11.5);
\draw [thick, line width = 1pt] (-7, -11.5) -- (-7, -12.5);
\draw [thick, line width = 1pt] (-7, -11.5) -- (-7.5, -11.5);
\draw [thick, line width = 1pt] (-7, -12.5) -- (-7.5, -12.5);
\draw [densely dashed, line width = 1pt] (-7, -10) -- (-6.5,-11.5);
\draw [thick, line width = 1pt] (-6.5, -11.5) -- (-6.5, -12.5);
\draw [thick, line width = 1pt] (-6.5, -11.5) -- (-6, -11.5);
\draw [thick, line width = 1pt] (-6.5, -12.5) -- (-6, -12.5);
\filldraw (-8, -10) circle(.15);
\filldraw (-7, -10) circle(.15);
\filldraw[fill = white] (-6, -10) circle(.15);
\filldraw (-5, -10) circle(.15);
\filldraw (-4, -10) circle(.15);
\filldraw (-8,-11.5) circle(.15);
\filldraw (-8,-12.5) circle(.15);
\filldraw (-7,-11.5) circle(.15);
\filldraw (-7,-12.5) circle(.15);
\filldraw (-7.5,-11.5) circle(.15);
\filldraw (-7.5,-12.5) circle(.15);
\filldraw (-6.5,-11.5) circle(.15);
\filldraw (-6.5,-12.5) circle(.15);
\filldraw (-6,-11.5) circle(.15);
\filldraw (-6,-12.5) circle(.15);

\node[font=\fontsize{20}{6}\selectfont] at (-3.5, -11) {$v_2$};
\node[font=\fontsize{20}{6}\selectfont] at (-6.5, -13) {$v_1$};

\draw [thick, line width = 1pt, color=green] (-6.5, -12.4) -- (-3, -11.5);

\draw [thick, line width = 1pt] (-3, -10) -- (-1, -10);
\draw [densely dashed, line width = 1pt] (-1, -10) -- (-1,-11.5);
\draw [thick, line width = 1pt] (-1, -11.5) -- (-1, -12.5);
\draw [thick, line width = 1pt] (-1, -11.5) -- (-1.5, -11.5);
\draw [thick, line width = 1pt] (-1, -12.5) -- (-1.5, -12.5);
\draw [densely dashed, line width = 1pt] (-3, -10) -- (-3,-11.5);
\draw [thick, line width = 1pt] (-3, -11.5) -- (-3, -12.5);
\draw [thick, line width = 1pt] (-3, -11.5) -- (-2.5, -11.5);
\draw [thick, line width = 1pt] (-3, -12.5) -- (-2.5, -12.5);
\filldraw (-3, -10) circle(.15);
\filldraw (-1, -10) circle(.15);
\filldraw (-3,-11.5) circle(.15);
\filldraw (-3,-12.5) circle(.15);
\filldraw (-2.5,-11.5) circle(.15);
\filldraw (-2.5,-12.5) circle(.15);
\filldraw (-1.5,-11.5) circle(.15);
\filldraw (-1.5,-12.5) circle(.15);
\filldraw (-1,-11.5) circle(.15);
\filldraw (-1,-12.5) circle(.15);

\draw [-latex, thick] (2, -11.5) to (6, -11.5);

\draw [thick, line width = 1pt] (8, -10) -- (9, -10);
\draw [thick, line width = 1pt] (11, -10) -- (12,-10);
\draw [thin, line width = 0.5pt] (9, -10) -- (10, -10);
\draw [thin, line width = 0.5pt] (10, -10) -- (11, -10);
\draw [densely dashed, line width = 1pt] (8, -10) -- (8, -11.5);
\draw [thick, line width = 1pt] (8, -11.5) -- (8, -12.5);
\draw [densely dashed, line width = 1pt] (9, -10) -- (9, -11.5);
\draw [thick, line width = 1pt] (9, -11.5) -- (9, -12.5);
\draw [thick, line width = 1pt] (9, -11.5) -- (8.5, -11.5);
\draw [thick, line width = 1pt] (9, -12.5) -- (8.5, -12.5);
\draw [thick, line width = 1pt] (9.5, -11.5) -- (9.5, -12.5);
\draw [thick, line width = 1pt] (9.5, -11.5) -- (10, -11.5);
\draw [thick, line width = 1pt] (9.5, -12.5) -- (10, -12.5);
\filldraw (8, -10) circle(.15);
\filldraw (9, -10) circle(.15);
\filldraw[fill = white] (10, -10) circle(.15);
\filldraw (11, -10) circle(.15);
\filldraw (12, -10) circle(.15);
\filldraw (8, -11.5) circle(.15);
\filldraw (8, -12.5) circle(.15);
\filldraw (9, -11.5) circle(.15);
\filldraw (9, -12.5) circle(.15);
\filldraw (9.5, -11.5) circle(.15);
\filldraw (9.5, -12.5) circle(.15);
\filldraw (10, -11.5) circle(.15);
\filldraw (10, -12.5) circle(.15);
\filldraw (8.5, -11.5) circle(.15);
\filldraw (8.5, -12.5) circle(.15);

\node[font=\fontsize{20}{6}\selectfont] at (13, -11) {$v_2$};
\node[font=\fontsize{20}{6}\selectfont] at (9.5, -13) {$v_1$};

\draw [thick, line width = 1pt, color=green] (9.5, -12.4) -- (13,-11.6);

\draw [thick, line width = 1pt] (13, -10) -- (15, -10);
\draw [thick, line width = 1pt] (13, -11.5) -- (13, -12.5);
\draw [thick, line width = 1pt] (13, -11.5) -- (12.5, -11.5);
\draw [thick, line width = 1pt] (13, -12.5) -- (12.5, -12.5);
\draw [densely dashed, line width = 1pt] (15, -10) -- (15,-11.5);
\draw [thick, line width = 1pt] (15, -11.5) -- (15, -12.5);
\draw [thick, line width = 1pt] (15, -11.5) -- (14.5, -11.5);
\draw [thick, line width = 1pt] (15, -12.5) -- (14.5, -12.5);
\filldraw (13, -10) circle(.15);
\filldraw (15, -10) circle(.15);
\filldraw (12.5,-11.5) circle(.15);
\filldraw (12.5,-12.5) circle(.15);
\filldraw (13,-11.5) circle(.15);
\filldraw (13,-12.5) circle(.15);
\filldraw (14.5,-11.5) circle(.15);
\filldraw (14.5,-12.5) circle(.15);
\filldraw (15,-11.5) circle(.15);
\filldraw (15,-12.5) circle(.15);
\end{tikzpicture}
\end{minipage}}
\end{center}
\captionsetup{width=1.0\linewidth}
\caption{An illustration of a representative possible case in Operation~\ref{op03},
	in which the green edge is $\{v_1, v_2\}$ and $v_2$ is in a satellite (respectively, center) element before (respectively, after) the operation.\label{fig07}}
\end{figure}

\begin{lemma}\label{lemma25} 
Operation~3 reduces the value of $g$ by at least~$1$.
\end{lemma}
\begin{proof}
Clearly, Operation~3 increases $n_{cc}$ by~$1$. 
Moreover, since both $S_1$ and $S_2$ are bad components of $H$, $s(K_3) \le 8$ and hence $K_3$ is not critical by Statement~1 of Fact~\ref{fact02}. 
For convenience, let $K'_i$ denote the modified $K_i$ after the operation for each $i\in\{1,2\}$. 

We claim that $K'_1$ is not critical. 
If $K_1 \ne K_2$, the claim holds because of Statement~2 of Fact~\ref{fact02}.
Thus, we may assume that $K_1 = K_2$.
By Statement~1 of Fact~\ref{fact02}, $K_1$ has exactly one or two $2$-anchors.
If $K_1$ has two $2$-anchors, then the operation removes two satellite elements from $K_1$ including a critical satellite $S_1$.
By the first statement of Fact~\ref{fact02}, $s(K_1) \ge 14$ and therefore, $K_1$ has at least three satellite elements.
It follows that $K'_1$ is not isolated.

If $K_1$ has only one $2$-anchor, then $K'_1$ has no $2$-anchor and hence is not critical by Statement~1 of Fact~\ref{fact02}. 
So, we may assume that $K_1$ has exactly two $2$-anchors. 
Then, by Lemmas~\ref{lemma19} and~\ref{lemma20}, $s(K_1) \in \{ 30, 32 \}$. 
Since $S_1$ is a bi-star, the operation removes two satellite-elements from $K_1$ including $S_1$ and hence changes $s(K'_1) \in \{ 22, 24, 26 \}$. 
Therefore, Statement~1 of Fact~\ref{fact02} implies that $K'_1$ is not critical. 
In summary, the claim always holds. 

By the above claim, the operation does not increase $n_c$. 
Moreover, by Notation~\ref{nota05}, the operation increases $n_0$ by at most~$5$ because $|V(K_3) \cap V(M)| \le 4$ and $K'_2$ has at most one more $0$-anchor than $K_2$.
Therefore, the operation decreases $g$ by at least~$1$.
\end{proof}

\begin{lemma}
\label{lemma26}
Operations~\ref{op01}--\ref{op03} can be performed $O(n)$ times in total. 
Moreover, when none of them is applicable, the following hold:
\begin{enumerate}
\parskip=0pt
\item $C$ is still a maximum-weighted path-cycle cover of $G'$.
\item Suppose that $G$ has an edge $\{v_1, v_2\}$ such that $v_1$ is in a critical satellite-element $S_1$ of $H + C$ and $v_2 \notin V(S_1)$.
	Then, $v_2$ is a $2$-anchor or a responsible $1$-anchor.
\end{enumerate}
\end{lemma}
\begin{proof}
Since $-6n \le g \le 6n$ by its definition and each operation reduces the value of $g$ by at least~$1$ (cf. Lemmas~\ref{lemma23}--\ref{lemma25}), 
the three operations can be performed $O(n)$ times in total. 
Moreover, it is easy to verify that no operation creates a new isolated bad component in $H+C$,
and hence $C$ is still a maximum-weighted path-cycle cover of $G'$ after each operation.

We next prove Statement~2.
If $v_2$ were in a satellite-element of $H+C$, then clearly Operation~\ref{op02} or~\ref{op03} would have been applicable, a contradiction.
So, we may assume that $v_2$ is in a center element of $H+C$.
Since $S_1$ is a bi-star, Statement~1 of Lemma~\ref{lemma05} implies that $v_2$ is in a $5$-path or $v_2 \in V(M)$.
Thus, by Definition~\ref{def08} and the construction of trunks $\widetilde{K}$, $v_2$ is an anchor. 
Now, since Operation~\ref{op01} is not applicable, $v_2$ is a $2$-anchor or a responsible $1$-anchor.
\end{proof}

\subsection{The complete algorithm}
Let $R$ denote the set of vertices $v \in V(H)$ such that $v$ is a $2$-anchor or a responsible $1$-anchor in $H+C$.
By Lemma~\ref{lemma26}, once none of Operations~\ref{op01}--\ref{op03} is applicable, critical satellite-elements of $H+C$ can be incident to only those vertices of $R$ in $G$. 
For each connected component $K$ of $H+C$, $|R \cap V(K)|$ is bounded by the total number of anchors in $K_c$.
So, $|R \cap V(K)| \le 5$ by Notation~\ref{nota05}. 
In particular, if $K$ is critical, then by Statements~1 and~3 of Fact~\ref{fact02}, 
$K$ contains exactly one or two $2$-anchors but no responsible $1$-anchor, i.e., $|R \cap V(K)| \in \{1, 2\}$.

\begin{itemize}
\parskip=0pt
\item Let $\mcK$ be the set of composite components or isolated $5$-paths of $H+C$.
\item For each $i\in\{0,1,2,3,4,5\}$, let $\mcK_i$ be the set of all $K\in \mcK$ with $|R \cap V(K)| = i$. 
\item For each $i\in\{1,2\}$, let $\mcK_{i,c}$ be the set of critical components in $\mcK_i$. 
\item Let $R_c \subseteq R$ be the set of critical anchors. 
\item $U_c = \bigcup_{v\in R_c} \{w\in V(H) \mid w$ is in a critical satellite-element whose rescue-anchor is $v\}$. 
\item Let $G_c = G[V(G) \setminus (R_c\cup U_c)]$. 
\end{itemize}

We next present a lemma which suggests that we may reduce the problem of finding a good feasible solution for $G$ to finding a good feasible solution for the smaller instance $G_c$.

\begin{lemma}
\label{lemma27} 
$opt(G) \le opt(G_c) + 9\sum_{i=1}^5 i|\mcK_i|$.
\end{lemma}
\begin{proof}
First of all, we can assume that each path component of $OPT(G)$ has at most $9$ vertices, 
because otherwise, we can break it into two or more $5^+$-paths without changing the number of vertices in the solution.

Consider a critical satellite-element $S$. 
Recall that $S$ must be a bi-star and thus cannot form a $5^+$-path by itself alone.
By Lemma~\ref{lemma26}, each vertex of $S$ can be adjacent to only those vertices of $R$ in $G$. 
So, removing all the vertices of $R_c \cup U_c$ from $G$ destroys at most $|R|$ paths in $OPT(G)$, 
and each path component of $OPT(G)$ containing no vertex of $R$ remains to be a $5^+$-path. 
Thus, $opt(G_c) \ge opt(G) - 9|R|$.
Now, since $|R| = \sum_{i=1}^5 i|\mcK_i|$, the lemma holds.
\end{proof}

Let $r = \frac {26 + \sqrt{3826}}{35} < 2.511$ be the unique positive root to the quadratic equation $35x^2 - 52x - 90 = 0$.
Our complete algorithm proceeds as follows.

\begin{enumerate}
\parskip=0pt
\item[0.] If $|V(G)| \le 5$, then find an optimal solution by brute-force search, output it, and halt.
\item Compute a maximum matching $M$ in $G$ and initialize a subgraph $H=(V(M),M)$ of $G$.
\item Modify $H$ and $M$ by performing Steps~1.1--1.3 presented in Section~\ref{subsec:k=5}.
\item Compute a set $C$ of edges between the connected components of $M$ in $G$ by performing Steps~2.1--2.3 presented in Section~\ref{subsec:rescue}.
\item Modify $C$ by repeatedly performing Operations~\ref{op01}--\ref{op03} presented in 	Section~\ref{subsec:op} until none of them is applicable.
	({\em Comments:} After this step, some edges of $C$ may become redundant and we may remove them as in Step~2.3.)
\item If no connected component of $H+C$ is critical, or $\frac {\sum_{i=1}^5 i|\mcK_i|}{|\mcK_{1,c}|+2|\mcK_{2,c}|} > \frac {7}9r$, then perform the following two steps:
\begin{itemize}
\parskip=0pt
\item[5.1.] For each connected component $K$ of $H+C$ that is not an isolated bad component of $H$, compute $OPT(\widetilde{K})$ as in Lemma~\ref{lemma10}.
\item[5.2.] Output the union of the sets of $5^+$-paths obtained in Step~5.1, and halt. 
\end{itemize}
\item If there is at least one critical component and $\frac {\sum_{i=1}^5 i|\mcK_i|}{|\mcK_{1,c}|+2|\mcK_{2,c}|} \le \frac {7}9r$, then perform the following three steps:
\begin{itemize}
\parskip=0pt
\item[6.1.] Recursively call the algorithm on the graph $G_c$ to obtain a set $ALG(G_c)$ of vertex-disjoint $5^+$-paths in $G_c$. 
\item[6.2.] For each $v \in R_c$, compute $P_v$ which is a $7^+$-path because $v \in T_2$.
\item[6.3.] Output the union of $ALG(G_c)$ and $\cup_{v\in R_c} P_v$, and halt.
\end{itemize}
\end{enumerate}

\subsection{The remaining analysis for the algorithm}\label{sec:ana4}
The next fact directly follows from Lemmas~\ref{lemma16}--\ref{lemma22}.

\begin{fact}
\label{fact03}
For each connected component $K$ of $H+C$, the following statements hold:
\begin{enumerate}
\parskip=0pt
\item If $K \in \mcK_0$, then $\frac {s(K)}{\eta(K)} < \frac {15}8$.
\item $\frac {s(K)}{\eta(K)} \preceq \{\frac {16}7, \frac {18}8\}$ if $K \in \mcK_{1, c}$;
	while $\frac {s(K)}{\eta(K)} \preceq \frac {2h+8}{h+5}$ for some $h = 0, 1, \ldots, 10$ if $K \in \mcK_1 \setminus \mcK_{1, c}$. 
\item $\frac {s(K)}{\eta(K)} \preceq \max\{\frac {30}{16}, \frac {32}{17}\}$ if $K \in \mcK_{2, c}$;
	while $\frac {s(K)}{\eta(K)} \preceq \frac {2h+16}{h+10}$ for some $h = 0, 1, \ldots, 8$ if $K \in \mcK_2 \setminus \mcK_{2, c}$. 
\item If $K \in \mcK_3$, then $\frac {s(K)}{\eta(K)} \preceq \frac {2h+24}{h+15}$ for some $h = 0, 1, \ldots, 6$. 
\item If $K \in \mcK_4$, then $\frac {s(K)}{\eta(K)} \preceq \{\frac {32}{20}, \frac {38}{24}, \frac {40}{28}\}$.
\item If $K \in \mcK_5$, then $\frac {s(K)}{\eta(K)} \preceq \{\frac {40}{25}, \frac {44}{33}\}$.
\end{enumerate}
\end{fact}

\begin{lemma}
\label{lemma28}
If there is no critical component, then $opt(G) < \frac {75}{32}\cdot alg(G)$.
\end{lemma}
\begin{proof}
By Definition~\ref{def07}, we have $\frac {s(K)}{\eta(K)} < \frac {15}8$ for each connected component $K$ of $H+C$.
So, $alg(G) = \sum_K \eta(K) \ge \frac 8{15} \sum_K s(K) =  \frac 8{15} |V(M_C)|$.
By Lemma \ref{lemma07}, the lemma is proven.
\end{proof}

\begin{lemma}
\label{lemma29}
If there exists a critical component of $H+C$ and $\frac {\sum_{i=1}^5 i |\mcK_i|}{|\mcK_{1, c}| + 2|\mcK_{2, c}|} > \frac {7}9r$, then $opt(G) \le r \cdot alg(G)$.
\end{lemma}
\begin{proof}
For each $i \in \{0, 3, 4, 5\}$ and each $K \in \mcK_i$, we define $s'(K) = s(K) + i(4r - 8)$. 
Similarly, for each $i \in \{1, 2\}$ and each $K \in \mcK_i \setminus \mcK_{i, c}$, we define $s'(K) = s(K) + i(4r - 8)$.
Moreover, for each $i \in \{1, 2\}$ and each $K \in \mcK_{i, c}$, we define $s'(K) = s(K) - i(16 - \frac {28}5 r)$.

We claim that $s'(K) \le \frac {4r}5 \cdot \eta(K)$ for every $K \in \mcK$.
If $K \in \mcK_0$, the claim holds because $s'(K) = s(K) \le \frac {15}8 \cdot\eta(K) < \frac {4r}5 \cdot \eta(K)$.
If $K \in \mcK_1 \setminus \mcK_{1, c}$, the claim holds because Statement~2 of Fact~\ref{fact03} 
implies that there is some $h \in \{ 0, 1, \ldots, 10\}$ such that 
\[
\frac {s'(K)}{\eta(K)} \preceq \frac {2h+8 + (4r-8)}{h+5} = \frac {2h+4r}{h+5} \le \frac {4r}5.
\]
If $K \in \mcK_2 \setminus \mcK_{2, c}$, the claim holds because Statement~3 of Fact~\ref{fact03} 
implies that there is some $h \in \{ 0, 1, \ldots, 8\}$ such that
\[
\frac {s'(K)}{\eta(K)} \preceq \frac {2h+16 + 2(4r-8)}{h+10} = \frac {2h+8r}{h+10} \le \frac {4r}5.
\]
If $K \in \mcK_3$, the claim holds because Statement~4 of Fact~\ref{fact03} implies that there is some $h \in \{ 0, 1, \ldots, 6\}$ such that 
\[
\frac {s'(K)}{\eta(K)} \preceq \frac {2h+24 + 3(4r-8)}{h+15} = \frac {2h+12r}{h+15} \le \frac {4r}5.
\]
If $K \in \mcK_4 \cup \mcK_5$, the claim holds because Statements~5 and~6 in Fact~\ref{fact03} imply that 
\[
\frac {s'(K)}{\eta(K)} \preceq \left\{ \frac {16r}{20}, \frac {16r+6}{24}, \frac {16r+8}{28}, \frac {20r}{25}, \frac {20r+4}{33} \right\} \preceq \frac {4r}5.
\]
If $K \in \mcK_{1,c} \cup \mcK_{2, c}$, the claim holds because Statements~2 and~3 in Fact~\ref{fact03} imply that 
\[
\frac {s'(K)}{\eta(K)} \preceq \left\{ \frac {4r}5, \frac {14r+5}{20}, \frac {28r-5}{40}, \frac {56r}{85} \right\} \preceq \frac {4r}5
\]
Thus, the claim always holds.

By Lemma \ref{lemma07}, to complete the proof, it suffices to show that $|V(M_C)| \le \frac {4r}5\cdot alg(G)$. 
To this end, first note that $\frac {(7r-9)(4r - 8)}9 = 16 - \frac {28}5 r$, because $r$ is the root to the equation $35x^2 - 52x - 90 = 0$.
Since $\frac {\sum_{i=1}^5 i |\mcK_i|}{|\mcK_{1, c}| + 2|\mcK_{2, c}|} > \frac {7r}9$, we have 
\begin{eqnarray}
\label{eq01}
& & \sum_{i \in \{0, 3, 4, 5\}} \sum_{K \in \mcK_i} i(4r-8) + \sum_{i \in \{1, 2\}} \sum_{K \in \mcK_i \setminus \mcK_{i, c}} i(4r-8) \nonumber \\
& = & (4r-8) \left (\sum_{i=1}^5 i|\mcK_i| -  |\mcK_{1,c}| - 2 |\mcK_{2, c}|\right) \nonumber \\
& \ge & \frac {(7r-9)(4r-8)}9 \left (|\mcK_{1,c}| + 2 |\mcK_{2, c}|\right) \nonumber \\
& = & (16 - \frac {28}5 r) \left (|\mcK_{1,c}| + 2 |\mcK_{2, c}|\right).
\end{eqnarray}
Since $|V(M_C)| = \sum_{K \in \mcK} s(K)$, we now have
\begin{eqnarray*}
|V(M_C)| & = & \sum_{i \in \{0, 3, 4, 5\}} \sum_{K \in \mcK_i} s(K) + \sum_{i \in \{1, 2\}} \sum_{K \in \mcK_i \setminus \mcK_{i, c}} s(K)
	 + \sum_{i \in \{1, 2\}} \sum_{K \in \mcK_{i, c}} s(K) \\
& = & \sum_{i \in \{0, 3, 4, 5\}} \sum_{K \in \mcK_i} s'(K) + \sum_{i \in \{1, 2\}} \sum_{K \in \mcK_i \setminus \mcK_{i, c}} s'(K)
	 + \sum_{i \in \{1, 2\}} \sum_{K \in \mcK_{i, c}} s(K) \\
&   & - \left( \sum_{i \in \{0, 3, 4, 5\}} \sum_{K \in \mcK_i} i(4r-8) + \sum_{i \in \{1, 2\}} \sum_{K \in \mcK_i \setminus \mcK_{i, c}} i(4r-8) \right) \\
& \le & \sum_{i \in \{0, 3, 4, 5\}} \sum_{K \in \mcK_i} s'(K) + \sum_{i \in \{1, 2\}} \sum_{K \in \mcK_i \setminus \mcK_{i, c}} s'(K)
	 + \sum_{i \in \{1, 2\}} \sum_{K \in \mcK_{i, c}} \left (s(K) - i(16-\frac {28}5 r) \right) \\
& = & \sum_{K \in \mcK} s'(K) \le \frac {4r}5 \sum_{K \in \mcK} \eta(K) 
= \frac {4r}5\cdot alg(G),
\end{eqnarray*}
where the first inequality follows from Eq.(\ref{eq01}), the second inequality follows from the above claim, 
and the last equality holds because of Step~5 of our algorithm.
\end{proof}

\begin{theorem}
\label{thm01}
The algorithm is an $O(n^{2.5} m^2)$-time $r$-approximation algorithm for $MPC^{5+}_v$, where $n = |V(G)|$, $m = |E(G)|$ and $r = \frac {26 + \sqrt{3826}}{35} < 2.511$.
\end{theorem}
\begin{proof} 
For each $i\in\{1,2,3\}$, performing Operation~$i$ takes $O(1)$ time and it takes $O(n+m)$ time to check whether Operation~$i$ is applicable. 
So, by Lemmas~\ref{lemma02} and~\ref{lemma06}, it is not difficult to see that other than the recursive call, each step of the algorithm takes $O(n^{1.5} m^2)$ time.
Since the recursion depth is $O(n)$, the total running time is $O(n^{2.5} m^2)$.

We claim that the approximation ratio achieved by the algorithm is $r$. 
We prove the claim by induction on $n$. 
In the base case, $n \le 5$ and the algorithm outputs an optimal solution of $G$, implying that the claim holds.
So, suppose that $n \ge 6$.
By Lemmas~\ref{lemma28} and~\ref{lemma29}, we may further assume that there exists at least one critical component in $H+C$ and 
$\frac {\sum_{i=1}^5 i |\mcK_i|}{|\mcK_{1, c}| + 2|\mcK_{2, c}|} \le \frac {7r}9$.
Then, $alg(G) \ge 7 \left( |\mcK_{1, c}| + 2|\mcK_{2, c}| \right) + alg(G_c)$ by Step~6 of the algorithm.
Moreover, by the inductive hypothesis, $opt(G_c) \le r \cdot alg(G_c)$. 
Now, by Lemma~\ref{lemma27}, we finally have 
\[
\frac {opt(G)}{alg(G)} \le \frac {9\sum_{i=1}^5 i |\mcK_i| + opt(G_c)}{7\left( |\mcK_{1, c}| + 2|\mcK_{2, c}| \right) + alg(G_c)} 
\le \max \left\{ \frac {9 \sum_{i=1}^5 i |\mcK_i|}{7 \left(|\mcK_{1, c}| + 2|\mcK_{2, c}| \right)} , \frac {opt(G_c)}{alg(G_c)}\right\} \le r,
\]
which completes the proof.
\end{proof}

\section{Conclusion}
We studied the problem $MPC^{k+}_v$ to cover as many vertices as possible by a collection of vertex-disjoint $k^+$-paths in the graph,
and the current paper focuses on $k = 5$.
The main contribution is to extend the design and analysis techniques developed in~\cite{GCL23} to an improved $2.511$-approximation algorithm,
from previously best ratio of $2.714$~\cite{GFL22,GEF24},
among which the key idea is to keep most endpoints of a maximum matching in the solution since they contribute up to $\frac 45$ of the optimum.
The processes for achieving this goal are all technically non-trivial.

The presented algorithm design and analysis techniques work well for $k = 4$ and $5$,
and we believe they also work for some other small values of $k$.
Extending the techniques for the general $k$ could be an interesting pursuit.
It would also be theoretically interesting to prove some inapproximability results for $MPC^{k+}_v$.


\begin{thebibliography}{10}

\bibitem{AN07}
K.~Asdre and S.~D. Nikolopoulos.
\newblock A linear-time algorithm for the $k$-fixed-endpoint path cover problem
  on cographs.
\newblock {\em Networks}, 50:231--240, 2007.

\bibitem{AN10}
K.~Asdre and S.~D. Nikolopoulos.
\newblock A polynomial solution to the $k$-fixed-endpoint path cover problem on
  proper interval graphs.
\newblock {\em Theoretical Computer Science}, 411:967--975, 2010.

\bibitem{BK06}
P.~Berman and M.~Karpinski.
\newblock 8/7-approximation algorithm for (1,2)-{TSP}.
\newblock In {\em Proceedings of SODA 2006}, pages 641--648, 2006.

\bibitem{CCC18}
Y.~Cai, G.~Chen, Y.~Chen, R.~Goebel, G.~Lin, L.~Liu, and An~Zhang.
\newblock Approximation algorithms for two-machine flow-shop scheduling with a
  conflict graph.
\newblock In {\em Proceedings of COCOON 2018}, LNCS 10976, pages 205--217,
  2018.

\bibitem{CCL22}
Y.~Chen, Y.~Cai, L.~Liu, G.~Chen, R.~Goebel, G.~Lin, B.~Su, and A.~Zhang.
\newblock Path cover with minimum nontrivial paths and its application in
  two-machine flow-shop scheduling with a conflict graph.
\newblock {\em Journal of Combinatorial Optimization}, 43:571--588, 2022.

\bibitem{CGL19a}
Y.~Chen, R.~Goebel, G.~Lin, B.~Su, Y.~Xu, and A.~Zhang.
\newblock An improved approximation algorithm for the minimum $3$-path
  partition problem.
\newblock {\em Journal of Combinatorial Optimization}, 38:150--164, 2019.

\bibitem{Gab83}
H.~N. Gabow.
\newblock An efficient reduction technique for degree-constrained subgraph and
  bidirected network flow problems.
\newblock In {\em Proceedings of STOC'83}, pages 448--456, 1983.

\bibitem{GW20}
R.~G{\'{o}}mez and Y.~Wakabayashi.
\newblock Nontrivial path covers of graphs: Existence, minimization and
  maximization.
\newblock {\em Journal of Combinatorial Optimization}, 39:437--456, 2020.

\bibitem{GCL23}
M.~Gong, Z.-Z. Chen, G.~Lin, and L.~Wang.
\newblock An approximation algorithm for covering vertices by $4^+$-paths.
\newblock In {\em Proceedings of COCOA 2023}, LNCS 14461, pages 459--470.

\bibitem{GEF24}
M.~Gong, B.~Edgar, J.~Fan, G.~Lin, and E.~Miyano.
\newblock Approximation algorithms for covering vertices by long paths.
\newblock {\em Algorithmica}, 86:2625--2651, 2024.

\bibitem{GFL22}
M.~Gong, J.~Fan, G.~Lin, and E.~Miyano.
\newblock Approximation algorithms for covering vertices by long paths.
\newblock In {\em Proceedings of MFCS 2022}, LIPIcs 241, pages 53:1--53:14,
  2022.

\bibitem{HHS06}
D.~Hartvigsen, P.~Hell, and J.~Szab\'{o}.
\newblock The $k$-piece packing problem.
\newblock {\em Journal of Graph Theory}, 52:267--293, 2006.

\bibitem{KLM22}
K.~Kobayashi, G.~Lin, E.~Miyano, T.~Saitoh, A.~Suzuki, T.~Utashima, and
  T.~Yagita.
\newblock Path cover problems with length cost.
\newblock In {\em Proceedings of WALCOM 2022}, LNCS 13174, pages 396--408,
  2022.

\bibitem{KLM23}
K.~Kobayashi, G.~Lin, E.~Miyano, T.~Saitoh, A.~Suzuki, T.~Utashima, and
  T.~Yagita.
\newblock Path cover problems with length cost.
\newblock {\em Algorithmica}, 85:3348--3375, 2023.

\bibitem{MV80}
S.~Micali and V.~V. Vazirani.
\newblock An {$O(\sqrt{|V|} |E|)$} algorithm for finding maximum matching in
  general graphs.
\newblock In {\em Proceedings of FOCS 1980}, pages 17--27, 1980.

\bibitem{PH08}
L.~L. Pao and C.~H. Hong.
\newblock The two-equal-disjoint path cover problem of matching composition
  network.
\newblock {\em Information Processing Letters}, 107:18--23, 2008.

\bibitem{RTM14}
R.~Rizzi, A.~I. Tomescu, and V.~M{\"a}kinen.
\newblock On the complexity of minimum path cover with subpath constraints for
  multi-assembly.
\newblock {\em BMC Bioinformatics}, 15:S5, 2014.

\end{thebibliography}

\end{document}